\documentclass[a4paper,11pt]{article}
\usepackage{amsmath,a4wide}
\usepackage{amssymb}
\usepackage{amscd}
\usepackage{bm}
\usepackage{bbm}
\usepackage{epsfig}
\usepackage{enumitem}
\usepackage{graphics}
\usepackage{graphicx}
\usepackage{float}
\usepackage{epsfig}
\usepackage{amsmath}
\usepackage{graphics}
\usepackage{graphicx}
\usepackage{color}
\usepackage{epstopdf}
\usepackage{color}
\usepackage{booktabs}
\usepackage{subfigure}
\usepackage{ifpdf}
\usepackage{amssymb}
\usepackage{cite}
\usepackage{amssymb}
\usepackage{amsthm}
\usepackage{amsmath}
\usepackage[justification=centering]{caption}
\usepackage[]{natbib}
\usepackage[titletoc]{appendix}
\usepackage{verbatim}
\newtheorem{definition}{Definition}
\newtheorem{theorem}{Theorem}
\newtheorem{lemma}{Lemma}
\newtheorem{pro}{Proposition}
\newtheorem{cor}{Corollary}
\newtheorem{remark}{Remark}
\newtheorem{assumption}{Assumption}
\begin{document}
	
	\title{Long-range dependent mortality modeling with cointegration}
	
	\author{Mei Choi Chiu\thanks{Department of Mathematics \& Information Technology, The Education University of Hong Kong, Tai Po, N.T., Hong Kong. \newline({\tt mcchiu@eduhk.hk})}	
         \and Ling Wang\thanks{Corresponding author. China Institute for Actuarial Science, Central University of Finance and Economics, Beijing 102206, China.\newline({\tt lingwang@cufe.edu.cn})}
		\and Hoi Ying Wong\thanks{Department of Statistics, The Chinese University of Hong Kong, Shatin, N.T., Hong Kong. \newline({\tt hywong@cuhk.edu.hk})}
	}
	\date{\today}
	
	\maketitle \pagestyle{plain} \pagenumbering{arabic}
	\begin{abstract}
    Empirical studies with publicly available life tables identify long-range dependence (LRD) in national mortality data.  Although the longevity market is supposed to benchmark against the national force of mortality, insurers are more concerned about the forces of mortality associated with their own portfolios than the national ones. Recent advances on mortality modeling make use of fractional Brownian motion (fBm) to capture LRD. A theoretically flexible approach even considers mixed fBm (mfBm). Using Volterra processes, we prove that the direct use of mfBm encounters the identification problem so that insurers hardly detect the LRD effect from their portfolios. Cointegration techniques can effectively bring the LRD information within the national force of mortality to the mortality models for insurers' experienced portfolios.
    Under the open-loop equilibrium control framework, the explicit and unique equilibrium longevity hedging strategy is derived for cointegrated   {forces of mortality} with LRD. Using the derived hedging strategy, our numerical examples show that  the accuracy of estimating cointegration is crucial for hedging against the longevity exposure of insurers with LRD national   {force of mortality}.
	\end{abstract}
	\vspace{2mm} {\em Keywords}: Cointegration; Mortality model; Long-range dependence; Mean-variance hedging; Time-consistency. \\

	\section{Introduction}
	Increasing life expectancy makes longevity a key risk factor for insurers and pension funds. The introduction of longevity bonds \citep{BCD} and other longevity securitizations \citep{Dowd2006,Dawson2010,BB2010} is an attempt to create a life market for insurers to hedge against longevity risk. Many studies on longevity risk management are inspired by the assumption of a liquid longevity market \citep{Cairns2014, Biagini2017, Li2017}.
	
	Accurate mortality modeling is the building block for measuring and managing longevity risk. Although continuous-time Markovian mortality models \citep{Biffis} are widely adopted in the literature \citep{WoCW}, recent empirical analysis suggests that long-range dependence (LRD) exists in the mortality data across genders, ages \citep{Yaya2019}, and countries \citep{Yan2021, Yan2020}. \cite{Yan2021} reveal that ignoring the LRD property leads to   {underestimation of life expectancy  and therefore influences the measurement of longevity risk.}  \cite{WaCW} propose the Volterra mortality model (VMM) as an analytically tractable tool to capture and analyze the LRD property in national   {forces of mortality}.  When the fractional kernel is adopted in the VMM, the model has the Riemann-Liousville (RL) fractional Brownian motion (fBm) as the noise term. Alternatively, \cite{Zhou2022} propose a mortality model with a mixed fBm (mfBm) to capture LRD, where the noise term is a linear combination of a fBm and the standard Brownian motion.
 
	
    In reality, the insurer's experienced   {force of mortality}  is different from the national   {force of mortality}. However, there should be some dependence between the two   {forces of mortality}. Specifically, longevity securities that serve as hedging instruments, such as longevity swaps, are linked to a specific mortality index, which is different from the   {force of mortality}  experienced by the insurers and pension funds. The difference between the   {forces of mortality}  is recognized as basis risk in the literature. \cite{LH2011} and \cite{Coughlan2011} empirically study the influence of basis risk on hedging effectiveness. A simple means of overcoming basis risk is to set a fixed ratio between the two   {forces of mortality}. However, this relies on the ratio between the two rates being stable over time.
    
     The notion of cointegration, including correlation, is widely adopted to describe the joint movement of   {forces of mortality}. The concept of cointegration, as originated by  \cite{EG1987}, asserts that a linear combination of two or more non-stationary economic series sharing an equilibrium relation could either be stationary or have a lower degree of integration than the original series. Cointegration techniques have many applications in economics  \citep{MK2010}, and led to a Nobel Memorial Prize in Economic Sciences for Clive Granger in 2003. \cite{Alexander1999} points out that the concept of cointegration is essential for hedging. Furthermore, empirical evidence for the existence of cointegration in mortality modeling and multi-country longevity risks is found by \cite{NS2011} and \cite{YW2013}.

    Although the effects of cointegration between   {forces of mortality} on longevity hedging strategies has been found to be essential in the literature  \citep{WoCW,WoCW2017}, the longevity hedging strategies have not been investigated when the cointegration  is considered under the LRD mortality environment. Given the fact that the LRD feature exists in national mortality data and recent advances in LRD mortality modeling, this paper investigates the longevity hedging strategy by considering LRD and cointegration simultaneously. 

   The first contribution of this paper is proposing a tractable mortality model which is able to reflect the cointegration between   {forces of mortality} and LRD property simultaneously. A two-dimensional affine Volterra process is adopted for the modeling. The general dynamic of the mortality model embraces the Markovian mortality model with cointegration in \cite{WoCW,WoCW2017} as a special case. As the insurer's client pool should be a sample of the national population, the insurer's experienced   {force of mortality}  tends to move around the national   {force of mortality}, which is a situation of cointegration. {In addition, our LRD mortality model with cointegration can also be viewed as a combination of \cite{WaCW} and \cite{Zhou2022},   {as} the national   {force of mortality}  $\hat{\mu}_1$ in our model in \eqref{mu1} follows the dynamics suggested by the former while the   {force of mortality}  of the insurer's experienced portfolio $\hat{\mu}_2$ in \eqref{mu2} is consistent with the latter   {when one of the Hurst parameters in our model is 1/2.}
   
   {As the LRD models suggested by \cite{WaCW} and \cite{Zhou2022} have their own advantages, it is natural to ask which one is preferred in different scenarios. The second contribution is the development of a new theory (Proposition \ref{pro:Hurst}) on LRD identification problem with mfBm. As Brownian motion has the Hurst regularity of $H=1/2$ and fBm for LRD has $H>1/2$, we prove that a mixture of them is identified to have $H=1/2$. In other words, the LRD component can hardly be detected for the   {force of mortality}  driven by mfBm. 
   Our proposed model becomes a remedy for the identification problem because the cointegration framework brings empirically identified LRD information in the national   {force of mortality}  \citep{Yan2020, Yan2021, Yaya2019} to the mortality of an insurer's experienced portfolio. Without cointegration, it is highly uncertain if an insurer is able to precisely measure the LRD feature solely based on their in-house data set.}
   
   {Using longevity hedging as a pertinent example for the application of our model, our third contribution is the derivation of an analytical hedging strategy under the LRD mortality model with cointegration, where the $\mu_2$ is driven by mfBm. To the best of our knowledge, we are the first one proposing a tractable mortality model with mfBm in the literature. This analytical result enables us to conduct numerical analysis which shows advantages brought by cointegration for longevity hedging with LRD.}
   
   
  



   We also overcome some mathematical difficulties in the derivation of the longevity hedging strategy. When we adopt mean-variance hedging for its wide application in risk management, the variance operator within the objective function induces a time-inconsistency issue \citep{S}. The time-inconsistency issue means that the optimal hedging rule may depend on the initial value of the state process and become suboptimal at a later time point. To enforce time-consistency, \cite{BM} and \cite{BKM2017} introduce an equilibrium feedback strategy by formulating a game-theoretic framework.  This approach leads to an extended Hamilton--Jacobi--Bellman (HJB) equation {and a feedback control. However, the extended HJB framework and the feedback control mainly work for a Markovian system. Clearly, our problem encounters long-range dependence which is non-Markovian in nature so that the feedback control and closed-loop control frameworks cannot be presumed in advance}. In addition, the uniqueness of the equilibrium hedging strategy is difficult to study within the extended HJB framework \citep{BKM2017}.

In this paper, we adopt an open-loop equilibrium control framework \citep{HJZ2012, HJZ2017} based on the theory of backward stochastic differential equations (BSDE). The key advantage of the open-loop equilibrium approach is that it allows for non-Markovian systems. Furthermore, the open-loop framework allows the studies of the uniqueness of equilibrium controls \citep{HJZ2017, WCW2022}. 
This paper generalizes the open-loop control framework to the first time-consistent longevity hedging problem with simultaneous cointegration and LRD. Although a two-dimensional Volterra process is encountered in our problem, the explicit solutions of equilibrium strategies are provided under both deterministic and volatility-driven market prices of risks. Furthermore, we prove that the derived hedging strategies are unique under certain technical conditions. {We stress that our consideration is very different from \cite{WW2021}, who study the hedging of risk induced by $\mu_1$ using longevity bonds contingent on the same $\mu_1$. Whereas, we consider here that an insurer wants to hedge against her experienced mortality risk induced by $\mu_2$ using longevity bonds contingent on the national   {force of mortality}  $\mu_1$. This subtle difference requires a rather different and non-trivial mathematical treatment.} In the numerical studies, this paper further examine the effect of cointegration and correlation between   {forces of mortality} and LRD property.  {Our numerical result shows that the additional benefit brought by cointegration is significant for longevity hedging. }
    
The remainder of the paper is organized as follows. Section \ref{Sec:Problem} introduces the mortality model, which can reflect both the LRD property and the cointegration feature. { To motivate the model development, we also provide a new theory to show the LRD identification problem with mfBm in this section.} Section \ref{Sec:explicit} {defines the open-loop equilibrium hedging problem and} provides the explicit longevity hedging strategies under both deterministic and volatility-driven market prices of risk by solving the relevant BSDEs. In addition, we prove the uniqueness of the equilibrium hedging strategy. Section \ref{Sec:numerical} is devoted to numerical studies, with the aim of examining the impact of ignoring cointegration when the national   {force of mortality}  exhibits LRD. Section \ref{Sec:conclude} concludes the paper.

	\section{Problem formulation}
	\label{Sec:Problem}
	In a filtered complete probability space $(\Omega, \mathcal{F},\{\mathcal{F}\}_{t\in[0,T]},\mathbb{P})$, for any positive constant $q$ and positive integer $d$, we define the following:
	\begin{itemize}
	    \item[] $L^q_\mathcal{F}(\Omega; \mathbb{R}^d)$: the set of random variables $X: (\Omega, \mathcal{F})\rightarrow (\mathbb{R}^d, \mathcal{B}(\mathbb{R}^d))$ with $\mathbb{E}[|X|^q] < \infty$; 
		\item[]$S_\mathcal{F}^q(t,T;\mathbb{R}^d, \mathbb{P})$: the set of all $\{\mathcal{F}_s\}_{s\in[t,T]}$-adapted $\mathbb{R}^d$-valued stochastic processes $X(\cdot)=\{X(s): t\leq s\leq T\}$ with $\mathbb{E}[{\rm sup}_{t\leq s\leq T}|X(s)|^q]< \infty$;
		\item[]$L_\mathcal{F}^\infty(t,T;\mathbb{R}^d, \mathbb{P})$: the set of all essentially bounded
		$\{\mathcal{F}_s\}_{s\in[t,T]}$-adapted $\mathbb{R}^d$-valued stochastic processes;
		\item[] $L_\mathcal{F}^q(t,T;\mathbb{R}^d, \mathbb{P})$: the set of all
		$\{\mathcal{F}_s\}_{s\in[t,T]}$-adapted $\mathbb{R}^d$-valued stochastic processes with $\mathbb{E}\left[\int_{t}^{T}|X(s)|^qds\right]< \infty$;
		\item[] $H_\mathcal{F}^q(t,T;\mathbb{R}^d, \mathbb{P})$: the set of all $\{\mathcal{F}_s\}_{s\in[t,T]}$-adapted $\mathbb{R}^d$-valued stochastic processes $X(\cdot)=\{X(s): t\leq s\leq T\}$ with $\mathbb{E}\left[\left(\int_{t}^{T}|X(s)|^2ds\right)^{q/2}\right]<\infty$.
	\end{itemize}
	Throughout the paper, we regard the elements of $\mathbb{R}^d$ as column vectors, whereas the elements of the spaces $(\mathbb{R}^d)^*$ are regarded as row vectors.
	\subsection{The model}
	Let $(\Omega, \mathcal{F},\{\mathcal{F}_t\}_{t\in[0,T]},\mathbb{P})$ be a filtered probability space. Following the classical doubly stochastic mortality modeling literature, we write $\mathcal{F}_t=\mathcal{G}_t\vee\mathcal{H}_t$, where $\mathcal{H}_t$ represents the flow of information available over time, including the historical processes and the current states, and $\mathcal{G}_t$ contains information regarding whether an individual has died.
	
 Let $\mu(t) = (\mu_1(t), \mu_2(t))^\top \in {\mathcal H}_t$ be a pair of   {forces of mortality}, where we interpret $\mu_1$ as the national   {force of mortality}  and $\mu_2$ as the insurer's experienced   {force of mortality}. Therefore, $N_1(t+\Delta t) - N_1(t)$ is the number of death from the national population over the period of $[t, t+\Delta t)$ so that $N_1(t) \in {\mathcal G}_t$ follows the Poisson process with intensity $\mu_1$. A similar description applies to the insurer's client pool, see Section \ref{state-section}.
 
 As $\mu_1$ describes the whole population and $\mu_2$ a sample of the population, \cite{WoCW, WoCW2017} postulate that the former is the cointegrating factor of the pair. In other words, the process $\mu_1$ follows a mean reversion itself whereas the trend of $\mu_2$ depends on $\mu_1$. Consider a simple example from their papers. The pair of mortality processes is postulated as follows.
	\begin{eqnarray}\label{Markov}
	\mu_1(t)&=& \mu_1(0) + \int_0^t (b_1  -\theta_{11}\mu_1(s))ds + \int_0^t \sigma_1dW_1(s),\\
	\mu_2(t)&=& \mu_2(0) + \int_0^t (b_2  -\theta_{21}\mu_1(s)- \theta_{22}\mu_2(s))ds + \int_0^t \sigma_2(\rho dW_1(s) + \sqrt{1-\rho^2}dW_2(s)), \nonumber
	\end{eqnarray}
	where the parameters $b_1, b_2, \theta_{11}, \theta_{21}, \theta_{22}, \rho, \sigma_1$, and $\sigma_2$ are all constants, and $W_1$ and $W_2$ are independent Wiener processes.     {
    It is clear that the vector $(\mu_1, \mu_2)^\top$ is non-stationary.
    However, the linear combination  $\theta_{21}\mu_1(s) + (\theta_{22} - \theta_{11} )\mu_2(s) $ is a stationary process other than  $\mu_1(t)$, implying that $\mu_1(t)$  and $\mu_2(t)$ constitute a cointegration system \citep{C1999}. The use of a stochastic integral equation in \eqref{Markov} aims to facilitate comparison, as the non-Markovian nature of LRD often uses stochastic integral equations to describe the model. {Such a model setting was inspired by the discrete-time models in the literature \citep{Zhou2014, Dowd2011}. To facilitate the use of the longevity market, insurers must scale their experienced portfolio's mortality to match the national figures. For instance, the gravity model of \cite{Dowd2011} suggests regressing the insurer's experienced force of mortality on the national figures, constituting the form of \eqref{Markov}.}

    The empirical studies show strong evidence that national mortality data exhibits the LRD feature \citep{Yan2021, Yaya2019}.  To incorporate LRD into the national   {force of mortality},   { the VMM is established  as $\hat{\mu}_1(t) = m_1(t) + \mu_1(t)$, where $\hat{\mu}_1$ is the  force of mortality, $m_1(t)$ is  a deterministic function  capturing the main part of the force of mortality, and  $\mu_1(t)$ satisfies the  stochastic Volterra integral equation (SVIE):}
		\begin{equation}\label{Wang-mu1}
	\mu_1(t)= \mu_1(0) + \int_0^t K_{11}(t-s)(b_1(s) -\Theta_{11}\mu_1(s))ds + \int_0^t K_{11}(t-s)\sigma_1(\mu_1(s))dW_1(s),
	\end{equation}
where $K_{11}$ is the kernel function for capturing LRD in various different sense. {To fix idea, we use the fractional kernel: $K_{11}(t) = \frac{t^{H-\frac{1}{2}}}{\Gamma(H+\frac{1}{2})}$, where $H\in (0, 1)$ is the Hurst parameter. Once $H = \frac{1}{2}$, we have $K_{11} = 1$ and, hence, $\mu_1$ in \eqref{Wang-mu1} reduces to a classic Markovian mortality model \citep{Biffis}. When $H>\frac{1}{2}$, $\mu_1$ in \eqref{Wang-mu1} reflects the LRD feature. In fact, the noise term in \eqref{Wang-mu1} is an RL-fBm if $\sigma_1(\mu_1(s)) \equiv \sigma_1$ is a constant.}
\begin{remark}
      {LRD, which is  also called long memory or long-range persistence, refers to the strength of statistical dependence between lagged observations  in a time series.  Given a stationary time series process $Y_{1:T} :=(Y_1, Y_2, \cdots, Y_T)$ with $Y_{1:T} \in (\mathbb{N} \cup  \{0\})$, the condition for a long-range dependence process provided by \cite{Beran1994} is 
    \begin{align*}
    \lim_{n \rightarrow \infty} \sum_{j = -n}^n |\rho(j)| \rightarrow \infty, \text{ where } \rho(j) = \frac{{\rm Cov}(Y_t, Y_{t+j})}{\sqrt{{\rm Var}(Y_t) {\rm Var}(Y_{t+j})}}.
    \end{align*}}
    
   { For continuous-time processes,  fractional Brownian motion (fBM) is a typical process that exhibits  LRD and is an extension of the standard Brownian motion.  Fractional Brownian motion  $\left\{W_t^H ; t \in \mathbb{R}\right\}$ is the unique Gaussian process with mean zero and autocovariance function
$$
\mathbb{E}\left[W_t^H W_s^H\right]=\frac{1}{2}\left\{|t|^{2 H}+|s|^{2 H}-|t-s|^{2 H}\right\}
$$
where $H \in(0,1)$ is the Hurst parameter.  We can see that  the fractional Brownian motion  reduces to a standard Brownian motion  when $H = \frac{1}{2}$.    If $H \in (\frac{1}{2}, 1)$, increments are positively correlated, and $\sum_{n=1}^\infty\mathbb{E}[W^H_1 (W^H_{n+1} -  W^H_{n} )] = \infty$, i.e., the process  exhibits LRD. In addition, almost-all trajectories of fBM are locally H\"older's continuous of any order strictly less than $H$.} 

  {LRD has been found in mortality data by numerous empirical studies \citep{Yan2020, Yan2021}. In order to incorporate LRD into mortality modeling and maintain tractability at the same time, \cite{WaCW} first propose to model  the forces of mortality by affine Volterra processes, which  have the Riemann-Liousville (RL)  fBm as the noise term.  }
\end{remark}

While cointegration and LRD have been widely investigated in the literature separately, no existing mortality model is able to capture them simultaneously.  To do so, we extend \eqref{Wang-mu1} to a two-dimensional VMM with cointegration as follows.   {The force of mortality $\hat{\mu}_i(t) = m_i(t) + \mu_i(t)$, where $m_i(t), i = 1, 2$, are bounded deterministic functions and}
\begin{align}
\mu_1(t) & = \mu_1(0) + \int_{0}^{t}K_1(t-s)(b_1 - \theta_1 \mu_1(s))ds + \int_{0}^{t}K_1(t-s)\sigma_1(\mu_1(s))dW_1(s),\label{mu1}\\
\mu_2(t) & = \mu_2(0) + \int_{0}^{t}\beta_1 K_1(t-s)(b_1 - \theta_1 \mu_1(s))ds  + \int_0^tK_2(t-s)(b_2 - \beta_2 \mu_1(s) - \theta_2 \mu_2(s))ds \nonumber \\
& + \int_{0}^{t}\beta_1 K_1(t-s)\sigma_1(\mu_1(s))dW_1(s) + \int_0^t K_2(t-s)\sigma_2(\mu_2(s)) dW_2(s), \label{mu2}
\end{align}
where $K_1(t) = \frac{t^{H_1 -\frac{1}{2}}}{\Gamma(H_1+\frac{1}{2})}$ and $K_2(t) = \frac{t^{H_2 -\frac{1}{2}}}{\Gamma(H_2+\frac{1}{2})}$  with Hurst parameters $H_1, H_2 \in (0, 1 )$. In addition, $b_1, b_2, \beta_1, \beta_2, \theta_1$ and $\theta_2$ are positive constants.  The national   {force of mortality}  $\mu_1$ is an affine Volterra process that can reflect the LRD property. $\mu_2$ disappears from the dynamic of $\mu_1$ under our setting, which indicates that the index mortality for the longevity bond is not affected for a certain portfolio of the insurer. The dynamic of $\mu_2$ depends on $\mu_1$ with the values $\beta_1$ and $\beta_2$ determining the correlation and cointegration between $\mu_1$ and $\mu_2$.  When $K_1 = K_2 \equiv 1$, our model is equivalent to those in \cite{WoCW,WoCW2017}. Certainly, our framework is more general than theirs, as it captures both LRD and cointegration within   {forces of mortality} simultaneously. {Hence, our model can be served as the LRD version of the gravity model of \cite{Dowd2011} and is useful for insurers who would like to manage their longevity risk with the longevity market.} {In addition, when $H_2$ is set to 1/2 in \eqref{mu2}, then the noise term in the process of $\mu_2$ appears as a linear combination of a fBm and Brownian motion, and hence an mfBm, consistent to the one considered by \cite{Zhou2022}.}

Under the model setting \eqref{mu1} and \eqref{mu2}, $\mu = (\mu_1, \mu_2)^\top$ satisfies the following SVIE:
	\begin{equation}\label{mu0}
	\mu(t)= \mu(0) + \int_0^t K(t-s)(b(s) -\Theta\mu(s))ds + \int_0^t K(t-s)\sigma(\mu(s))dW(s),
	\end{equation}
 with $\mu(0) = (\mu_1(0), \mu_2(0))^\top$, $b(s) = (b_1, b_2)^\top$,  
 \begin{align*}\label{K}
K = \left( \begin{array}{cc}
\frac{t^{H_1 -\frac{1}{2}}}{\Gamma(H_1+\frac{1}{2})} & 0\\
\beta_1 \frac{t^{H_1 -\frac{1}{2}}}{\Gamma(H_1+\frac{1}{2})} & \frac{t^{H_2 -\frac{1}{2}}}{\Gamma(H_2+\frac{1}{2})}
\end{array} \right), ~ \Theta = \left(\begin{array}{cc}
\theta_1 & 0\\
\beta_2 & \theta_2
\end{array}\right), ~ 
\sigma(\mu(s)) = \left(\begin{array}{cc}
    \sigma_1(\mu_1(s)) & 0 \\
     0 & \sigma_2(\mu_2(s))
\end{array} \right)
\end{align*}
and $W = (W_1, W_2)^\top$  a two-dimensional standard Brownian motion under the physical measure $\mathbb{P}$. We assume that the $(\sigma_1(\mu_1))^2$ and $(\sigma_2(\mu_2))^2$ are linear in $\mu_1$ and $\mu_2$ respectively.  It follows that the covariance matrix $a(\mu) = \sigma(\mu)\sigma(\mu)^\top$ has an affine dependence on $\mu$, i.e., 
	\begin{align*}
	a(\mu) =  A^0 + A^1\mu_1 + A^2 \mu_2
	\end{align*}
	for some two-dimensional constant matrices $A^0$, $A^1$, and $A^2 \in \mathbb{R}^{2\times 2}$. For any row vector $u \in \mathbb{R}^2$, we define the row vector
	\[ A(u) = (uA^1u^\top, uA^2u^\top).\]
 In contrast with \cite{WaCW}, our mortality model \eqref{mu0} is a two-dimensional affine Volterra process, which enables us to incorporate cointegration between the two morality rates and the LRD property simultaneously.  The flexible model structure allows for different mortality models, such as the mfBm noise, as shown in Section \ref{Sec:explicit}. 

\subsection{Motivation and the identification issue of mfBm}
{To further motivate the use of cointegration in the LRD environment,  we develop a novel important Proposition \ref{pro:Hurst} to show the identification problem of mfBm modelling. We then show how cointegration can resolve the difficulty.

\begin{pro}\label{pro:Hurst}
For a fixed constant $T>0$, the $\mu_2$ in \eqref{mu2} is H\"older continuous on $[0, T]$ of any order strictly smaller than $\min\{H_1, H_2\}$. 
\end{pro}
\begin{proof}
    See Appendix \ref{appendix:Hurst}.
\end{proof}

  As $m_2$ is a deterministic function, Proposition \ref{pro:Hurst} asserts that the regularity of the insurer's experienced   {force of mortality $\hat{\mu}_2$} resembles the regularity of a process with the smallest Hurst parameter value between $H_1$ and $H_2$.  In other words, if we model $\mu_2$ by mfBm, then it is equivalent to setting $H_1>1/2$ for a RL-fBm and $H_2 = 1/2$ for the Brownian motion.  In such a situation, the insurer may fail to detect the LRD feature and has an impression that $H=1/2 = \min(H_1, H_2)$ or models driven by Brownian motion are sufficient for application purposes. Certainly, if  the insurer has the prior knowledge to model the noise term as mfBm, then parameter estimation is still possible as demonstrated empirically by \cite{Zhou2022}.

Using the publicly available national mortality tables, the LRD feature has been successfully identified in the literature \citep{Yan2020, Yan2021}  using statistical techniques without assuming mfBm. Combining this with Proposition \ref{pro:Hurst}, we are motivated to propose \eqref{mu0} for integrating the LRD feature estimated from the national mortality tables into $\mu_2$ through cointegration.

 \subsection{Useful mathematical property of the model}
	By Theorem 3.4 in \cite{AJ}, the SVIE \eqref{mu0} admits a continuous weak solution for any initial value $\mu(0) \in \mathbb{R}^2$. The following lemma\footnote{For two functions $F$ and $K$, the convolution $F*K$ is defined by
  $F*K(t)=\int_0^t F(s)K(t-s)ds.$ For a kernel $K$, the resolvent, or resolvent of the second kind of $K$, is the kernel $R$ such that $K*R = R*K = K-R$. For a fractional kernel $K(t) = c \frac{t^{\alpha-1}}{\Gamma(\alpha)}$ where $c$ and $\alpha \in (\frac{1}{2}, \frac{3}{2})$ are constants, its resolvent is given by $R(t) = ct^{\alpha -1}E_{\alpha, \alpha}(-ct^\alpha)$. Here, $E_{\alpha, \beta}(z) =  \sum_{n = 0}^\infty \frac{z^n}{\Gamma(\alpha n + \beta )}$ denotes the Mittag–Leffler function.} provides the Laplace--Fourier transform of $\mu$, which ensures the analytical tractability of the proposed mortality model \eqref{mu0}. 
	\begin{lemma}\label{lemma:expmu}
		\citep{AJ} Suppose that $\mu$ follows the SVIE \eqref{mu0} and a constant row vector $f \in (\mathbb{R}^2)^*$. For any $0\leq t \leq T$, assume that $\psi\in L^2([0, T],(\mathbb{R}^2)^*)$ is a solution to the Volterra--Riccati equation $\psi = (f - \psi\Theta + \frac{1}{2}A(\psi))*K$. Then:
		\begin{equation}\label{expmu}
		\mathbb{E}\left[\left.e^{\int_{0}^{T}f \mu(s)ds}\right|\mathcal{F}_t\right] = \exp(Y_t(T)), 
		\end{equation}
		where 
		\begin{align}\label{Y}
		\begin{split}
		Y_t(T) &= Y_0(T) + \int_{0}^{t}\psi(T-s)\sigma(\mu(s))dW(s) - \frac{1}{2}\int_{0}^{t}\psi(T-s)a(\mu(s))\psi(T-s)^\top ds,\\
		Y_0(T) &= \int_{0}^{T}\left[f\mu(0) + \psi(s)(b(s) - \Theta\mu(0)) + \frac{1}{2}\psi(s)a(\mu(0))\psi(s)^\top\right]ds.
		\end{split}
		\end{align}
		In addition, $Y$ has an alternative expression: 
		\begin{equation}\label{Y2}
		Y_t(T) = \int_{0}^{T}f \mathbb{E}\left[\mu(s)|\mathcal{F}_t\right]ds + \frac{1}{2}\int_{t}^{T}\psi(T-s)a(\mathbb{E}[\mu_s|\mathcal{F}_t])\psi(T- s)^\top ds, 
		\end{equation}
		where 
		\begin{equation}\label{expectation:mu}
		\mathbb{E}[\mu(T)|\mathcal{F}_t] = \left(I - \int_{0}^{T}R_\Theta(s)ds\right)\mu_0 + \int_{0}^{T}E_\Theta(T-s)b(s)ds + \int_{0}^{t}E_\Theta(T- s)\sigma(\mu)dW(s),
		\end{equation}
		where $R_\Theta$ is the resolvent of $K\Theta$ and $E_\Theta = K- R_\Theta*K$.
	\end{lemma}
	
	\subsection{Longevity bond price}

     To mitigate the longevity risk associated with the experienced mortality portfolio, the insurance company  purchases the longevity bonds as the  hedging instrument, which is referenced to the national mortality $\mu_1$. And also, investment in a risk-free bond is allowed.  In this part, we present the setting of the financial market and provide the closed-form expressions for zero-coupon risk-free bond and longevity bond.

      {We assume  that  financial	 market	 and	 mortality rates	 are	 independent.}  Although empirical studies in literature \citep{CT2007} also document the LRD property of  the interest rate, we assume that the  interest rate $r(\cdot)$ a Markov affine process to focus on the effect of longevity risk. Specifically, $r(\cdot)$ follows the dynamic:
	\begin{equation}\label{interest}
	dr(t)=(b_r -  \theta_r r(t))dt +\sigma_r(r(t))dW_r(t),
	\end{equation}
	where $b_r$ and $\theta_r$ are two positive constants, and $W_r$ is a standard Brownian motion independent of $W$. Furthermore, the variance of the interest rate $\left(\sigma_r(r(t))\right)^2$ linearly depends on $r(t)$ with constant coefficients. Note that LRD of interest rate can be easily incorporated in our framework by modelling $r$ with an affine Volterra process.

	Let $\mathbb{Q}$ be the pricing measure equivalent to $\mathbb{P}$. The zero-coupon risk-free bond price is given by
    \begin{equation}\label{Bond-expectation}
    \mathcal{B}(t,T) = \mathbb{E}^{\mathbb{Q}}\left[\left.e^{-\int_{t}^{T}r(s)ds}\right|\mathcal{F}_t\right].
    \end{equation}
     The longevity bond is based on the national index.   {Suppose that $l(t)$ is  the  number of survivors from the national population at time $t$. Zero-coupon longevity bond $\mathcal{B}_L(t,T)$ pays $l(T)/l(t)$ at time $T$. As $\hat{\mu}_1$ is the force of mortality for the national population, we have $\frac{l(T)}{l(t)} = \exp\left( -\int_t^T\hat{\mu}_1(s)ds\right)$}.  As the financial market is independent of the   {force of mortality},  the zero-coupon longevity bond admits the following representation:
    \begin{equation}\label{Long-B-expectation}
     \mathcal{B}_L(t,T) =\mathbb{E}^{\mathbb{Q}}\left[\left.e^{-\int_{t}^{T}r(s)+   {\hat{\mu}_1(s)} ds}\right|\mathcal{F}_t\right] 
    =\mathcal{B}(t,T)\mathbb{E}^{\mathbb{Q}}\left[\left.e^{-\int_{t}^{T}  {\hat{\mu}_1(s)} 
 ds}\right|\mathcal{F}_t\right], 
     \end{equation}
	
	To derive the stochastic dynamics of both $\mathcal{B}(t,T)$ and $\mathcal{B}_L(t,T)$, we specify the pricing measure $\mathbb{Q}$ by using the measure change that is modeled by the Girsanov theorem. Let $\boldsymbol{W} \triangleq (W^\top, W_r)^\top$. The pricing measure $\mathbb{Q}$ is given by
	\begin{equation}
	\frac{d\mathbb{Q}}{d\mathbb{P}}= e^{\int_{0}^{t}\zeta(s)^\top d\boldsymbol{W}(s)-\frac{1}{2}|\zeta(s)|^2ds},\nonumber
	\end{equation} 
	where $\zeta(s) = (\varphi(s)^\top, \vartheta(s))^\top$ with $\varphi:[0, T] \rightarrow \mathbb{R}^2$, and $\vartheta: [0, T] \rightarrow \mathbb{R}$ satisfies Novikov's condition:
	\begin{equation}\label{Novikov}
	\mathbb{E}\left[e^{\frac{1}{2}\int_{0}^{T}\zeta(s)^\top \zeta(s)ds}\right] < \infty. 
	\end{equation}
    Under the pricing measure $\mathbb{Q}$, we have $$dW^{\mathbb{Q}}(t) = dW(t) - \varphi(t)dt, ~ dW_r^{\mathbb{Q}}(t) = dW_r(t) -\vartheta(t)dt.$$  
    Inspired by \cite{WaCW}, we apply an affine retaining transform so that the mortality and interest rates maintain an affine structure under the measure $\mathbb{Q}$, as in \eqref{mu0} and \eqref{interest}. Then, the zero-coupon bond price in \eqref{Bond-expectation} admits a closed-form expression by classical results from transforms of affine diffusion (see Appendix \ref{appendix:affine}). By It\^o's formula, the $\mathbb{P}$-dynamic of the bond price is given by
	\[d\mathcal{B}(t,T) = \mathcal{B}(t,T)(r(t)+\nu_{\mathcal{B}}(t))dt + \mathcal{B}(t,T)\sigma_b(t)dW_r(t),\]
	where  $\nu_{\mathcal{B}}(t) = \vartheta(t)\sigma_b(t)$ for some $\mathcal{F}_t$-adapted volatility rate $\sigma_b(t)$. Examples of $\sigma_b$ under specific settings are presented in Propositions \ref{pro:Vas} and \ref{Pro:CIR} later in this paper. With the affine structure maintained under measure $\mathbb{Q}$, Lemma \ref{lemma:expmu} shows that the zero-coupon longevity bond price in \eqref{Long-B-expectation} also admits a closed-form expression. By using Lemma \ref{lemma:expmu} and It\^o's formula, the longevity bond price $\mathcal{B}_L(t, T)$ satisfies the stochastic differential equation (SDE):
	\[d\mathcal{B}_L(t,T) = \mathcal{B}_L(t,T)(r(t) +   {\hat{\mu}_1(t)} 
+ \nu_{\mathcal{L}}(t))dt + \mathcal{B}_L(t,T)\sigma_l(t)^\top dW(t) +\mathcal{B}_L(t,T)\sigma_bdW_r(t),\]
	where $\nu_{\mathcal{L}}(t) = \nu_{\mathcal{B}}(t) + \varphi(t)^\top\sigma_l(t)$ for a certain $\mathcal{F}_t$-adapted process $\sigma_l(t)$ in $\mathbb{R}^2$.  Under concrete settings, the forms of $\sigma_l$ are derived in Proposition \ref{pro:Vas} and \ref{Pro:CIR} later in this paper. As an amount $\mathcal{B}_L(t,T)$ invested in the longevity bond at time $t$ becomes $e^{-\int_{t}^{\tau}\hat{\mu}_1(s)ds}\mathcal{B}_L(\tau,T)$ at $\tau >t$, the value of holding one unit of the zero-coupon longevity bond $\mathcal{B}_L(t)$ satisfies
	\begin{equation}\label{longevity}
	d\mathcal{B}_L(t,T) = \mathcal{B}_L(t,T)(r(t) + \nu_{\mathcal{L}}(t))dt + \mathcal{B}_L(t,T)\sigma_l(t)^\top dW(t) +\mathcal{B}_L(t,T)\sigma_bdW_r(t).
	\end{equation}
	The vector $\varphi(t)$ is the market price of mortality risk and $\vartheta(t)$ is the market price of the interest rate risk.  In fact, because the longevity bond pricing only involves the national   {force of mortality}  and is not related to the insurer's experienced mortality portfolio, we have $\sigma_l \triangleq (\sigma_{l1}, \sigma_{l2})^\top = (\sigma_{l1}, 0)^\top$.
	
\begin{remark}
   { The assumption that the financial market and mortality rate  are independent is often used in the literature (see, e.g., \cite{Biffis}) to simplify the modeling framework.  Actually, it is possible to relax this assumption in this paper.  If we consider that the interest rate and force of mortality are correlated, for instance, by considering that $W_r$ is  correlated with $W$, the valuation for the longevity bond would be complicated. In \eqref{Long-B-expectation}, we cannot separate the expectations; however, it is also possible to derive the  closed-form solution by forming the vector $(r, \mu_1)^\top$ as an affine Volterra process and applying  Lemma \ref{lemma:expmu}. }
\end{remark}

    \section{The longevity hedging problem}\label{state-section}
    Consider an insurance company that provides annuities/pension scheme and life insurances to the insurance market. While the insurance company aims to manage the risk generated from $  {\hat{\mu}_2} $, the available hedging vehicle (i.e. the longevity bond) is based on $  {\hat{\mu}_1}$. This further complicates the mathematical treatment of the hedging problem. We use a Poisson process $N_2(t)$ with a stochastic intensity $c_1  {\hat{\mu}_1(s)}$ where $c_1 \in \mathbb{R}_+$  to denote the frequency of life insurance claims experienced by the company. 
    Suppose that $c_2\in \mathbb{R}_+$ is the amount of the pension/annuity   {payment} that the insurer must pay at the initial time. At time $t$, the survival proportion of the population becomes $c_2e^{-\int_{0}^{t}  {\hat{\mu}_2(s)} ds}$. Therefore, the pension annuity net cash outflow is represented by $\pi(t) =  c_2e^{-\int_{0}^{t}  {\hat{\mu}_2(s)} ds}$ in this paper. The insurer's hedging horizon is $T_0$, and the company uses the longevity bond with a maturity time $T> T_0$ as the hedging instrument.
    
	Let $u_0(t)$, $u_1(t)$, and $u_2(t)$ denote the amounts of investment in a bank account, zero-coupon longevity bond, and zero-coupon risk-free bond, respectively. The wealth process of the insurer is
	\begin{equation}
	\bar{X}(t)=u_0(t)+ u_1(t)+ u_2(t)- \sum_{i=1}^{N_2(t)}z_i-\Pi(t),~ t\in[0,T_0],
	\end{equation}
	where $\Pi(t) = \int_{0}^{t}\pi(s)ds$,  and $\{z_i\}_{i=1}^{\infty}$ denotes the positive independent and identically distributed (iid) insurance claims. Note that the insurer's pension annuity net cash outflow and insurance claims depend on $  {\hat{\mu}_2}$, the insurer's experienced   {force of mortality}.   {In this paper, we consider that the insurer is only looking at one specific cohort  of individuals. In other words,  the individuals are all the same age at time $t$.} 
	
	We further postulate that the insurance claims $z$ have the  probability density function $g(z)$ on $\mathbb{R}_+$ with finite moments $\mathbb{E}[z], \mathbb{E}[z^2]$ (i.e., $\mathbb{E}[z^2] = \int_{\mathbb{R}_+}z^2g(z)dz< \infty$). For convenience, we use a Poisson random measure to denote the compound Poisson process. Specifically,  $\sum_{i=1}^{N_2(t)}z_i = \int_{0}^{t}\int_{\mathbb{R}_+}zN_2(ds,dz)$, where $N_2(dt , dz)$ is a Poisson random measure on $\Omega\times[0,T_0]\times\mathbb{R}_+$ with the compensator $c_1  {\hat{\mu}_2(t)} g(z)dzdt \triangleq \delta(dz)dt$. We use $F^2(t, T; \mathbb{R})$ to denote the set of all $\{\mathcal{F}_s\}_{s \in[t, T]}$-predictable processes $X(\cdot, \cdot): \Omega\times[t, T]\times\mathbb{R}_+\rightarrow \mathbb{R}$ such that $\mathbb{E}\left[\int_{[t, T]\times\mathbb{R}_+}\|X(s, \cdot)\|^2_{L^2}ds\right] < \infty$, where $\|X(s, z)\|^2_{L^2}$ := $\int_{\mathbb{R}_+}|X(s, z)|^2\delta(dz)$. 
	By setting $\widetilde{N}(dt,dz) = N_2(dt,dz)-\delta(dz)dt$, the insurer's wealth $\bar{X}(t)$ satisfies
	\begin{equation}\label{wealth}
	d\bar{X}(t)= (\bar{X}(t)r(t)+\nu(t)^\top u(t)-\pi(t) - \mathbb{E}[z]c_1  {\hat{\mu}_2(t)})dt + u(t)^\top\sigma_S(t)^\top d\boldsymbol{W}(t)-\int_{\mathbb{R}_+} z\widetilde{N}(dt, dz),
	\end{equation}
	where   {$u(t) = (u_1(t), u_2(t))^\top$}, $\nu(t)=(\nu_{\mathcal{L}}, \nu_{\mathcal{B}})^\top$, and
\begin{equation*}
    \sigma_S(t)^\top = \left(\begin{array}{ccc}
 \sigma_{l1}(t)  & 0 &  \sigma_b(t) \\
    0 & 0 & \sigma_b(t)
\end{array}  \right).
\end{equation*}
	
	In this paper, we study the problem with respect to the discounted wealth process $X \triangleq e^{-\int_{0}^{t}r(s)ds}\bar{X}(t)$. The objective of an insurer is to minimize
	\begin{align}\label{objective}
	\begin{split}
	J(t, X; u(\cdot)) & = \frac{1}{2}{\rm Var}_t(X(T_0)) - \lambda\mathbb{E}_t[X(T_0)] \\
	& = \frac{1}{2}\mathbb{E}_t[X(T_0)^2] - \frac{1}{2}\mathbb{E}_t^2[X(T_0)] - \lambda\mathbb{E}_t[X(T_0)]
	\end{split}
	\end{align}
	by an admissible control $u$ for $X(t)$ at time $t$, where $\mathbb{E}_t[\cdot] = \mathbb{E}[\cdot|\mathcal{F}_t]$ and the parameter $\lambda$ reflects the insurer's risk aversion level. By \eqref{wealth}, 
	\begin{align}
	\begin{split}\label{Xtilde}
	dX(t) = e^{-\int_{0}^{t}r(s)ds}\left(\nu(t)^\top u(t)-\pi(t) - \mathbb{E}[z]c_1  {\hat{\mu}_2(t)} \right)dt
       + e^{-\int_{0}^{t}r(s)ds}u(t)^\top\sigma_S(t)^\top d\boldsymbol{W}(t)\\ -\int_{\mathbb{R}_+}ze^{-\int_{0}^{t}r(s)ds}\widetilde{N}(dt, dz). 
	\end{split}
	\end{align}

\begin{definition}\label{def1}
$u(\cdot) $ is called an admissible control if $u(\cdot) \in H_\mathcal{F}^2(0,T_0;\mathbb{R}^2, \mathbb{P})$, and if Equation \eqref{Xtilde} admits a unique strong solution $X \in S_\mathcal{F}^2(0,T_0;\mathbb{R}, \mathbb{P})$. 
\end{definition}
\begin{lemma}\label{lemma:E_mu}
Let $\mu$ be the continuous solution to \eqref{mu0}. For any constant $p \geq 2$, there exists a constant $C$ such that
\[\sup_{t\leq T}\mathbb{E}[|\mu_t|^p] \leq C.  \]
\end{lemma}
\begin{proof}
This is a direct consequence of Lemma 3.1 in \cite{AJ} and \eqref{mu0}.
\end{proof}
Under the interest rate model setting in \eqref{interest}, we prove the following Lemma \ref{Lemma:interest} concerning the boundedness property of $r(\cdot)$ in Appendix \ref{appendix:affine}. The general result for the admissibility of the equilibrium control is presented in Theorem \ref{Thm:admissible}.
\begin{lemma}\label{Lemma:interest}
$\sup_{0\leq s \leq T_0}\mathbb{E}\left[e^{-\int_0^s 4r(v)dv}\right] < \infty$ and $\sup_{0\leq s \leq T_0}\mathbb{E}\left[e^{\int_0^s qr(v)dv}\right] < \infty$ for any $q\geq 0$. 
\end{lemma}
\begin{theorem}\label{Thm:admissible}
Suppose that $\pi \in L_\mathcal{F}^4(0, T_0; \mathbb{R}, \mathbb{P})$. 
If $u \in H^2_\mathcal{F}(0, T_0; \mathbb{R}^2, \mathbb{P})$ such that $\nu^\top u\in L_\mathcal{F}^4(0, T_0;  \mathbb{R}, \mathbb{P})$, and $\sigma_S u \in L_\mathcal{F}^4(0, T_0;  \mathbb{R}^3, \mathbb{P})$, then $u$ is an admissible control.
\end{theorem} 
\begin{proof}
    See Appendix \ref{Appendix:Thm1}.
\end{proof}
The variance term of the objective in \eqref{objective} causes a time-inconsistency issue, which means that the optimal hedging strategy derived at the initial time  is dependent  on the initial wealth level and would be sub-optimal at a later time point,  as widely documented in the literature. To overcome the time-inconsistency issue, we adopt the well-received open-loop framework \cite{HJZ2012, HJZ2017} to investigate equilibrium controls. 
\begin{definition}\label{def2}
Let $u^*(\cdot) \in H_\mathcal{F}^2(0,T_0;\mathbb{R}^2, \mathbb{P})$ be a given control and $X^*(\cdot)$ be the corresponding state process. Then, the control $u^*$ is called an equilibrium strategy if for any $t\in [0,T_0]$ and $\epsilon>0$, $\eta \in H_{\mathcal{F}_t}^2(\Omega;\mathbb{R}^2)$, 
\begin{equation}\label{eq:def}
\mathop{ \liminf}\limits_{\epsilon \downarrow 0} \frac{J(t, X^*(t), u^{t,\epsilon, \eta}(\cdot)) - J(t, X^*(t), u^*(\cdot))}{\epsilon} \geq 0,
\end{equation}
where $ u^{t,\epsilon, \eta}(\cdot) = u^*(\cdot) + \eta\boldsymbol{1}_{[t, t+\epsilon)}(\cdot)$.
\end{definition}

The control defined in Definition \ref{def2} is known as an open-loop control form because the control $u^{t,\epsilon, \eta}(\cdot)$ is not presumed to depend on the state process $X$. Instead, it depends on any adapted process $\eta\in H_{\mathcal{F}_t}^2(\Omega;\mathbb{R}^2)$ in the sense of spike variation, embracing the possibility of feedback and closed-loop control forms. This open-loop control definition is necessary for our problem because the state process has LRD (non-Markovian) coefficients. The equilibrium control may depend on the entire historical paths of the coefficients and the state process.

\subsection{Preliminary results}
To derive the equilibrium strategy under Definition \ref{def2}, we adopt the BSDEs approach in \cite{HJZ2012}. 
For any $ t \in [0,T_0]$, we define the adjoint process $ (p^*(s; t), k^*_1(s; t), k^*_2(s, z; t)) \in S_\mathcal{F}^2(t,T_0;\mathbb{R}, \mathbb{P}) \times L_\mathcal{F}^2(t,T_0;\mathbb{R}^3, \mathbb{P}) \times F^2(t,T_0;\mathbb{R})$ in the interval $[t, T_0]$ to be the solution of the following BSDE, with the discount wealth process $X^* \in S_\mathcal{F}^2(0,T_0;\mathbb{R}, \mathbb{P})$ under $u^*$: 
\begin{align}\label{BSDE:p}
\left\{
\begin{array}{lr}
dp^*(s; t) = k_1^*(s;t)^\top d\boldsymbol{W}(s) + \int_{\mathbb{R}_+} k_2^*(s,z;t)\widetilde{N}(ds,dz),\\
p^*(T_0; t) = X^*(T_0) - \mathbb{E}_t[X^*(T_0)] - \lambda. 	\end{array}
\right.
\end{align}
With the adjoint process $ (p^*(s; t), k^*_1(s; t), k^*_2(s, z; t))$ defined above, we have the following theorem to estimation the term $J(t, X^*(t), u^{t, \epsilon,\eta}) -  J(t, X^*(t), u^*)$ in \eqref{eq:def} in Definition \ref{def2}. 
\begin{theorem}\label{Thm:spike}
	For any $t \in [0,T_0]$, $\epsilon>0$ and $\eta \in H_\mathcal{F}^2(\Omega;\mathbb{R}^2)$, let $u^*(\cdot) \in H_\mathcal{F}^2(0,T_0;\mathbb{R}^2, \mathbb{P})$ be a given control and $u^{t,\epsilon, \eta}(\cdot)$ be as defined in Definition \ref{def2}. If $\nu(\cdot) \in H_\mathcal{F}^2(0,T_0;\mathbb{R}^2, \mathbb{P})$ and each component of $\sigma_S^\top\sigma_S$ belongs to $H_\mathcal{F}^2(0, T_0; \mathbb{R}, \mathbb{P})$, then
	\[J(t, X^*(t), u^{t, \epsilon,\eta}) -  J(t, X^*(t), u^*) = \mathbb{E}_t\int_{t}^{t + \epsilon} \left[\langle \Lambda(s; t), \eta_s \rangle +\frac{1}{2} \langle \Phi(s) \eta_s, \eta_s \rangle\right] ds + o(\epsilon),\]
	where
	\begin{equation}\label{Lamtilde}
     \Lambda(s; t) = e^{-\int_{0}^{s}r(v)dv}\nu(s)p^*(s;t) + e^{-\int_{0}^{s}r(v)dv}\sigma_S(s)^\top k_1^*(s; t),
	\end{equation}
	and $\Phi(s) =  e^{-\int_{0}^{s}2r(v)dv}$.
\end{theorem}
\begin{proof}
    See Appendix \ref{Appendix:Thm2}. 
\end{proof}
The following Proposition \ref{Pro:Lam} presents a special form for the process $\Lambda$. By Proposition \ref{Pro:Lam} and Theorem \ref{Thm:spike}, the sufficient and necessary condition for an equilibrium hedging strategy is provided in Proposition \ref{Pro:Lam0}. The proofs mimic those of Propositions 5 and 6 in \cite{WW2021} and hence we omit them here. Although Propositions \ref{Pro:Lam} and \ref{Pro:Lam0} are not the key results of this paper, they are useful in the proofs for our main contributions in the next section.
\begin{pro}\label{Pro:Lam}
	For a given control $u^*$ and the corresponding state process $X^*$, the solution to the BSDE \eqref{BSDE:p} satisfies $k_1^*(s; t_1) = k_1^*(s; t_2)$ and $k_2^*(s,z; t_1) = k_2^*(s,z; t_2)$ for $a.e.~s \geq \max(t_1, t_2)$. Moreover, there exist processes $\Lambda_0$ in $\mathbb{R}^2$ and $\omega \in S^2_\mathcal{F}(t, T_0; \mathbb{R}, \mathbb{P})$  such that
	\[ \Lambda(s; t) =  \Lambda_0(s) + e^{-\int_{0}^{s}r(v)dv}\nu(s)\omega_t. \]
\end{pro}
\begin{pro}\label{Pro:Lam0}
	Let $ (p^*(s; t), k^*_1(s; t), k^*_2(s, z; t)) \in S_\mathcal{F}^2(t,T_0;\mathbb{R}, \mathbb{P}) \times L_\mathcal{F}^2(t,T_0;\mathbb{R}^3, \mathbb{P}) \times F^2(t,T_0;\mathbb{R})$ be the unique solution to the BSDE \eqref{BSDE:p}. Then, $u^*$ is an equilibrium control if and only if
	\begin{equation}\label{Lam}
	\Lambda(t; t) =  e^{-\int_{0}^{t}r(v)dv}\nu(t) p^*(t;t) + e^{-\int_{0}^{t}r(v)dv}\sigma_S(t)^\top k_1^*(t; t) = 0, ~ a.s., ~a.e.,~ t \in [0, T_0].
	\end{equation}
\end{pro}
\subsection{Explicit equilibrium hedging strategies}
\label{Sec:explicit}
We first derive the general solution to the time-consistent hedging problem \eqref{objective}, which involves the solution of a BSDE. Under certain conditions, the general uniqueness result for the equilibrium investment and hedging strategies is provided in Theorem \ref{Thm:unique}.  Under the settings of a deterministic and volatility-driven market price of risks, the BSDE is solved explicitly and the closed-form investment and hedging strategies are derived together with rigorous proofs about their admissibility and uniqueness.

By condition \eqref{Lam}, we derive the general form of the equilibrium strategy in the following. Consider the solution form for $p^*(s; t)$ as
\[p^*(s; t) = X^*(s) + \Gamma_s^{(1)} - \lambda -\mathbb{E}_t[X^*(s) + \Gamma_s^{(1)}], \] 
where $\left(\Gamma^{(1)}, \gamma^{(1)}\right)$ is the solution of the BSDE: 
\begin{equation*}
d\Gamma_s^{(1)} = -F(s)ds + \gamma^{(1)}(s)^\top d\boldsymbol{W}(s), ~ \Gamma^{(1)}(T_0) = 0
\end{equation*}
with $\gamma^{(1)} = (\gamma^{(1)}_1, \gamma^{(1)}_2, \gamma^{(1)}_3)^\top$. 
We suppress the dependence of $s$ when no confusion occurs. Mimicking the procedure in \cite{HJZ2012}, we deduce that
\begin{equation}\label{u*}
u^*(t) = e^{\int_{0}^{t}r(v)dv}\left(\sigma_S(t)^\top\sigma_S(t)\right)^{-1}\left(\lambda\nu(t) - \sigma_S(t)^\top\gamma^{(1)}(t)\right)
\end{equation}
and 
\[F(t) = \nu(t)^\top\left(\sigma_S(t)^\top\sigma_S(t)\right)^{-1}\left(\lambda \nu(t)  -\sigma_S(t)^\top\gamma^{(1)}(t)\right) + e^{-\int_{0}^{t}r(v)dv}\left( - \pi(t) -\mathbb{E}[z]c_1  {\hat{\mu}_2(t)}  \right). \] 
Then, $\Gamma^{(1)}$ satisfies the following BSDE:
\begin{align}\label{BSDE:gamma1}
\left\{
\begin{array}{ll}
d\Gamma^{(1)}_t & = -\Big\{\nu(t)^\top\left(\sigma_S(t)^\top\sigma_S(t)\right)^{-1}\left(\lambda \nu(t)  -\sigma_S(t)^\top\gamma^{(1)}(t)\right) + e^{-\int_{0}^{t}r(v)dv}\big(-\pi(t)\\
 & -\mathbb{E}[z]c_1  {\hat{\mu}_2(t)} \big)\Big\}dt + \gamma^{(1)}(t)^\top d\boldsymbol{W}(t), \\
\Gamma^{(1)}_{T_0} & = 0. 
\end{array}
\right.
\end{align}
We provide a uniqueness result for the general equilibrium control \eqref{u*} as follows. 
\begin{theorem}\label{Thm:unique}
Suppose that the following conditions are satisfied:
\begin{enumerate}[label=(\roman*)]
\item\label{condition1} $\sigma_S(t)\left(\sigma_S(t)^\top \sigma_S(t)\right)^{-1}\nu(t) \in H_\mathcal{F}^q(0, T_0; \mathbb{R}, \mathbb{P})$ for any $q > 1$, and
\begin{equation}
    \mathbb{E}\left[\exp\left(C\int_0^{T_0}\left|\sigma_S(s)\left(\sigma_S(s)^\top \sigma_S(s)\right)^{-1}\nu(s)\right|^2ds\right)\right] < \infty
\end{equation}
for a sufficiently large constant $C$.
\item The BSDE \eqref{BSDE:gamma1} admits a unique solution $(\Gamma^{(1)}_t, \gamma^{(1)}(t)) \in S_\mathcal{F}^q(0, T_0; \mathbb{R}, \mathbb{P})\times H_\mathcal{F}^q(0, T_0; \mathbb{R}^3, \mathbb{P})$ for a constant $q \geq 4$, and the $u^*$ in \eqref{u*} satisfies the condition in Theorem \ref{Thm:admissible} such that it is admissible. 
\end{enumerate}
Then, the equilibrium control $u^*$ given in \eqref{u*} is unique. 
\end{theorem}
The proof of Theorem \ref{Thm:unique} is shown in Appendix \ref{Appendix:unique}, where we clarify that $C > 6$ is considered to be sufficiently large in Condition \ref{condition1}.
To derive the closed-form solution for the equilibrium investment and hedging strategies, we want to solve for the BSDE \eqref{BSDE:gamma1} explicitly.


\subsection{ Deterministic market price of risk}
In this part, we assume that the market prices of risks are deterministic functions of time, which is a common assumption in financial modeling. 
\begin{assumption}\label{assume1}
The magnitudes of the market prices of the interest rate and mortality risk are both continuous deterministic functions. Specifically, $\varphi(s)= (\varphi_1(s), 0)^\top \in \mathbb{R}^2$ and $\vartheta(s)\in \mathbb{R}$, where $\varphi_1$ and $\vartheta$ are bounded deterministic functions.
\end{assumption}
It directly follows from Assumption \ref{assume1} that Novikov's condition \eqref{Novikov} is satisfied.
\subsubsection{Vasicek setting}
In the following, we adopt constant volatilities. Under Assumption \ref{assume_Vasicek}, the interest rate follows the classic Vasicek model and  $\mu_1$ in \eqref{mu1} becomes a Volterra type of Vasicek model.
\begin{assumption}\label{assume_Vasicek}
The volatility matrix of the   {force of mortality} 
$$\sigma(\mu(s)) \equiv \left(\begin{array}{cc}
\sigma_1 & 0\\
0  & \sigma_2
\end{array} \right) \triangleq \sigma_\mu,$$ 
where $\sigma_1, \text{and } \sigma_2$ are two positive constants. The volatility of the interest rate $\sigma(r(s))\equiv \sigma_r$ where $\sigma_r$ is a positive constant.
\end{assumption}
\begin{pro}\label{pro:Vas}
Under Assumption \ref{assume1} and \ref{assume_Vasicek}, $\sigma_l(s)$ in \eqref{longevity} is given by
 $$\sigma_l(s) = (\psi_1(T-s)\sigma_1, 0)^\top$$ 
 where $\psi_1$ solves $\psi_1 = (-1-\psi_1\theta_1)*K_1$; and $\sigma_b(s) = d_1(s)\sigma_r$ where $d_1(s) = \frac{1}{\theta_r}\left(e^{-\theta_r(T-s)} -1\right)$. 
In addition, $\pi(t) \in L_\mathcal{F}^q(0, T_0; \mathbb{R}, \mathbb{P})$ for any $q > 1$.
\end{pro}

To solve the BSDE \eqref{BSDE:gamma1} explicitly, we introduce the pricing measure $\bar{\mathbb{P}}$:
\begin{equation}\label{P_bar}
\frac{d\bar{\mathbb{P}}}{d\mathbb{P}} =  \exp{\left( \int_0^t -(\varphi_1(s), 0, \vartheta(s))d\boldsymbol{W}(s)- \frac{1}{2}(\varphi_1(s)^2 + \vartheta(s)^2)ds \right)}.
\end{equation}
Under the measure $\bar{\mathbb{P}}$, $\bar{\boldsymbol{W}}(t) \triangleq \boldsymbol{W}(t) + \int_0^t (\varphi_1(s), 0, \vartheta(s))^\top ds$ is a standard Brownian motion. We use $\bar{\mathbb{E}}[\cdot]$ to denote the expectation under the measure $\bar{\mathbb{P}}$. 
\begin{theorem}\label{thm:gam_Vas}
Under Assumptions \ref{assume1} and \ref{assume_Vasicek}, the BSDE \eqref{BSDE:gamma1} admits a unique solution $(\Gamma^{(1)}_t, \gamma^{(1)}(t)) \in S_\mathcal{F}^q(0, T_0; \mathbb{R}, \mathbb{P})\times H_\mathcal{F}^q(0, T_0; \mathbb{R}^3, \mathbb{P})$ for any $q > 1$, which is given by
\begin{align*}
\Gamma^{(1)}_t = \int_t^{T_0} \lambda\left( \varphi_1(s)^2 + \vartheta(s)^2 \right) - e^{-\int_{0}^{t}r(v)dv}e^{\bar{d}_0(s, t) + \bar{d}_1(s, t)r(t)}\Bigg( c_2\exp\left(  {\int_0^s - m_2(v) dv + }  Y^1_t(s)\right) \\
+ \mathbb{E}[z]c_1\bar{\mathbb{E}}_t[  {\hat{\mu}_2(s)} ]\Bigg)ds,
\end{align*}
and 
\begin{align*}
  \left(\begin{array}{c}
    \gamma^{(1)}_1(t) \\
    \gamma^{(1)}_2(t)
  \end{array} \right)  = -\int_t^{T_0} e^{-\int_{0}^{t}r(v)dv}e^{\bar{d}_0(s, t) + \bar{d}_1(s, t)r(t)}\Bigg( c_2\exp\left(   {\int_0^s - m_2(v) dv + }  Y^1_t(s)\right)\sigma_\mu\psi(s-t)^\top\\
  + \mathbb{E}[z]c_1\sigma_\mu \left(\begin{array}{c}
     E_\Theta^{21}(s-t) \\
     E_\Theta^{22}(s-t)
  \end{array} \right)\Bigg)ds,
\end{align*}
\begin{align*}
\gamma_3^{(1)} = -\int_t^{T_0}e^{-\int_{0}^{t}r(v)dv}e^{\bar{d}_0(s, t) + \bar{d}_1(s, t)r(t)}\bar{d}_1(s, t)\sigma_r \left( c_2\exp\left(  {\int_0^s- m_2(v) dv + } 
 Y^1_t(s)\right) + \mathbb{E}[z]c_1\bar{\mathbb{E}}_t[  {\hat{\mu}_2(s)} ]\right)ds,
\end{align*}
where $\bar{d}_1(s, t) = \frac{1}{\theta_r}\left(e^{-\theta_r(s -t)}-1 \right)$; $\bar{d}_0(s, t) = \int_t^s(b_r -\sigma_r\vartheta(v)) + \frac{1}{2}\sigma_r^2\bar{d}_1(v, t)^2dv$; $Y^1 _t(s)$ is equivalent to $Y_t(s)$, as defined in \eqref{Y} with $\sigma(\mu(s)) = \sigma_\mu$, $f = (0, -1)$, $W(s)$ replaced by $\bar{W}(s)$, and $\psi$ solving the Volterra--Riccati equation $\psi = ((0, -1) - \psi \Theta)*K$;   {$\bar{\mathbb{E}}_t[\hat{\mu}(s)] = (m_1(s), m_2(s))^\top + \bar{\mathbb{E}}_t[\mu(s)] $ where} $\bar{\mathbb{E}}_t[\mu(s)]$ is given in \eqref{expectation:mu} with $b(s) = (b_1 - \sigma_1\varphi_1(s), b_2)$, $\sigma(\mu(s)) = \sigma_\mu$ and $W(s)$ replaced by $\bar{W}(s)$; and $E_\Theta = \left(\begin{array}{cc}
    E_\Theta^{11} & E_\Theta^{12}  \\
    E_\Theta^{21} & E_\Theta^{22}
\end{array}\right)$ in \eqref{expectation:mu}.
\end{theorem}
\begin{proof}
By calculation,
\begin{align*}
& \nu(t)^\top\left(\sigma_S(t)^\top\sigma_S(t)\right)^{-1}\nu(t) \\
& = (\varphi_1(t)\sigma_{l1}(t)+ \vartheta(t)\sigma_b(t), \vartheta(t)\sigma_b(t))\left(\begin{array}{cc}
   \sigma_b^2 +\sigma_{l1}^2  & \sigma_b^2 \\
   \sigma_b^2  & \sigma_b^2
\end{array} \right)^{-1}\left(\begin{array}{c}
     \varphi_1(t)\sigma_{l1}(t)+ \vartheta(t)\sigma_b(t) \\
      \vartheta(t)\sigma_b(t)
\end{array}\right)\\
& = \varphi_1(t)^2 + \vartheta(t)^2,
\end{align*}
and 
\begin{align*}
\nu(t)^\top\left(\sigma_S(t)^\top\sigma_S(t)\right)^{-1}\sigma_S(t)^\top =  (\varphi_1(t), 0 , \vartheta(t)) \triangleq \xi(t). 
\end{align*} 
Under the measure $\bar{\mathbb{P}}$, the BSDE \eqref{BSDE:gamma1} becomes
\begin{align*}
\left\{
\begin{array}{ll}
d\Gamma^{(1)}_t & = -\left\{\lambda\left( \varphi_1(t)^2 + \vartheta(t)^2 \right) + e^{-\int_{0}^{t}r(s)ds}\left( - \pi(t) -\mathbb{E}[z]c_1  {\hat{\mu}_2(t)} \right)\right\}dt + \gamma^{(1)}(t)^\top d\bar{\boldsymbol{W}}(t), \\
\Gamma^{(1)}_{T_0} & = 0.
\end{array}
\right.
\end{align*}
By Lemmas \ref{lemma:expmu} and \ref{lemma:E_mu}, and Proposition \ref{pro:Vas}, we have
\[e^{-\int_{0}^{t}r(s)ds}\left( - \pi(t) -\mathbb{E}[z]c_1  {\hat{\mu}_2(t)} \right) \in H_\mathcal{F}^{q}(0,T_0;\mathbb{R}, \bar{\mathbb{P}}). \]
for any $q > 1$. By Theorem 5.1 in \cite{EPQ}, there exists a unique solution $ (\Gamma^{(1)}, \gamma^{(1)}) \in S_\mathcal{F}^{q}(0,T_0;\mathbb{R}, \bar{\mathbb{P}}) \times H_\mathcal{F}^{q}(0,T_0;\mathbb{R}^3, \bar{\mathbb{P}})$ for any $q > 1$. 
We denote the $\frac{d\bar{\mathbb{P}}}{d\mathbb{P}}$ in \eqref{P_bar} by $\mathcal{E}_t(-\xi)$. For any constant $q > 1$, there exists a constant $C_1$ such that
\begin{align*}
\mathbb{E}\left[\sup_{0\leq t\leq T_0}|\Gamma_t^{(1)}|^q \right] & = \bar{\mathbb{E}}\left[ \sup_{0\leq t\leq T_0}|\Gamma_t^{(1)}|^q \mathcal{E}_{T_0}^{-1}(\xi)\right]\leq \left\{\bar{\mathbb{E}}\left[\sup_{0\leq t\leq T_0}|\Gamma_t^{(1)}|^{2q} \right]\mathbb{E}[\mathcal{E}_{T_0}^{-1}(\xi)] \right\}^{\frac{1}{2}}\\
& \leq C_1 \left\{\mathbb{E}\left[\mathcal{E}_{T_0}(2\xi) \right]\mathbb{E}\left[\exp\left(3 \int_0^{T_0}|\xi(s)|^2ds\right) \right] \right\}^{\frac{1}{4}} < \infty. 
\end{align*}
Therefore, $\Gamma^{(1)} \in S_\mathcal{F}^{q}(0,T_0;\mathbb{R}, \mathbb{P})$ for any $q > 1$. Similarly, we have $\gamma^{(1)} \in H_\mathcal{F}^{q}(0,T_0;\mathbb{R}^3, \mathbb{P})$ for any $q > 1$. 
Furthermore, 
\begin{align*}
\Gamma^{(1)}_t & = \int_t^{T_0} \bar{\mathbb{E}}_t\left[\lambda\left( \varphi_1(s)^2 + \vartheta(s)^2 \right) -  e^{-\int_{0}^{s}r(v)dv}\left( \pi(s)  + \mathbb{E}[z]c_1  {\hat{\mu}_2(s)}\right) \right]ds \\
& = \int_t^{T_0} \lambda\left( \varphi_1(s)^2 + \vartheta(s)^2 \right) - \bar{\mathbb{E}}_t\left[e^{-\int_{0}^{s}r(v)dv} \right]\left( c_2\bar{\mathbb{E}}_t\left[e^{-\int_{0}^{s}  {\hat{\mu}_2(v)}dv} \right] + \mathbb{E}[z]c_1\bar{\mathbb{E}}_t[  {\hat{\mu}_2(s)}]\right)ds.
\end{align*}
Under the measure $\bar{\mathbb{P}}$,   $\mu$ follows
\begin{align*}
\mu(t)= \mu(0) + \int_0^t K(t-s)(b(s) -\Theta\mu(s))ds + \int_0^t K(t-s)\sigma_\mu d\bar{W}(s),
\end{align*}
where $b(s) = (b_1 - \sigma_1\varphi_1(s), b_2)$, $K = \left( \begin{array}{cc}
K_1 & 0\\
\beta_1 K_1 & K_2
\end{array} \right)$, and $\Theta = \left(\begin{array}{cc}
\theta_1 & 0\\
\beta_2 & \theta_2
\end{array}\right)$, which maintains the affine structure. By Lemma \ref{lemma:expmu}, we have
\[\bar{\mathbb{E}}_t\left[e^{-\int_{0}^{s}\mu_2(v)dv}\right] = \exp(  {\int_0^s - m_2(v) dv + } Y^1_t(s)).\]
The interest rate under the measure $\bar{\mathbb{P}}$ follows the dynamic:
\[ dr(t) = (b_r - \sigma_r\vartheta(t) - \theta_r r(t))dt + \sigma_rd\bar{W}_r(t).\]
By Appendix \ref{appendix:affine}, $ \bar{\mathbb{E}}_t\left[e^{-\int_{0}^{s}r(v)dv} \right] = e^{-\int_{0}^{t}r(v)dv}e^{\bar{d}_0(s, t) + \bar{d}_1(s, t)r(t)}$. The expression for $(\Gamma^{(1)}, \gamma^{(1)})$ follows.
\end{proof}

\begin{cor}\label{cor:Vas}
  Under Assumption \ref{assume1} and \ref{assume_Vasicek},  the unique admissible equilibrium investment and hedging strategies $u^* = (u^*_1, u^*_2)^\top$ is given by 
  \begin{align*}
   u^*_1(t) & =   e^{\int_0^tr(s)ds}\left(  \frac{\lambda\varphi_1(t)}{\psi_1(T-t)\sigma_1} - \frac{\gamma^{(1)}_1(t)}{\psi_1(T-t)\sigma_1}\right) \label{u1*Vas} \\
   u^*_2(t) & = e^{\int_0^tr(s)ds}\left( - \frac{\lambda\varphi_1(t)}{\psi_1(T-t)\sigma_1} + \frac{\lambda\vartheta(t)}{d_1(t)\sigma_r}+\frac{1}{\psi_1(T-t)\sigma_1}\gamma^{(1)}_1(t) - \frac{1}{d_1(t)\sigma_r}\gamma^{(1)}_3(t)\right)
  \end{align*}
where $\psi_1$ solves $\psi_1 = (-1-\psi_1\theta_1)*K_1$, $d_1(t) = \frac{1}{\theta_r}\left(e^{-\theta_r(T-t)} -1\right)$, and  $\gamma^{(1)}_1$ and $\gamma^{(1)}_3$ are given in Theorem \ref{thm:gam_Vas}.
\end{cor}
\begin{proof}
    As $\sigma_S(t)\left(\sigma_S(t)^\top \sigma_S(t)\right)^{-1}\nu(t)= (\varphi_1(t), 0 , \vartheta(t))^\top$ is 
a bounded deterministic function in this case, condition \ref{condition1} in Theorem \ref{Thm:unique} is satisfied. By calculation, the expression for $u^*$ folows. 
By Lemma \ref{Lemma:interest} and H\"older's inequality, it holds that $u^* \in H_\mathcal{F}^2(0, T_0; \mathbb{R}^2, \mathbb{P})$. Similarly, we have $\nu^\top u^*\in L_\mathcal{F}^4(0, T_0; \mathbb{R}, \mathbb{P})$ and $\sigma_S u^* \in L_\mathcal{F}^4(0, T_0; \mathbb{R}^3, \mathbb{P})$, then $u^*$ is a unique admissible control. 
\end{proof}

 As our problem has LRD property, the equilibrium hedging strategy $u^*$ in Corollary \ref{cor:Vas} naturally depends on the entire historical life tables of the nation. Most of the functions there are time deterministic functions and depened on the model parameters, except for $\gamma_1^{(1)}$ and $\gamma_3^{(1)}$. The $\gamma_1^{(1)}$ and $\gamma_3^{(1)}$ defined in Theorem \ref{thm:gam_Vas} appear as functionals of $Y^1_t(\cdot)$, which is a collection of the historical paths of $\mu_1$ and depends on the state process $X_t$. In other words, the equilibrium hedging strategy is not a function on the state process $X$ alone but the historical path of the national death data. As the open-loop control framework embraces feedback and closed-loop controls as special cases, this funding justifies the suboptimality of them for our application.

In Corollary \ref{cor:Vas},  the closed-form solution of the unique admissible strategy is provided. We further investigate the effects of incorporation LRD property and cointegration.  The effect of cointegration is mainly reflected in the term $c_2\exp\left(  {\int_0^s m_2(v) dv + } Y^1_t(s)\right)\sigma_\mu\psi(s-t)^\top
  + \mathbb{E}[z]c_1\sigma_\mu \left(\begin{array}{c}
     E_\Theta^{21}(s-t) \\
     E_\Theta^{22}(s-t)
  \end{array} \right)$ in the expression of $\left(\begin{array}{c}
    \gamma^{(1)}_1(t) \\
    \gamma^{(1)}_2(t)
  \end{array} \right)$. Indeed, the function $\psi$ admits an explicit solution
  $\psi = (0 , -1)*E_\Theta$. In the following, we provide the explicit form for $E_\Theta$.
  
\begin{pro}\label{pro:RB}
$E_\theta$ admits the form
\begin{align}
    E_\Theta = \left(\begin{array}{cc} 
    \frac{1}{\theta_1}R_{1\theta_1} & 0 \\
    \frac{\beta_1}{\theta_1}R_{1\theta_1} -(\frac{\beta_1}{\theta_1} + \frac{\beta_2}{\theta_1\theta_2})R_{2\theta_2}*R_{1\theta_1} & \frac{1}{\theta_2}R_{2\theta_2}
\end{array} \right),
\end{align}
where $R_{1\theta_1}$ and $R_{2\theta_2}$ are the resolvents for $\theta_1 K_1$ and $\theta_2 K_2$, respectively.
\end{pro} 

By Proposition \ref{pro:RB}, we can alternatively express $\gamma_1^{(1)}$ as follows. 
\begin{align*}
  \gamma_1^{(1)}(t) =  \sigma_1 e^{-\int_{0}^{t}r(v)dv} \int_t^{T_0} e^{\bar{d}_0(s, t) + \bar{d}_1(s, t)r(t)}\Bigg(c_2\exp\left(  {\int_0^s m_2(v) dv + }  Y^1_t(s)\right)\int_0^{s-t}E_\Theta^{21}(u)du \\
  -\mathbb{E}[z]c_1E_\Theta^{21}(s-t)\Bigg)ds,
\end{align*}
where $E_\Theta^{21} =  \frac{\beta_1}{\theta_1}R_{1\theta_1} -(\frac{\beta_1}{\theta_1} + \frac{\beta_2}{\theta_1\theta_2})R_{2\theta_2}*R_{1\theta_1}$. It is then clear that the cointegration and correlation affect the equilibrium hedging  strategy through the values of $\beta_1$ and $\beta_2$ in $E_\Theta^{21}$. In our model setting, the coefficient $\beta_1$ is associated with both cointegration and correlation between $\mu_1$ and $\mu_2$, and $\beta_2$ only contributes to the cointegration. When $\mu_1$ and $\mu_2$ are independent, we have $\beta_1 =\beta_2 = 0$ and $E_\Theta^{21} = 0$ so that $\gamma_1^{(1)} = 0$. The equilibrium hedging strategy $u_1^*$ in 
\eqref{u1*Vas}  reduces to  
\[ u^*_1(t) =   e^{\int_0^tr(s)ds} \frac{\lambda\varphi_1(t)}{\psi_1(T-t)\sigma_1},\]
where $\psi_1$ is defined in Corollary \ref{cor:Vas}.

By setting $H_1 = H_2 = \frac{1}{2}$, we have $K_1 = K_2 \equiv 1$ and recover the equilibrium hedging strategy under the Markovian model. Hence, the effect of LRD property on the equilibrium strategy is reflected by the terms $\psi_1$ and $\gamma_1^{(1)}$ in \eqref{u1*Vas}. Note that $\psi_1 = \frac{1}{\theta_1}(e^{-\theta_1 t} -1)$ under the Markovian case. Therefore, the LRD property primarily adjusts $E_\Theta$. In fact, the term $\exp(Y_t^1)$ is influenced by the related $\psi$ through $E_\Theta$. Under the Markovain model,
 \[E_\Theta^{21}(t) = \beta_1 e^{-\theta_1 t} - \frac{\beta_1 \theta_2 + \beta_2}{\theta_2 -\theta_1}\left(e^{-\theta_1 t} - e^{-\theta_2 t}\right). \]
The direct comparison shows that the LRD adjustment on $E_\Theta^{21}$ reads
\begin{align*}
   \frac{\beta_1}{\theta_1}R_{1\theta_1} -(\frac{\beta_1}{\theta_1} + \frac{\beta_2}{\theta_1\theta_2})R_{2\theta_2}*R_{1\theta_1} - \beta_1 e^{-\theta_1 t} + \frac{\beta_1 \theta_2 + \beta_2}{\theta_2 -\theta_1}\left(e^{-\theta_1 t} - e^{-\theta_2 t}\right). 
\end{align*}

\subsection{Volatility-driven market price of risk}
Another consideration is that the market prices of risks are related to their volatility. In this part of the paper, we investigate the hedging strategy under the following assumption.
\begin{assumption}\label{assume2}
The market prices of mortality and interest risks are related to their volatility. Specifically, $\varphi(s) = \sigma(\mu(s))(\widetilde{\varphi}_1, 0)^\top$ and $\vartheta(s) = \widetilde{\vartheta} \sigma_r(r(s))$, where $\widetilde{\varphi}_1$ and $\widetilde{\vartheta}$ are constant values.
\end{assumption} 
\subsubsection{CIR setting}
\begin{assumption}\label{assume_CIR}
The volatility matrix of the   {force of mortality}  $$\sigma(\mu(s)) = \left(\begin{array}{cc}
\widetilde{\sigma}_1\sqrt{\mu_1(s)} & 0\\
0  & \widetilde{\sigma}_2\sqrt{\mu_2(s)}
\end{array} \right)$$ with constant values $\widetilde{\sigma}_1$ and $\widetilde{\sigma}_2$. The volatility of the interest rate $\sigma_r(r(s)) = \widetilde{\sigma}_r\sqrt{r(s)}$ for a positive constant $\widetilde{\sigma}_r$. 
\end{assumption}
Further, under the CIR setting, we make the following assumption.
\begin{assumption}\label{assume:expmu}
\begin{equation}\label{condition:expmu1}
\mathbb{E}\left[\exp\left(C\int_0^{T_0}\mu_1(s)ds\right)\right] < \infty    
\end{equation}
for a sufficiently large constant $C$, and $\sup_{0\leq t \leq T_0}\mathbb{E}\left[e^{-16\int_0^t\mu_2(s)ds}\right] < \infty$.
\end{assumption}
Under Assumption \ref{assume:expmu}, Novikov's condition in \eqref{Novikov} is satisfied and $\pi \in L_\mathcal{F}^{16}(0, T_0; \mathbb{R}, \mathbb{P})$.
Under the CIR setting in Assumption \ref{assume_CIR},     {$\mu$}  in \eqref{mu0} becomes
\begin{align}
\mu_1(t) & = \mu_1(0) + \int_{0}^{t}K_1(t-s)(b_1 - \theta_1 \mu_1(s))ds + \int_{0}^{t}K_1(t-s)\widetilde{\sigma}_1\sqrt{\mu_1(s)}dW_1(s),\label{mu1CIR}\\
\mu_2(t) & = \mu_2(0) + \int_{0}^{t}\beta_1 K_1(t-s)(b_1 - \theta_1 \mu_1(s))ds  + \int_0^tK_2(t-s)(b_2 - \beta_2 \mu_1(s) - \theta_2 \mu_2(s))ds \nonumber \\
& + \int_{0}^{t}\beta_1 K_1(t-s)\widetilde{\sigma}_1\sqrt{\mu_1(s)}dW_1(s) + \int_0^t K_2(t-s)\widetilde{\sigma}_2\sqrt{\mu_2(s)} dW_2(s); \label{mu2CIR}
\end{align}
The interest rate follows 
\[	dr(t)=(b_r -  \theta_r r(t))dt +\widetilde{\sigma}_r \sqrt{r(s)} dW_r(t).\]
\begin{pro}\label{Pro:CIR}
Under Assumptions \ref{assume2}, \ref{assume_CIR}, and \ref{assume:expmu}, we have
$$\sigma_l(s)  =  \left(\psi'_1(T-s)\widetilde{\sigma}_1\sqrt{\mu_1(s)},~ 0 \right)^\top. $$
where $\psi'_1$ solves the Volterra--Riccati equation $\psi'_1 = (-1 - (\theta_1 - \widetilde{\varphi}_1\widetilde{\sigma}_1^2)\psi'_1+ \frac{1}{2}\widetilde{\sigma}_1^2(\psi'_1)^2)*K_1$.
In addition, $\sigma_b(s) = d'_1(s)\widetilde{\sigma}_r \sqrt{r(s)}$ with $d'_1( s) = \frac{1}{\theta_r-\widetilde{\vartheta}\widetilde{\sigma}_r^2}\left(e^{-(\theta_r-\widetilde{\vartheta}\widetilde{\sigma}_r^2)(T-s)} -1\right)$.
\end{pro}

To present the closed-form solution to BSDE \eqref{BSDE:gamma1} explicitly, we introduce the pricing measure $\widetilde{\mathbb{P}}$ as follows.
 $$\frac{d\widetilde{\mathbb{P}}}{d\mathbb{P}} =  \exp{\left( \int_0^t -(\widetilde{\varphi}_1\widetilde{\sigma}_1\sqrt{\mu_1(s)}, 0, \widetilde{\vartheta}\widetilde{\sigma}_r\sqrt{r(s)})d\boldsymbol{W}(s)- \frac{1}{2}(\widetilde{\varphi}_1^2\widetilde{\sigma}_1^2\mu_1(s) + {\widetilde{\vartheta}}^2\widetilde{\sigma}_r^2 r(s))ds \right)}.$$
 $\widetilde{\boldsymbol{W}}(t) \triangleq \boldsymbol{W}(t) + \int_0^t (\widetilde{\varphi}_1\widetilde{\sigma}_1\sqrt{\mu_1(s)}, 0, \widetilde{\vartheta}\widetilde{\sigma}_r\sqrt{r(s)})^\top ds$ is a standard Brownian motion under measure $\widetilde{\mathbb{P}}$. Let $\widetilde{\mathbb{E}}[\cdot]$ denote the expectation under the measure $\widetilde{\mathbb{P}}$. 
\begin{theorem}\label{solutionCIR}
Under Assumptions \ref{assume2}, \ref{assume_CIR}, and \ref{assume:expmu},  the BSDE \eqref{BSDE:gamma1} admits a unique solution $(\Gamma^{(1)}_t, \gamma^{(1)}(t)) \in S_\mathcal{F}^4(0, T_0; \mathbb{R}, \mathbb{P})\times H_\mathcal{F}^4(0, T_0; \mathbb{R}^3, \mathbb{P})$ given by
\begin{align*}
\Gamma^{(1)}_t = \int_t^{T_0} \lambda \left(\widetilde{\varphi}_1^2\widetilde{\sigma}_1^2\widetilde{\mathbb{E}}_t[  {\hat{\mu}_1(s)}] + \widetilde{\vartheta}^2\widetilde{\sigma}_r^2 \widetilde{\mathbb{E}}_t[r(s)]\right) - e^{-\int_{0}^{t}r(v)dv}e^{\widetilde{d}_0(s, t) + \widetilde{d}_1(s, t)r(t)}\\
\bigg( c_2\exp\left(  {\int_0^s - m_2(v) dv + }  Y^2_t(s)\right)  + \mathbb{E}[z]c_1\widetilde{\mathbb{E}}_t[  {\hat{\mu}_2(s)}]\bigg)ds,
\end{align*}
and 
\begin{align*}
  \left(\begin{array}{c}
    \gamma^{(1)}_1(t) \\
    \gamma^{(1)}_2(t)
  \end{array} \right) = \int_t^{T_0}\lambda\widetilde{\varphi}_1^2\widetilde{\sigma}_1^2\left(\begin{array}{c}
      E_{B_1}(s-t)\widetilde{\sigma}_1\sqrt{\mu_1(t)}   \\
      0 
  \end{array}\right) -e^{-\int_{0}^{t}r(v)dv}e^{\widetilde{d}_0(s, t) + \widetilde{d}_1(s, t)r(t)}\\
  \Bigg( c_2\exp\left(Y^2_t(s)\right)\sigma(\mu(t))\widetilde{\psi}(s-t)^\top
  + \mathbb{E}[z]c_1\left(\begin{array}{c}
      \widetilde{\sigma}_1\sqrt{\mu_1(s)}E_{\widetilde{\Theta}}^{21}(s-t) \\
      \widetilde{\sigma}_2\sqrt{\mu_2(s)}E_{\widetilde{\Theta}}^{22}(s-t)
  \end{array} \right)\Bigg)ds,
\end{align*}
\begin{align*}
\gamma_3^{(1)}(t) = \int_t^{T_0}\lambda \widetilde{\vartheta}^2\widetilde{\sigma}_r^2 e^{-\theta_r(s-t)}\widetilde{\sigma}_r\sqrt{r(t)} -e^{-\int_{0}^{t}r(v)dv}e^{\widetilde{d}_0(s, t) + \widetilde{d}_1(s, t)r(t)}\widetilde{d}_1(s, t)\widetilde{\sigma}_r\sqrt{r(t)}\\
\Bigg(c_2\exp\left(  {\int_0^s - m_2(v) dv + } 
 Y^2_t(s)\right)
+ \mathbb{E}[z]c_1\widetilde{\mathbb{E}}_t[  {\hat{\mu}_2(s)}]\Bigg)ds,
\end{align*}
where $B_1 = -\theta_1$;  $Y^2_t(s)$ is equivalent to $Y_t(s)$ in \eqref{Y} and  $\widetilde{\mathbb{E}}_t[\hat{\mu}_2(s)] = m_2(s) + \widetilde{\mathbb{E}}_t[\mu_2(s)]$ where    $\widetilde{\mathbb{E}}_t[\mu_2(s)]$ is as defined in \eqref{expectation:mu}, with $f = (0, -1)$, $b(s) = (b_1, b_2)$, $\Theta$ replaced by $\widetilde{\Theta} \triangleq \left(\begin{array}{cc}
\theta_1 + \widetilde{\varphi}_1\widetilde{\sigma}_1^2 & 0\\
\beta_2 & \theta_2
\end{array}\right)$, $W(s)$ replaced by $\widetilde{W}(s)$, and $\psi$ replaced by $\widetilde{\psi}$, which solves the Volterra--Riccati equation $\widetilde{\psi} = ((0, -1) - \widetilde{\psi}\widetilde{\Theta} + \frac{1}{2}(\widetilde{\sigma}_1^2\widetilde{\psi}_1^2, \widetilde{\sigma}_2^2\widetilde{\psi}_2^2))*K$; $E_\Theta$ replaced by $E_{\widetilde{\Theta}} \triangleq \left(\begin{array}{cc}
    E_{\widetilde{\Theta}}^{11} & E_{\widetilde{\Theta}}^{12}  \\
    E_{\widetilde{\Theta}}^{21} & E_{\widetilde{\Theta}}^{22}
\end{array}\right)$ in \eqref{expectation:mu}.
\end{theorem}
\begin{proof}
The proof is similar to the proof of Theorem \ref{thm:gam_Vas} and is thus omitted.
\end{proof}

\begin{cor}
    Under Assumption \ref{assume2}, \ref{assume_CIR}, and \ref{assume:expmu},  the unique admissible equilibrium investment and hedging strategies $u^*$ is given by
    \begin{align}
       u^*_1(t) & =  e^{\int_0^tr(s)ds}\left(\frac{\lambda\widetilde{\varphi}_1}{\psi'_1(T-t)} - \frac{\gamma^{(1)}_1(t)}{\psi'_1(T-t)\widetilde{\sigma}_1\sqrt{\mu_1(t)}} \right), \\
       u^*_2(t) & =  e^{\int_0^tr(s)ds}\left(  - \frac{\lambda \widetilde{\varphi}_1}{\psi'_1(T-t)} + \frac{\lambda\widetilde{\vartheta}}{d'_1( t)} + \frac{\gamma^{(1)}_1(t)}{\psi'_1(T-t)\widetilde{\sigma}_1\sqrt{\mu_1(t)}} - \frac{\gamma^{(1)}_3(t)}{d'_1(t)\widetilde{\sigma}_r \sqrt{r(t )}}\right),
    \end{align}
    where  $\psi'_1$ solves $\psi'_1 = (-1 - (\theta_1 - \widetilde{\varphi}_1\widetilde{\sigma}_1^2)\psi'_1+ \frac{1}{2}\widetilde{\sigma}_1^2(\psi'_1)^2)*K_1$, $d'_1( s) = \frac{1}{\theta_r-\widetilde{\vartheta}\widetilde{\sigma}_r^2}\left(e^{-(\theta_r-\widetilde{\vartheta}\widetilde{\sigma}_r^2)(T-s)} -1\right)$,  and $\gamma^{(1)}_1$ and $\gamma^{(1)}_3$ are given in Theorem \ref{solutionCIR}.
\end{cor}
Under the volatility-driven market price of risk and the CIR setting, Theorem \ref{solutionCIR} provides the unique equilibrium hedging strategy under the condition in Assumption \ref{assume:expmu}. From the proof, $C > \max\{\frac{1}{2}\widetilde{\varphi}_1\widetilde{\sigma}_1^2, 8\}$ in Assumption \ref{assume:expmu} is sufficient for Novikov's condition \eqref{Novikov}, the admissibility and uniqueness of the control. 
\section{Numerical study}
\label{Sec:numerical}
In this section, we numerically examine the effects of cointegration and the LRD property of the   {force of mortality}  on longevity hedging strategies. The deterministic market price of risk and the Vasicek setting in Assumptions \ref{assume1} and \ref{assume_Vasicek} are considered.  Then,   {the dynamic of $\mu$  is given by} 
\begin{align}\label{Vas_num}
\begin{split}
\mu_1(t) & = \mu_1(0) + \int_{0}^{t}\frac{(t-s)^{H_1 -\frac{1}{2}}}{\Gamma(H_1+\frac{1}{2})}(b_1 - \theta_1 \mu_1(s))ds + \int_{0}^{t}\frac{(t-s)^{H_1 -\frac{1}{2}}}{\Gamma(H_1+\frac{1}{2})}\sigma_1dW_1(s) \\ 
\mu_2(t) & = \mu_2(0) + \beta_1\int_{0}^{t} \frac{(t-s)^{H_1 -\frac{1}{2}}}{\Gamma(H_1+\frac{1}{2})}(b_1 - \theta_1 \mu_1(s))ds  + \int_0^t \frac{(t-s)^{H_2 -\frac{1}{2}}}{\Gamma(H_2+\frac{1}{2})}(b_2-\beta_2\mu_1(s) -\theta_2 \mu_2(s))ds \\
    & + \beta_1 \int_{0}^{t} \frac{(t-s)^{H_1 -\frac{1}{2}}}{\Gamma(H_1+\frac{1}{2})}\sigma_1dW_1(s) + \int_0^t \frac{(t-s)^{H_2 -\frac{1}{2}}}{\Gamma(H_2 +\frac{1}{2})}\sigma_2 dW_2(s).
\end{split}
\end{align}
The interest rate follows the dynamic:
\begin{equation}\label{interest_num}
dr(t) = (b_r - \theta_r r(t))dt + \sigma_rdW_r(t).
\end{equation}

\subsection{Effect of cointegration between   {forces of mortality}}
\label{Sec:num_coin}
We first examine the effect of cointegration under the LRD mortality environment in this part. Here, we focus on a group of the population with the same age, 65, at the hedging start point $t = 0$.   {To perform numerical experiments,  inspired by \cite{Biffis}, we assume that $m_i(t)$, $i =  1, 2$, is calibrated to the  SIM92 table, which is a dataset  from the Italian National Institute of Statistics (ISTAT) that reports Italian population life tables. The numbers of persons surviving to ages from 40 to 100 in SIM92 table are used for  calibration. We assume $m_i(x) = a_mx^{g-1} \text{ for } x\in [40, 100], i = 1, 2$ where $a_m$ and $g$ are two constants. This form is same to  the force of mortality under Weibull mortality model.  By calibration,  $m_i(t) = 0.000004212*(t+25)^{2.68}$, $i = 1, 2$ for $t\in [0, T]$ is obtained.
According to  the empirical studies in \cite{Yan2021}, the Hurst parameter for the national force of mortality is around $0.83$.  In this part, we choose $H_1 = 0.83$ and $H_2 = 0.5$ for the true model. }  
 Then,   $\mu_2$ in \eqref{Vas_num} becomes
\begin{align*}
\mu_2(t) & = \mu_2(0) + \beta_1\int_{0}^{t} \frac{(t-s)^{H_1 -\frac{1}{2}}}{\Gamma(H_1+\frac{1}{2})}(b_1 - \theta_1 \mu_1(s))ds  + \int_0^t (b_2-\beta_2\mu_1(s) -\theta_2 \mu_2(s))ds \\
    & + \beta_1 \int_{0}^{t} \frac{(t-s)^{H_1 -\frac{1}{2}}}{\Gamma(H_1+\frac{1}{2})}\sigma_1dW_1(s) + \int_0^t \sigma_2 dW_2(s)
\end{align*}
which is an RL mixed fractional Brownian motion.    {For the national force of mortality,  we set  $b_1 = 0.0001$, $\theta_1 = 0.5$, $\sigma_1 = 0.002$, and  $\mu_1(0) = 0.00008$. These parameter values  are of  the same orders of magnitude as those in \cite{Biffis} and \cite{WaCW}. In addition,  the parameter values for $\mu_2$ are $b_2 = 0.0001$, $\theta_2 = 0.6$, $\sigma_2 = 0.001$, $\beta_1 = 1$, $\beta_2 = 0.6$,  and $\mu_2(0) = 0.00008$. Consider the hedging  horizon $T_0 = 5$.}   A pair of sample paths for $\mu(\cdot)$ are simulated as shown in Figure \ref{fig:mu}.   {After adding $m_i(t)$  to  $\mu_i(t)$,  the forces of mortality $\hat{\mu}_i$ are presented in Figure \ref{fig:mu_hat}.}
	\begin{figure}[H]
			\centering
			\subfigure{\includegraphics[width= 7.5 cm ]{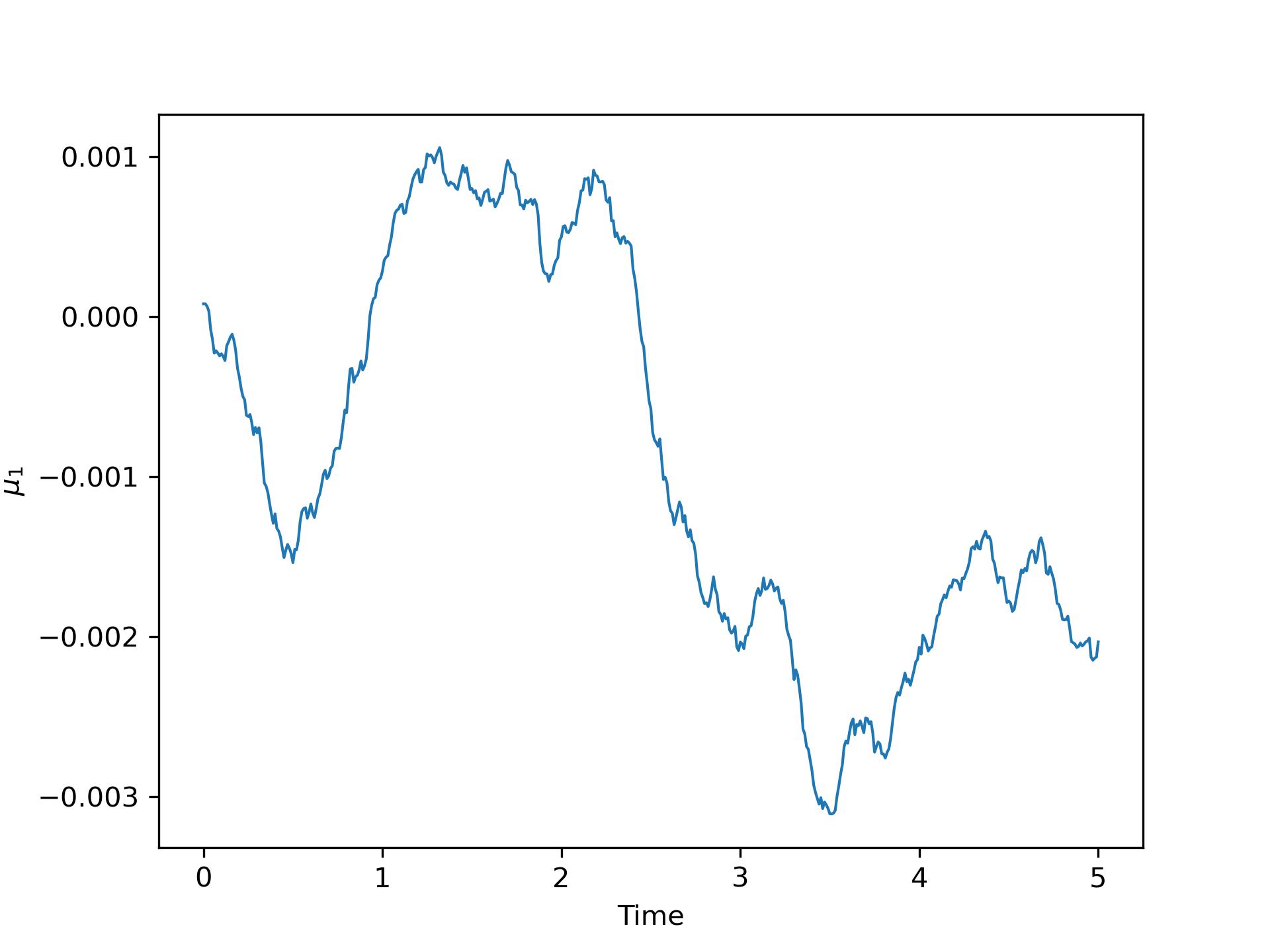}}
			\subfigure{\includegraphics[width= 7.5cm]{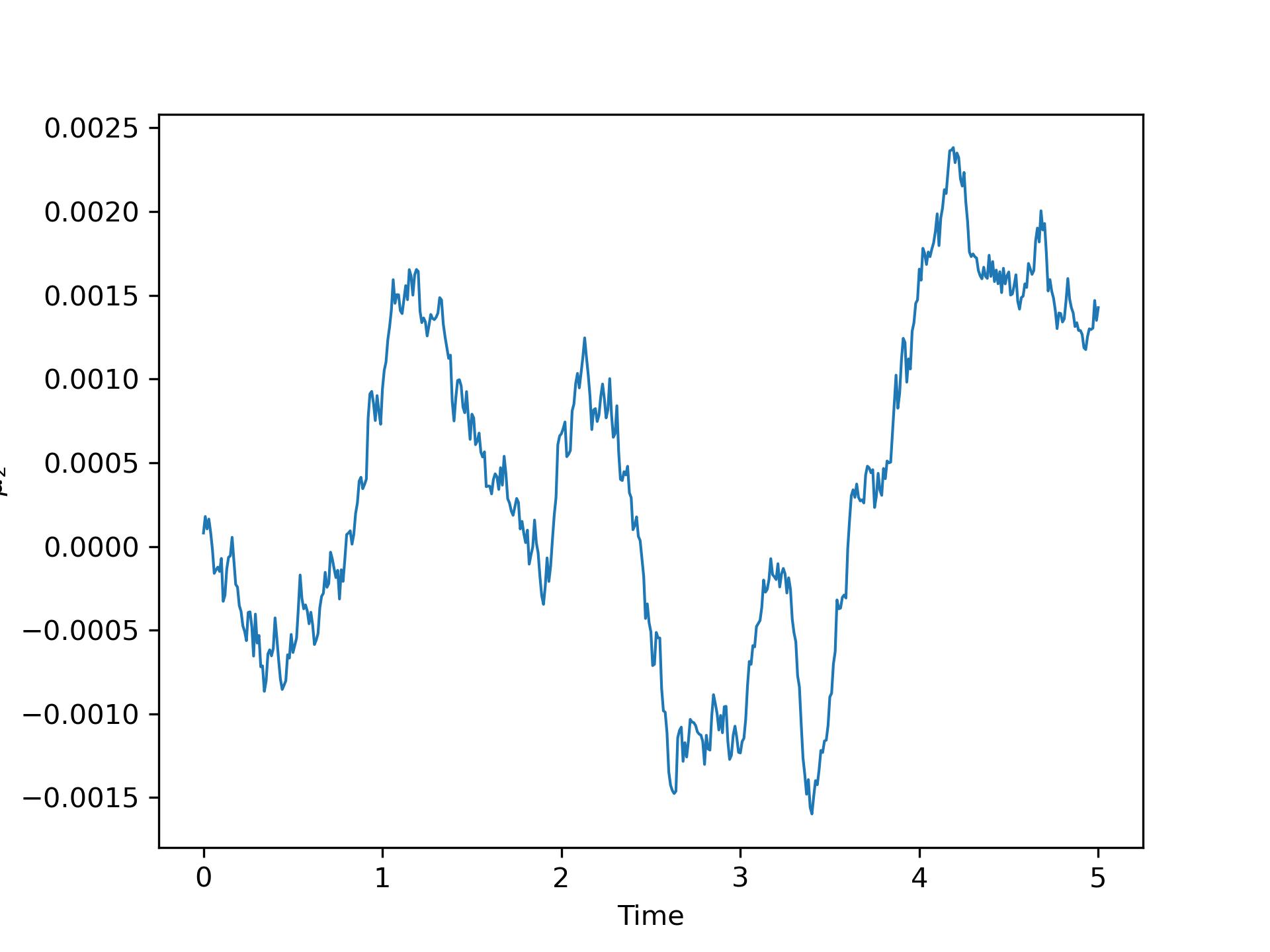}}
			\caption{Simulated sample paths of $\mu_1$ and $\mu_2$}
			\label{fig:mu}
		\end{figure}
	\begin{figure}[H]
			\centering
	\subfigure{\includegraphics[width= 7.5 cm ]{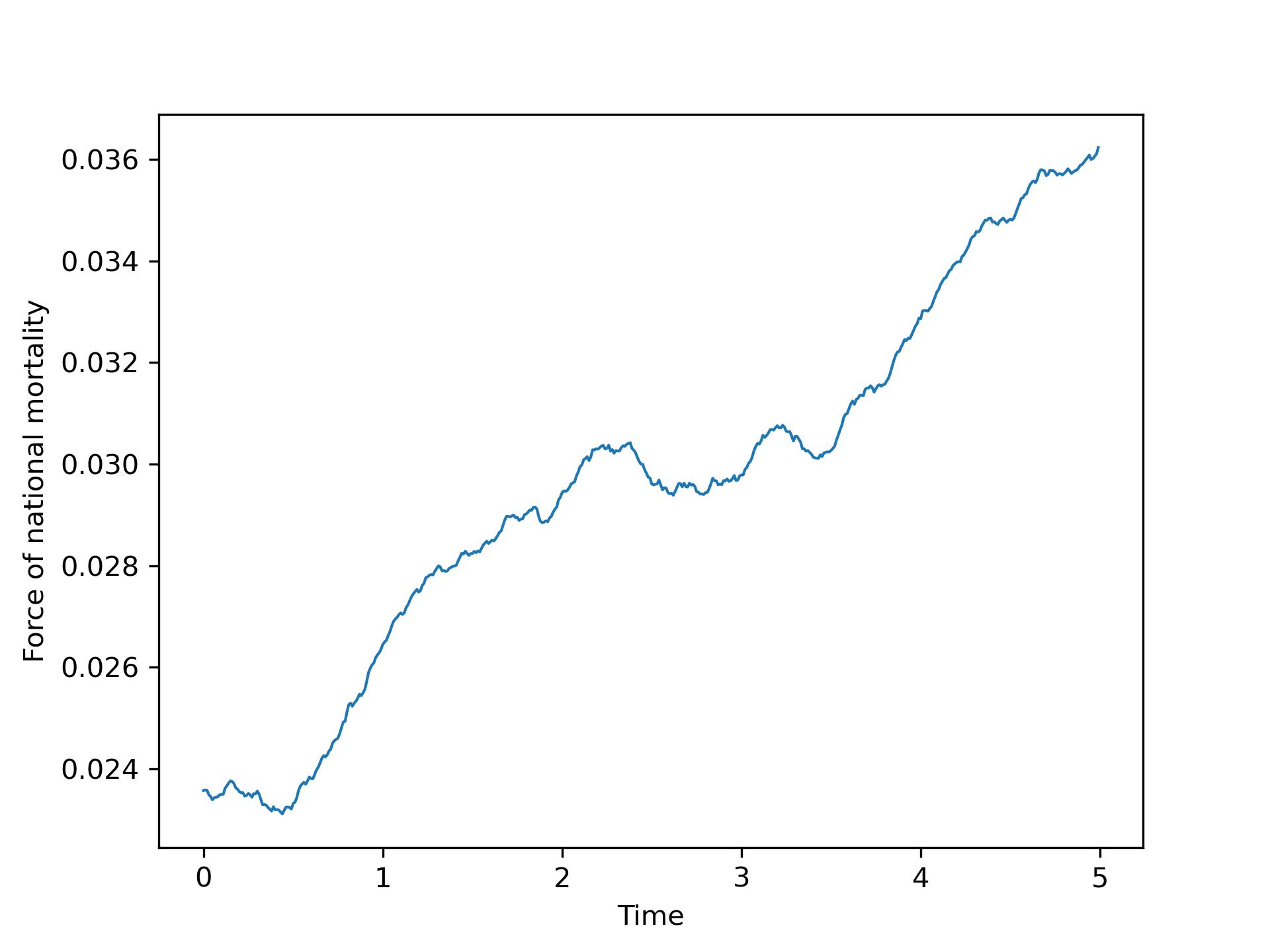}}
	\subfigure{\includegraphics[width= 7.5cm]{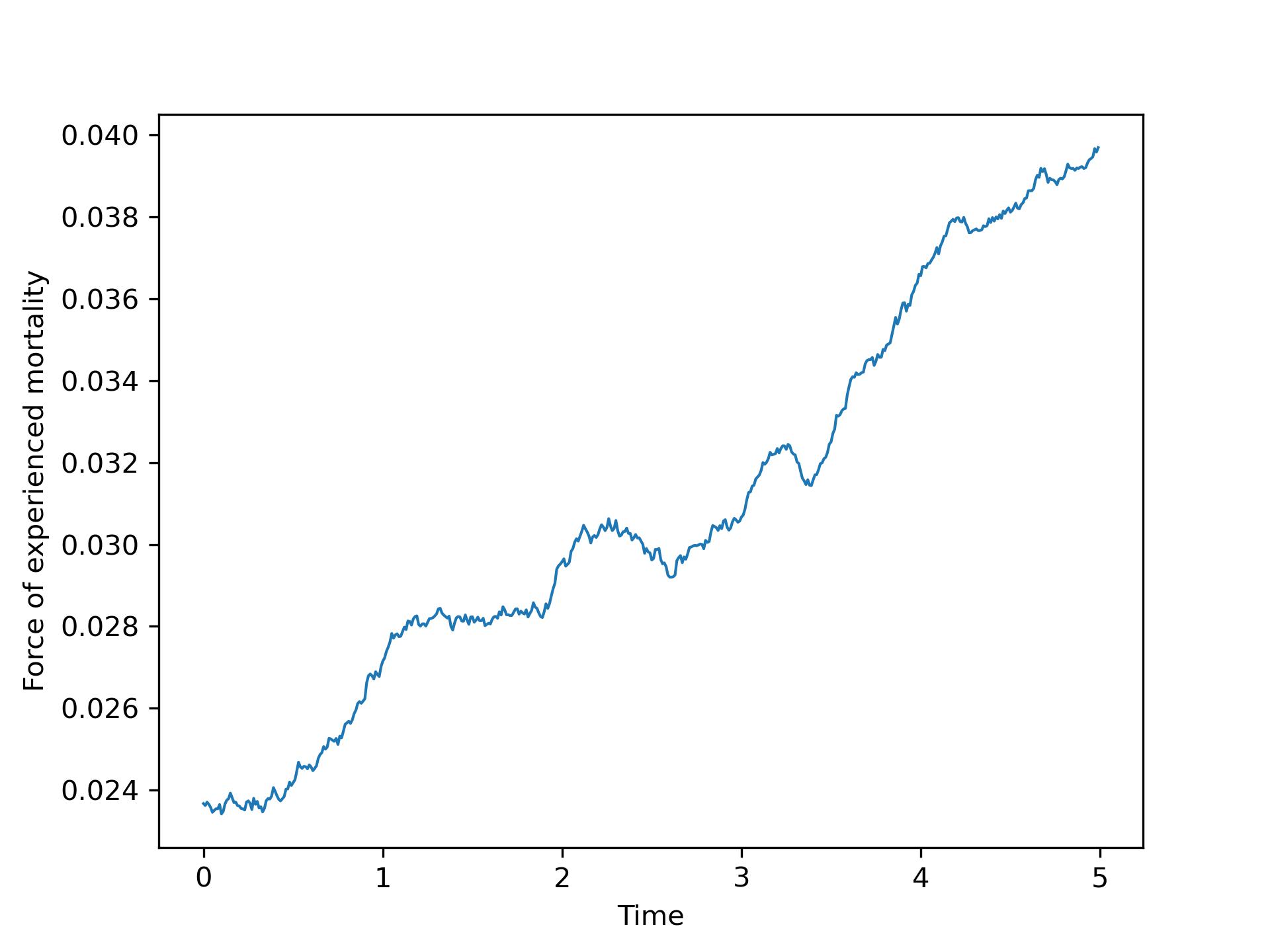}}
			\caption{The sample paths of the forces of mortality}
			\label{fig:mu_hat}
		\end{figure}
To examine the effect of the cointegration, we consider the following two models.
\begin{enumerate}
\item[\textbullet] Model 1: Assumption \ref{assume1} and Vasicek model in \eqref{Vas_num};
\item[\textbullet] Model 2: Assumption \ref{assume1} and Vasicek model in \eqref{Vas_num} with $\beta_1 = \beta_2 = 0$.
\end{enumerate}
The parameter values are estimated using the simulated path of $\mu_2$ under the two different models, in which Model 1 corresponds to the case where the cointegration is considered by the insurer, whereas in Model 2, the insurer ignores the cointegration. The estimated parameters are given in Table \ref{tab:parameter}.
\begin{table}[h]
    \centering
    \begin{tabular}{c|c|c|c|c|c}
    \hline
      & $\beta_1$ & $\beta_2$ & $b_2$ & $\theta_2$ & $\sigma_2$ \\
    \hline
      True value & 1 & 0.6 & 0.0001 & 0.6 & 0.001 \\
    \hline
    Model 1 & 1.044 & 0.586 & 0.0001 & 0.563 &  0.001  \\
    \hline
    Model 2 & 0  & 0  & 0.0037 & 0.165 & 0.001 \\
    \hline
    \end{tabular}
    \caption{Parameter values of $\mu_2$. }
    \label{tab:parameter}
\end{table}
By setting the parameter values of the interest rate \eqref{interest_num} at $r(0) = 0.04$, $b_r = 0.02$, $\theta_r = 0.6$, and $\sigma_r = 0.01$, a sample path for $r(t)$ is simulated, as shown in Figure \ref{fig:r}.
\begin{figure}[H]
    \centering
    \includegraphics[width = 7.5cm]{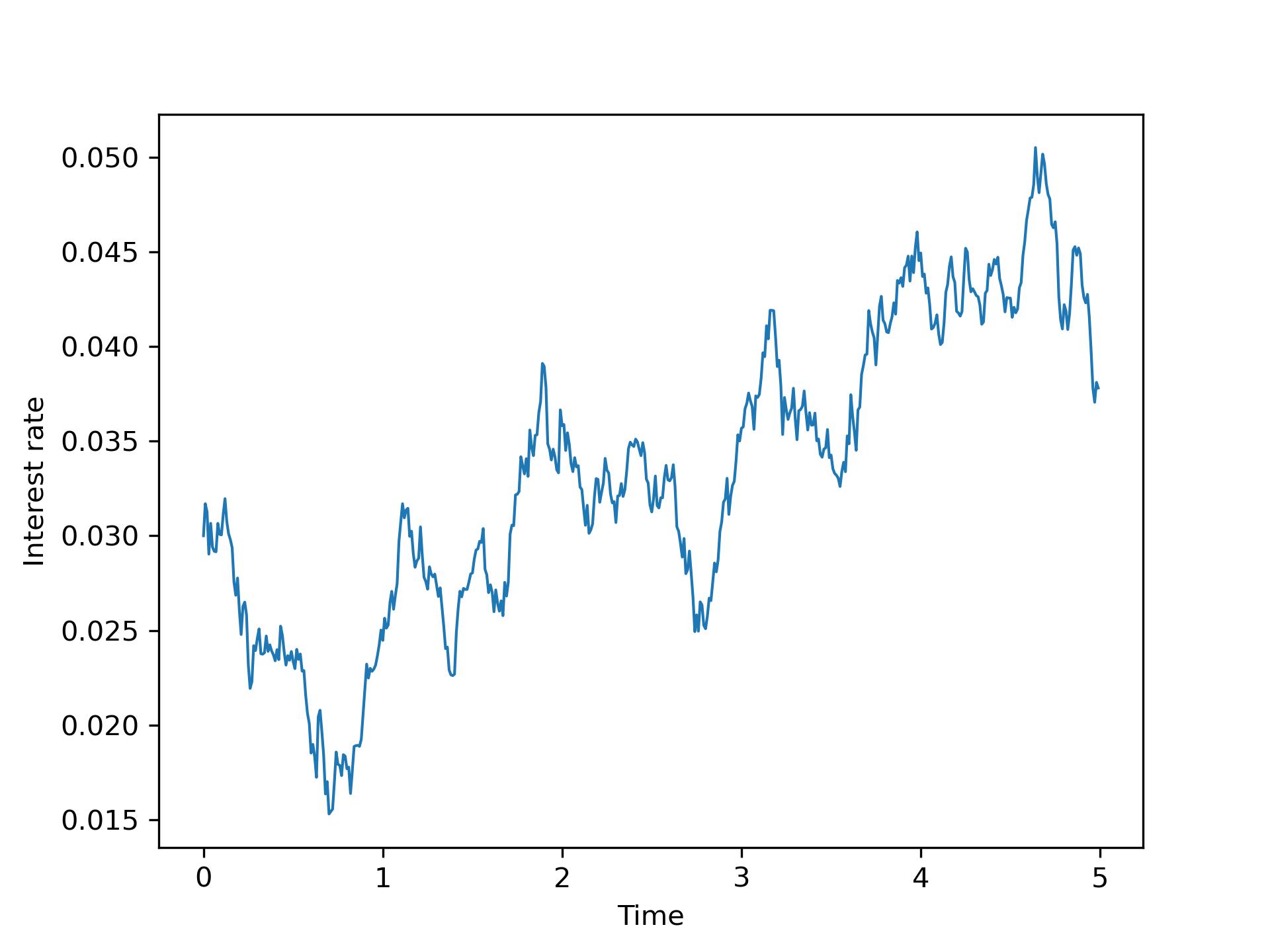}
    \caption{Simulated sample path of the interest rate}
    \label{fig:r}
\end{figure}
  {In addition to Model 1 and 2, we also investigate the no-hedging case, where the insurer can only invest in the zero-coupon bond.  Therefore, $u_1^*(t)\equiv 0$ in this case. 
It can be seen that the equilibrium investment  strategy for no hedging case is $$ u_2^*(t) = e^{\int_0^t r(s)ds}\frac{\lambda\vartheta(t) - \gamma_3^{(1)}(t)}{d_1(t)\sigma_r},$$ where $d_1(t) = \frac{1}{\theta_r}\left(e^{-\theta_r(T-t)} -1\right)$, and $\gamma_3^{(1)}$ is defined in Theorem \ref{thm:gam_Vas}. To examine the hedging performance of the models,  we further set   the frequency of life insurance claims experienced by  the insurance company $c_1 = 8$,  the amount of the  pension annuity  payment at the initial time $c_2 = 4$,  the maturity time of the zero-coupon bond and the zero-coupon longevity bond $T$  at 10,  the risk aversion parameter $\lambda$ at 10, and the initial wealth level $X(0) = 100$.   The insurer is allowed to both buy and sell the bond and the longevity bond. The market price of mortality risk $\varphi$ and the market price of interest rate risk $\vartheta$  are both set at 0.1.} The equilibrium hedging strategies under the Model 1 and 2 are calculated according to the closed-form solution in \eqref{u*} and Theorem \ref{thm:gam_Vas}.   {The equilibrium strategies under Model 1 and 2 are shown in Figure \ref{fig:u*}. From Figure \ref{fig:u*}, we can see that the evolution of  the equilibrium hedging strategy $u_1^*$  is dominated by the changes of $\psi_1(T-t)$, the solution to the Volterra-Riccati equation.}
The discounted wealth processes under different equilibrium strategies corresponding to the simulated sample path are shown in Figure \ref{fig:X}. From Figure \ref{fig:X}, it is evident that ignoring the cointegration between the models, given that it exists, significantly influences the performance of the hedging strategy.

\begin{figure}[H]
	\centering
	\subfigure{\includegraphics[width= 7.5cm ]{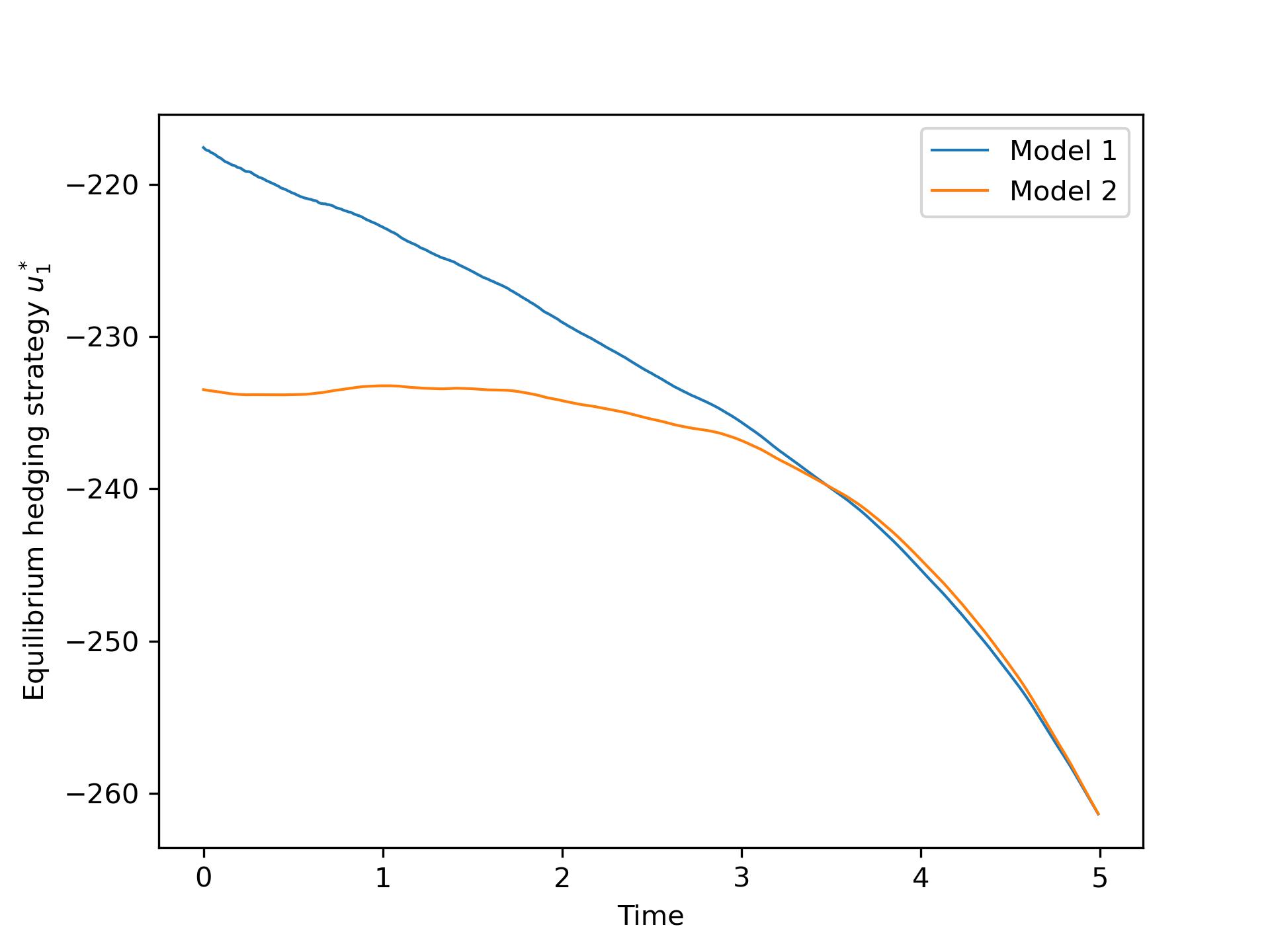}}
	\subfigure{\includegraphics[width= 7.5cm]{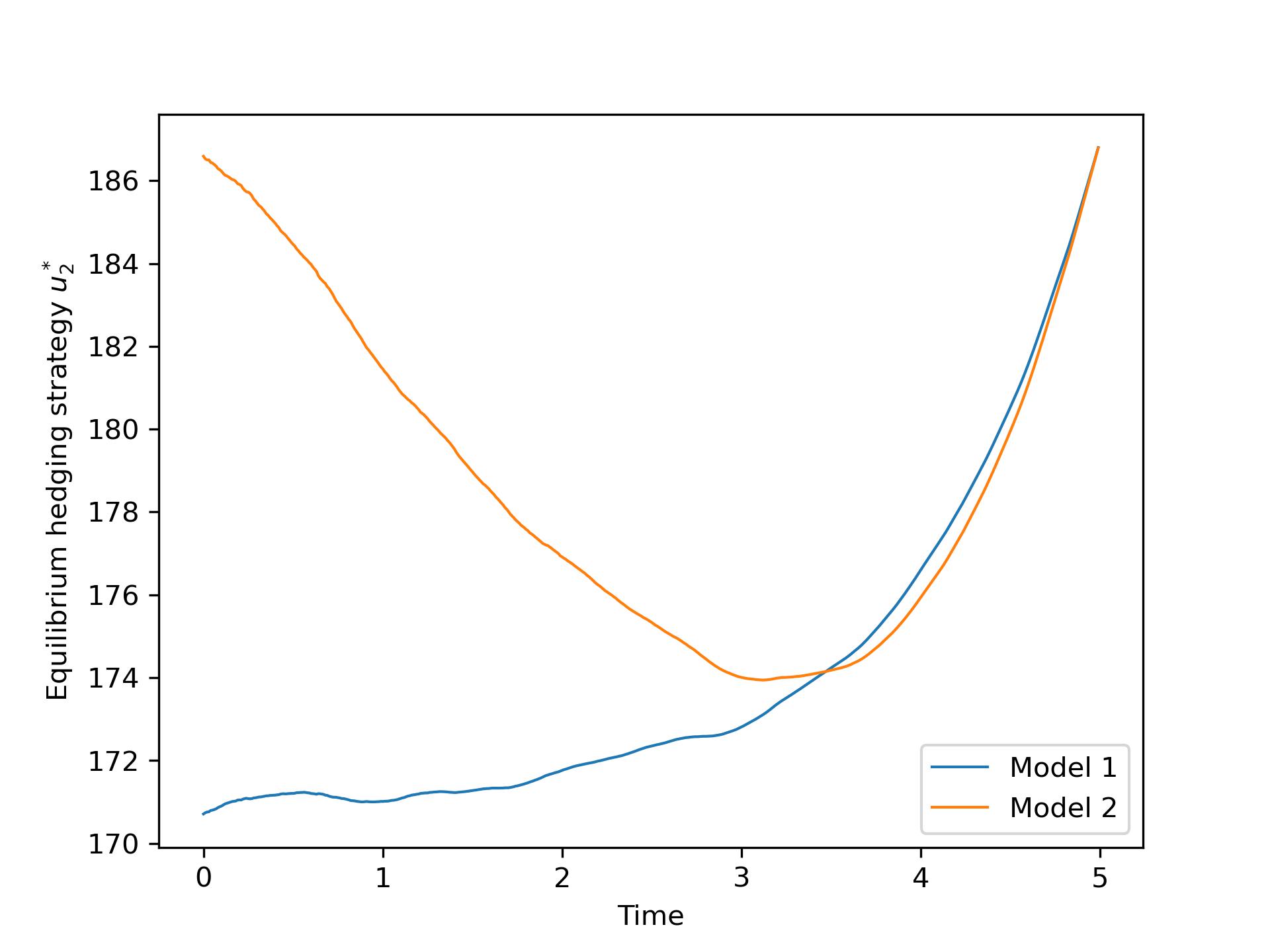}}
		\caption{The equilibrium hedging strategies under Model 1 and 2.}
	\label{fig:u*}
\end{figure}
  
\begin{figure}[H]
    \centering
    \includegraphics[width = 7.5cm]{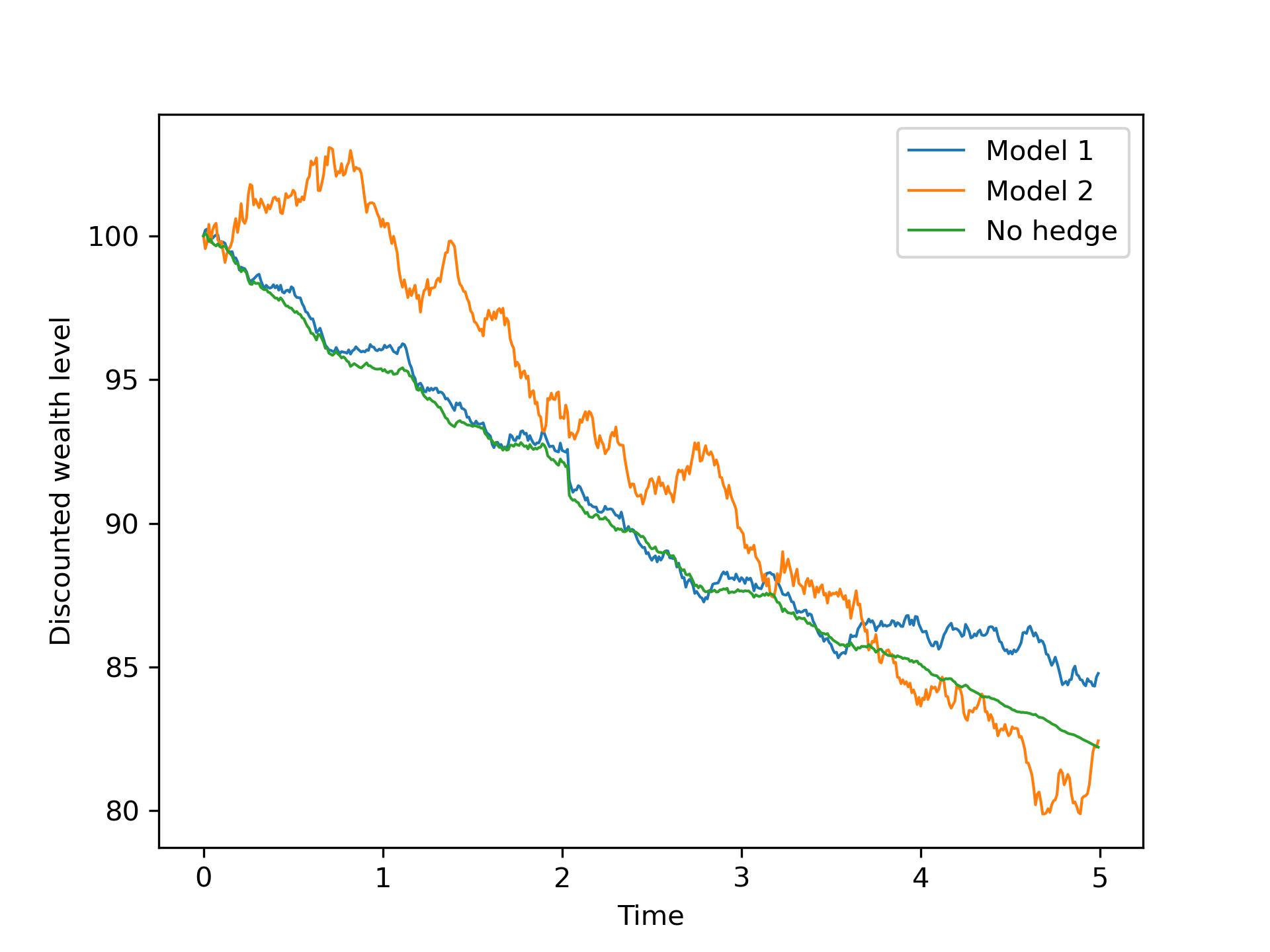}
    \caption{Discounted wealth processes under the different hedging strategies.}
    \label{fig:X}
\end{figure}
To investigate the general performance of equilibrium hedging strategies under Models 1 and 2, we plot the efficient frontiers. For a fixed risk aversion parameter $\lambda$, we simulate 10,000 pairs of paths for the forces of mortality and interest rates. Equilibrium strategies are calculated for each pair of sample paths. Then, 10,000 samples of the terminal discounted wealth level under Model 1, Model 2, and the no-hedging case are collected, respectively. By varying $\lambda$ from 1 to 40, the sample expectations and variances of the terminal discounted wealth $X_T$ are plotted in Figure \ref{fig:EF-coin}. Figure \ref{fig:EF-coin} shows that the equilibrium strategy under the LRD mortality model with cointegration outperforms the strategy when the cointegration between mortality rates is ignored.
\begin{figure}[H]
    \centering
    \includegraphics[width = 8 cm]{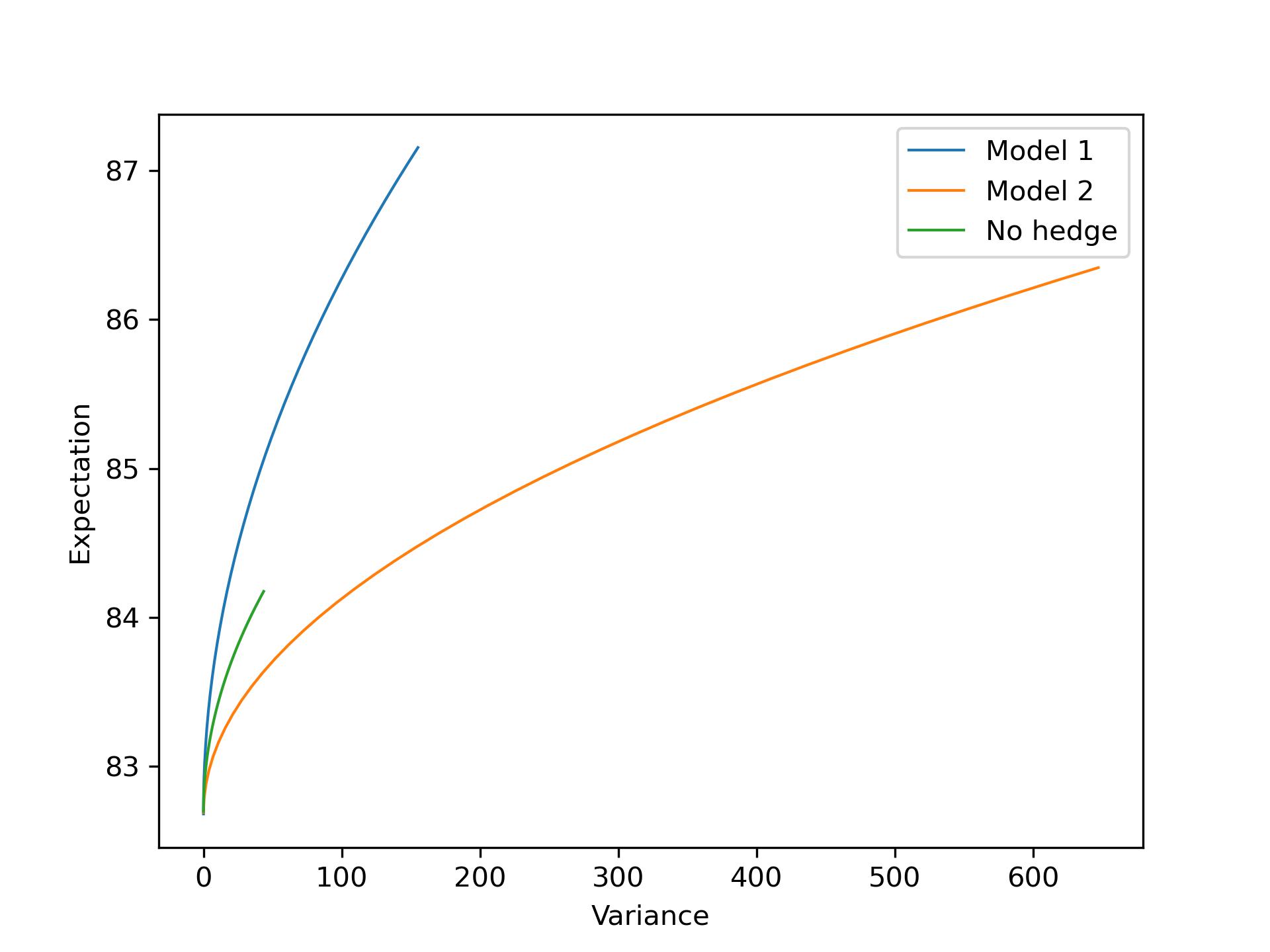}
    \caption{Efficient frontiers under  Model 1, Model 2, and no hedging case.}
    \label{fig:EF-coin}
\end{figure}

  {We further investigate how the value changes of  $c_1$ and $c_2$ affect the equilibrium hedging strategy. For the effect of $c_1$, we vary $c_1$ from 6 to $16$ with $c_2$ fixed at 4. The equilibrium hedging strategies correspond to different $c_1$ under Model 1 are plotted in Figure \ref{fig:u1-var}. Figure \ref{fig:u1-var} shows that when $c_1$ increases, the absolute value of the equilibrium hedging strategy $u_1^*$ decreases during the early period of the hedging horizon and increases later in  the hedging horizon.   For the effect of $c_2$, we fix $c_1 = 8$ and vary $c_2$ from $2$ to $6$. The corresponding equilibrium hedging strategies under Model 1 are plotted in Figure \ref{fig:u1-var2}. Figure \ref{fig:u1-var2} shows that the absolute value of $u_1^*$ decreases when  $c_2$ increases.}
\begin{figure}[H]
	\centering
	\subfigure[]{\includegraphics[width= 7.5cm ]{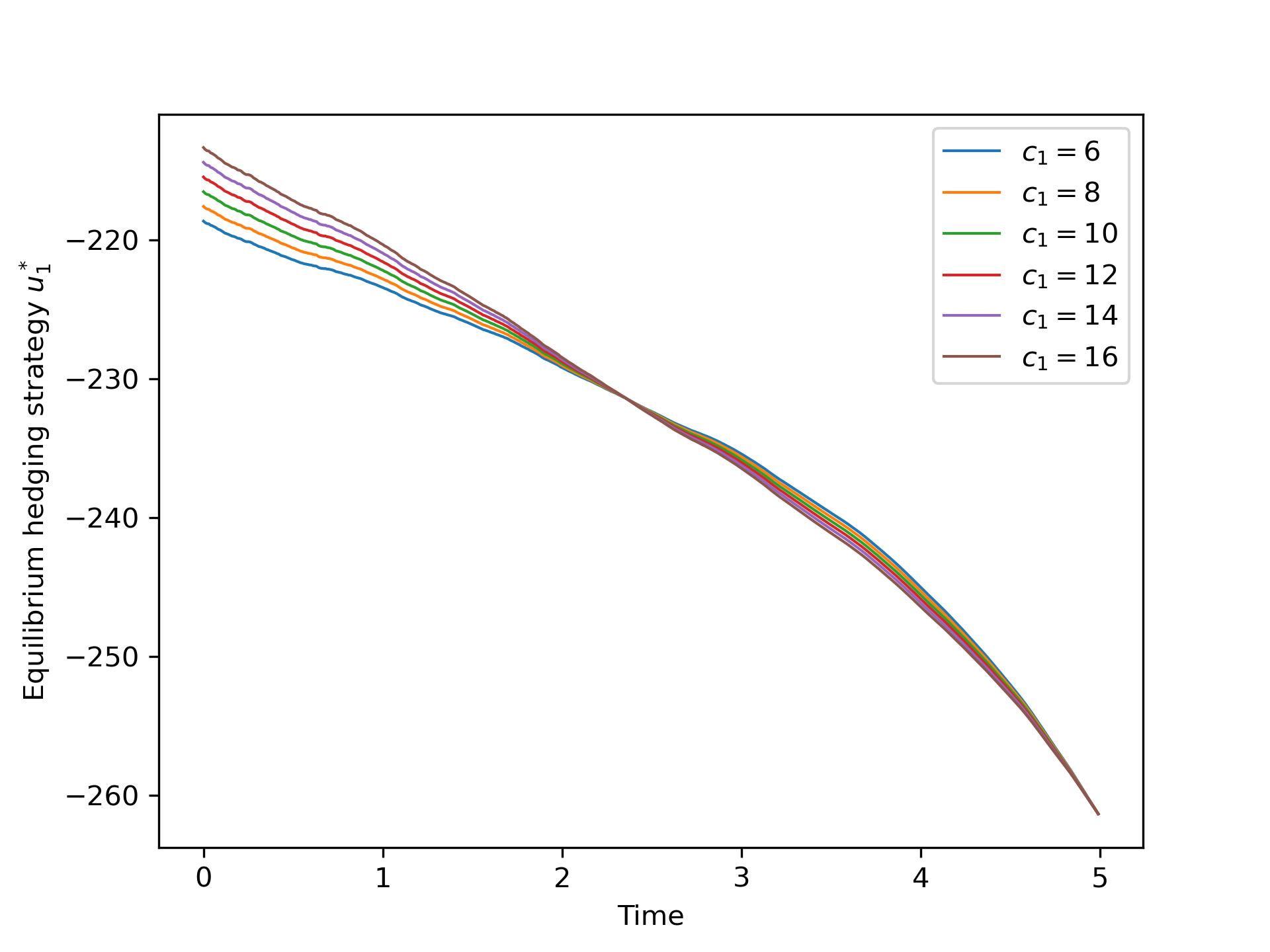}\label{fig:u1-var}}
	\subfigure[]{\includegraphics[width= 7.5cm]{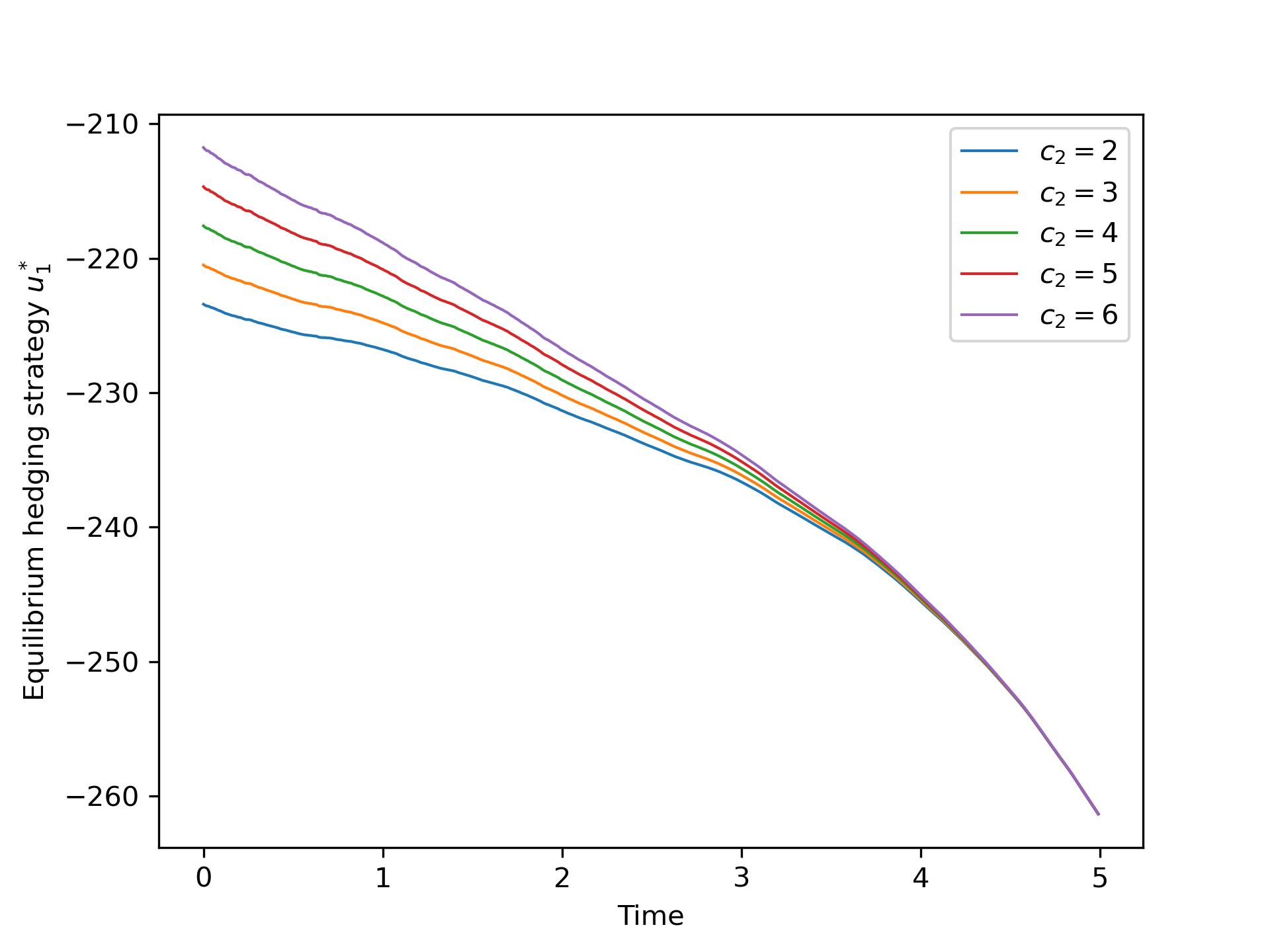}\label{fig:u1-var2}}
		\caption{The equilibrium hedging strategies under different values of $c_1$ and $c_2$.}
\end{figure}
  {Furthermore, one can consider  an alternative to Model 2 by incorporating the correlation between $\mu_1$ and $\mu_2$. Specifically,  
\begin{equation}\label{Model2'}
    d\mu_2(t) = (b_2 - \theta_2\mu_2(t))dt + \sigma_2 dW'_2(t),
\end{equation}
where $dW'_2 = \rho dW_1 + \sqrt{1-\rho^2}dW_2$ with $\rho \in (0, 1)$.  Under this setting, $\mu = (\mu_1, \mu_2)^\top$ is also an affine Volterra process with 
\begin{eqnarray}\label{Cor:parameter}
    \Theta = \left(\begin{array}{ll}
    \theta_1 & 0 \\
    0 & \theta_2
\end{array}\right), K = \left(\begin{array}{ll}
    K_1 & 0 \\
    0 & 1
\end{array}\right), \text{ and } \sigma(\mu) = \left(\begin{array}{ll}
    \sigma_1 & 0 \\
   \rho\sigma_2 & \sqrt{1-\rho^2}\sigma_2
\end{array}\right). 
\end{eqnarray} 
Then, the equilibrium hedging strategy  is still of the form in Corollary \ref{cor:Vas} and Theorem \ref{thm:gam_Vas} with  $\theta$, $K$, and $\sigma(\mu)$ given by Equation \eqref{Cor:parameter}. Denote the model under this setting as Model 2'.  Based on the sample path of $\mu$, the correlation  coefficient  is estimated as $\rho = 0.089$.  We then calculated the corresponding equilibrium hedging strategy under Model 2'. The discounted wealth process is plotted in Figure \ref{fig:X-2}.  From Figure \ref{fig:X-2}, we can see that the performances of Model 2' and Model 2 are similar. } 
\begin{figure}[H]
    \centering
    \includegraphics[width = 7.5cm]{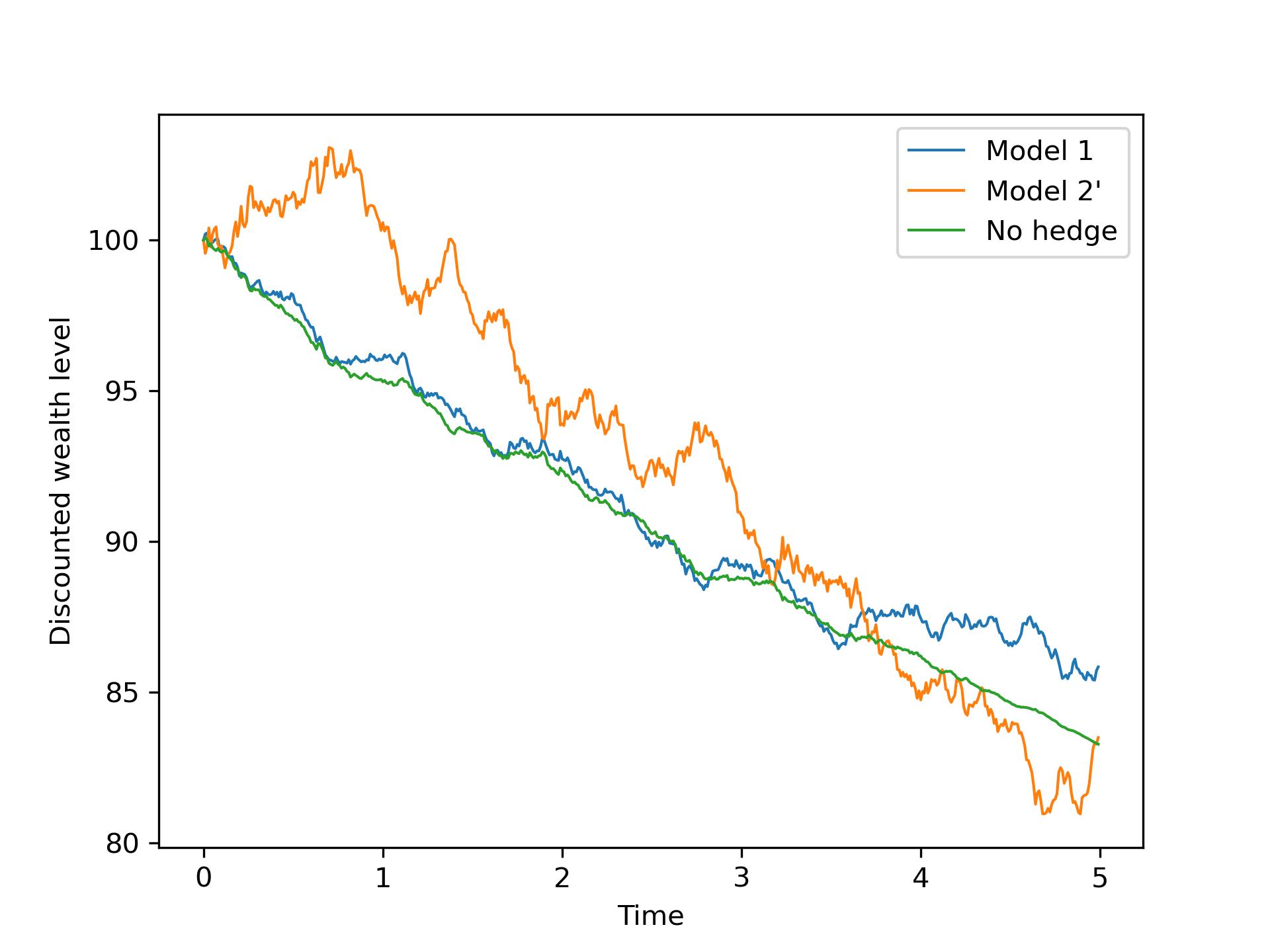}
    \caption{Discounted wealth processes under the Model 1, Model 2' and the no hedging case.}
    \label{fig:X-2}
\end{figure}
\subsection{Effect of the LRD property of the   {force of mortality} }
Although \cite{WW2021} show that the LRD property in the national   {force of mortality}  has a significant effect on the longevity hedging strategy, they assume that there is no difference between the national   {force of mortality}  and the insurer-experienced rate. Therefore, in this part, we examine the effect of the LRD property when the LRD cointegration is considered. We are interested in the longevity hedging strategy under Model \eqref{Vas_num} compared with its Markovian counterpart.
In this part, we set the parameter values in Model \eqref{Vas_num} the same as in Section \ref{Sec:num_coin}. We simulate a pair of sample paths of   {$\mu_1$ and $\mu_2$}, as shown in Figure \ref{fig:mu_LRD}.   {After adding $m_i(t)$ to $\mu_i(t)$, $i = 1, 2$, the sampled paths of the forces of national mortality and experienced mortality are shown in Figure \ref{fig:mu_hat-LRD}.} A sample path of the interest rate is shown in Figure \ref{fig:interest-lrd}.
    \begin{figure}[H]
	\centering
	\subfigure{\includegraphics[width= 7.5 cm]{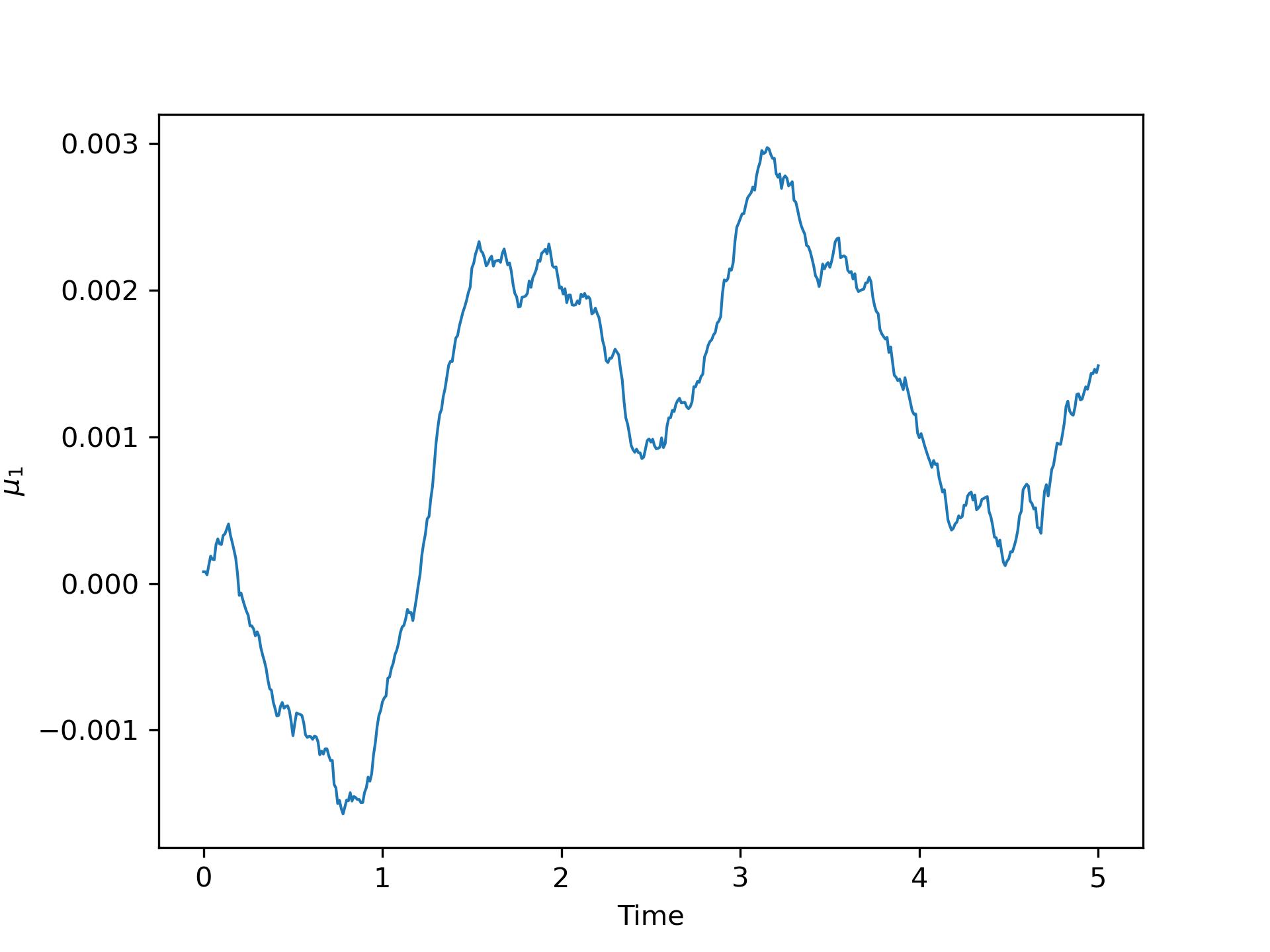}}
	\subfigure{\includegraphics[width= 7.5cm]{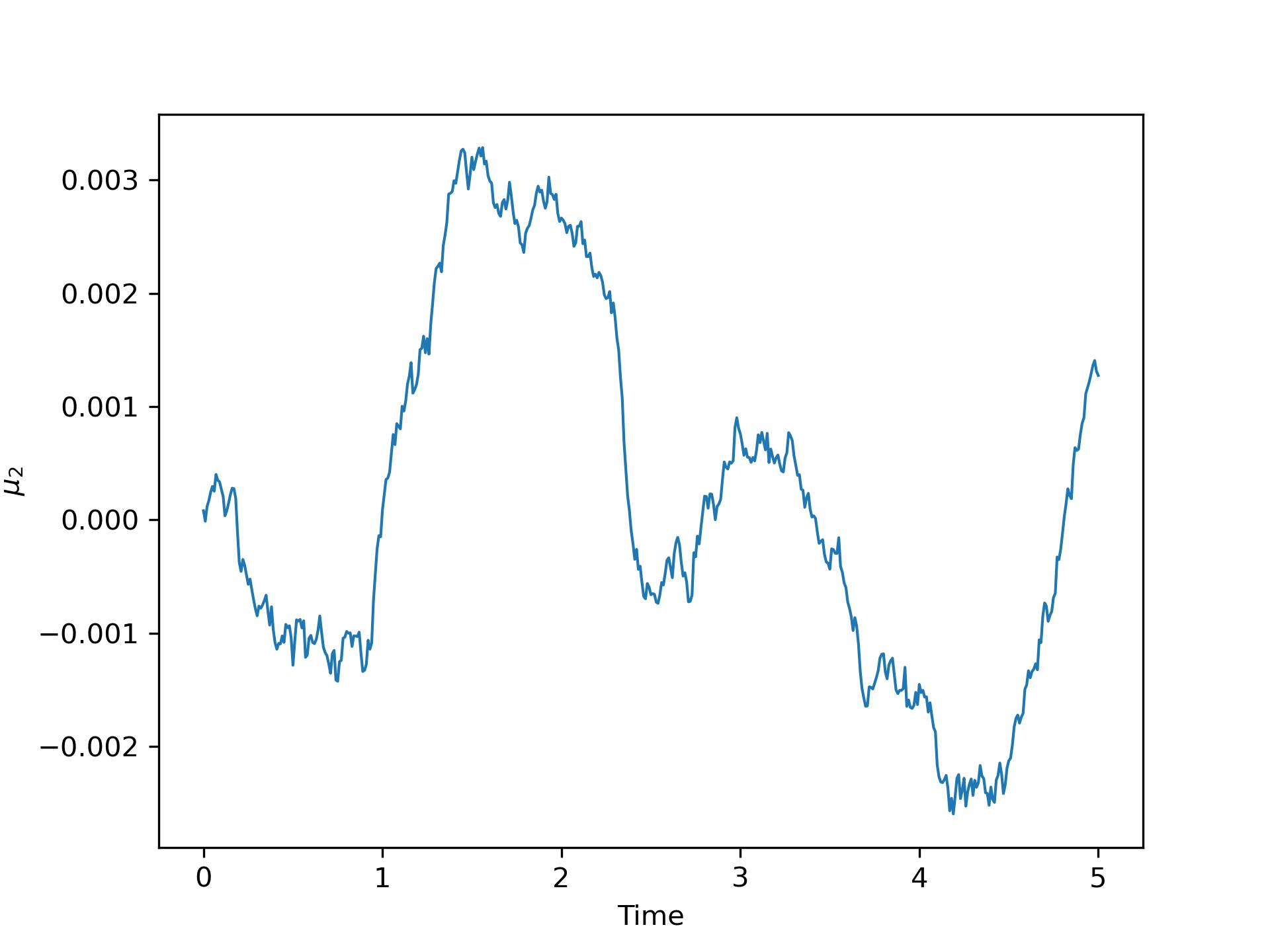}}
	\caption{Simulated sample paths of $\mu_1$ and $\mu_2$}
	\label{fig:mu_LRD}
	\end{figure}
 \begin{figure}[H]
			\centering
	\subfigure{\includegraphics[width= 7.5 cm ]{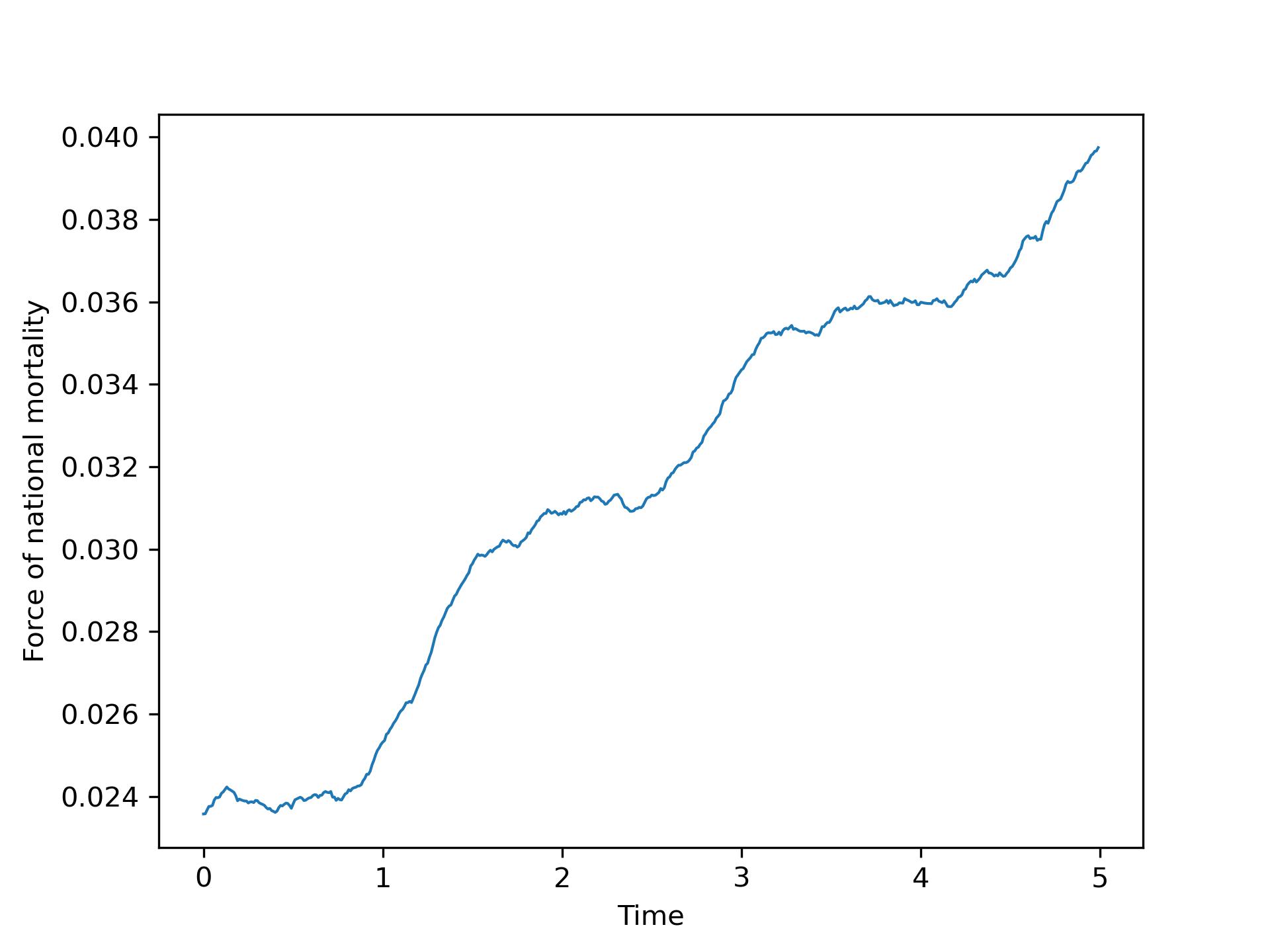}}
	\subfigure{\includegraphics[width= 7.5cm]{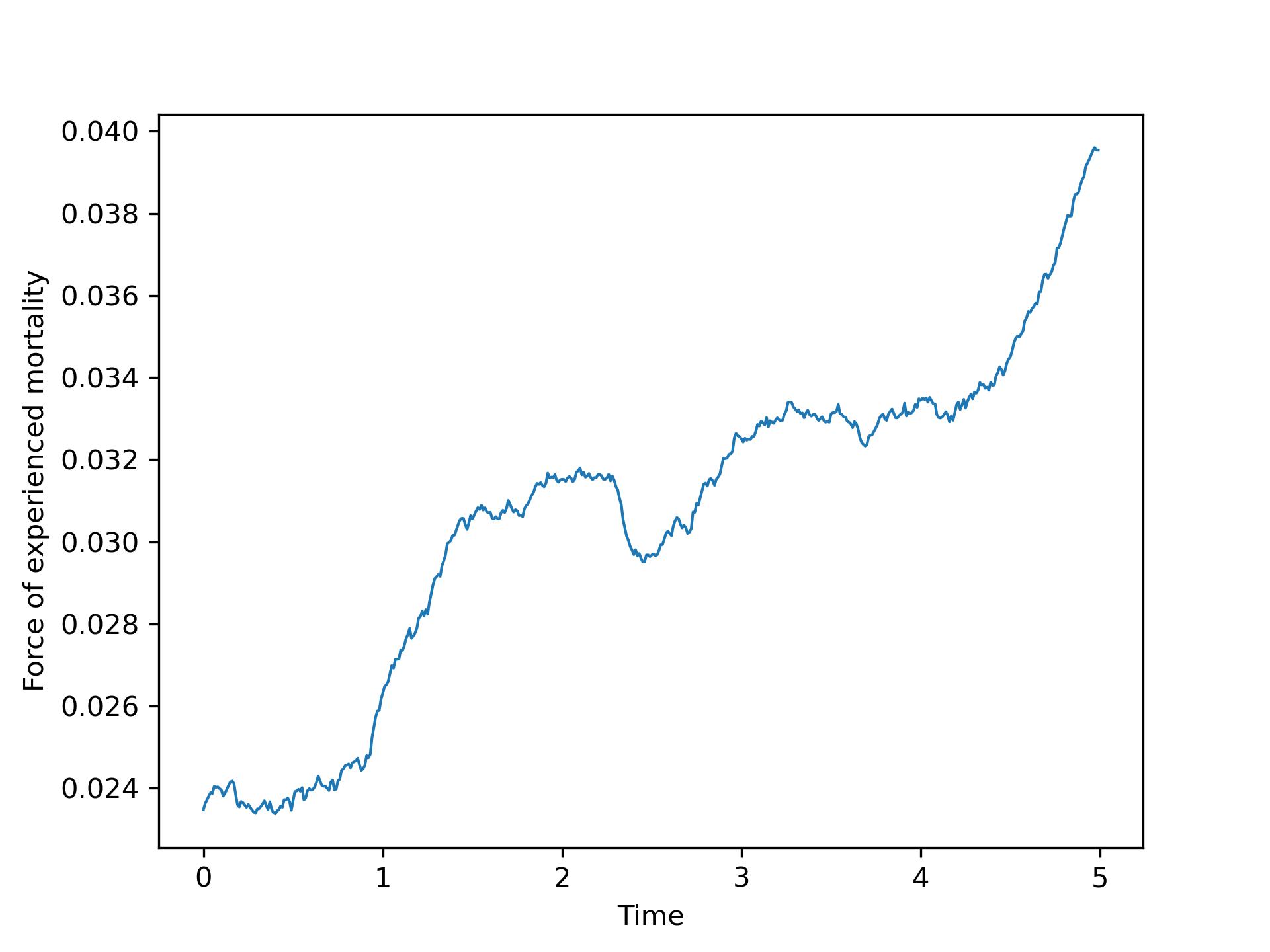}}
			\caption{The sample paths of the forces of mortality}
			\label{fig:mu_hat-LRD}
		\end{figure}
   \begin{figure}[H]
    \centering
    \includegraphics[width = 8 cm]{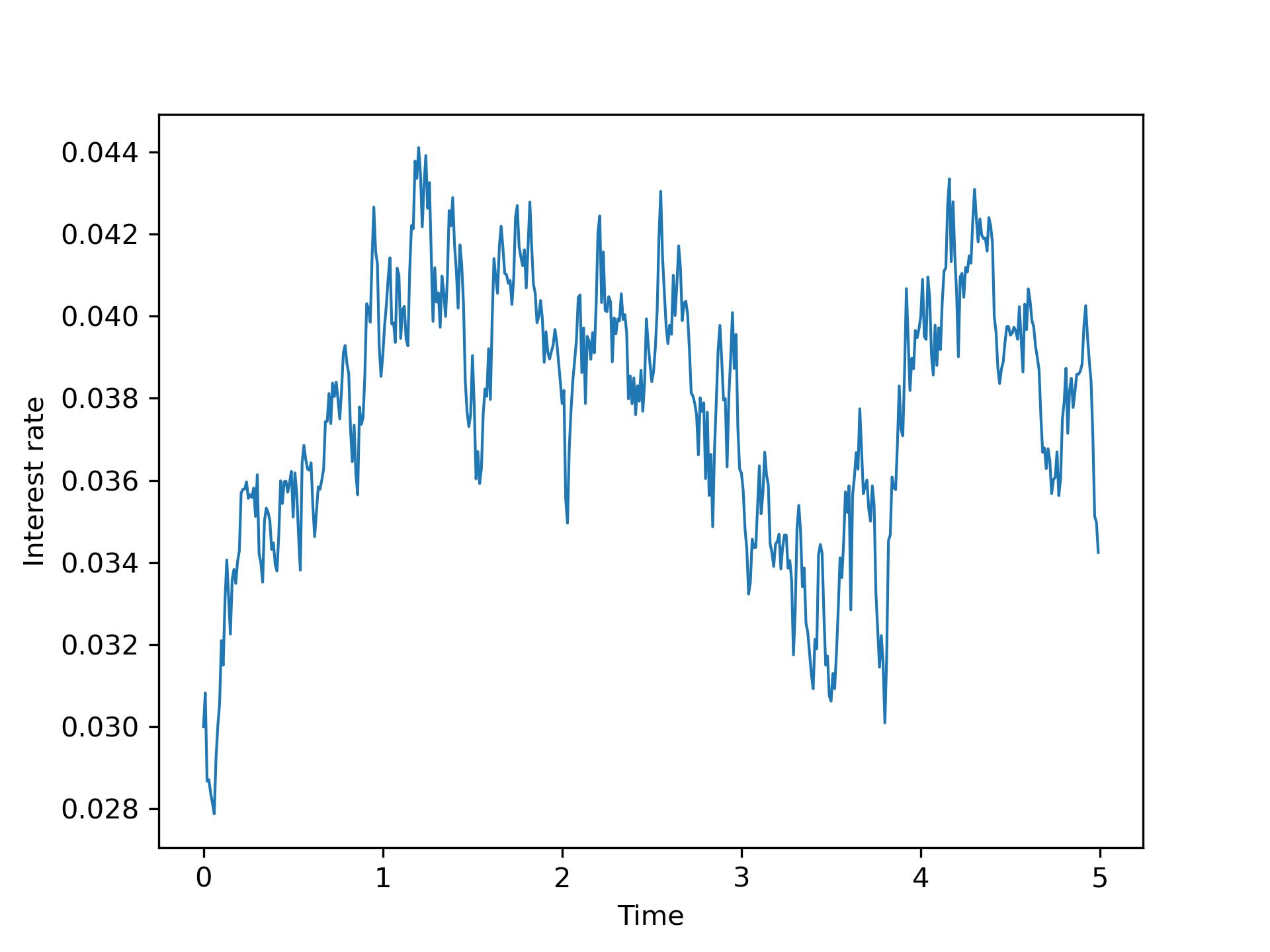}
    \caption{Simulated sample path of the interest rate}
    \label{fig:interest-lrd}
\end{figure}
Corresponding to the sample paths, the longevity hedging strategies are calculated according to the closed-form solution in \eqref{u*} and Theorem \ref{thm:gam_Vas} under the true model ($H_1 = 0.83$, $H_2 = 0.50$) and its Markovian counterpart ($H_1 = H_2 = 0.50$).   {The equilibrium hedging strategy in the no hedging case is the same as that in Section \ref{Sec:num_coin}, where the insurer is only allowed to invest in the zero-coupon bond.} The discounted wealth processes under the longevity hedging strategies in the three cases are shown in Figure \ref{fig:X_lrd}. From Figure \ref{fig:X_lrd}, we can see that the LRD mortality model with cointegration outperforms its Markovian counterparts, given that the LRD feature exists. 
\begin{figure}[H]
	\centering
	\subfigure{\includegraphics[width= 7.5cm ]{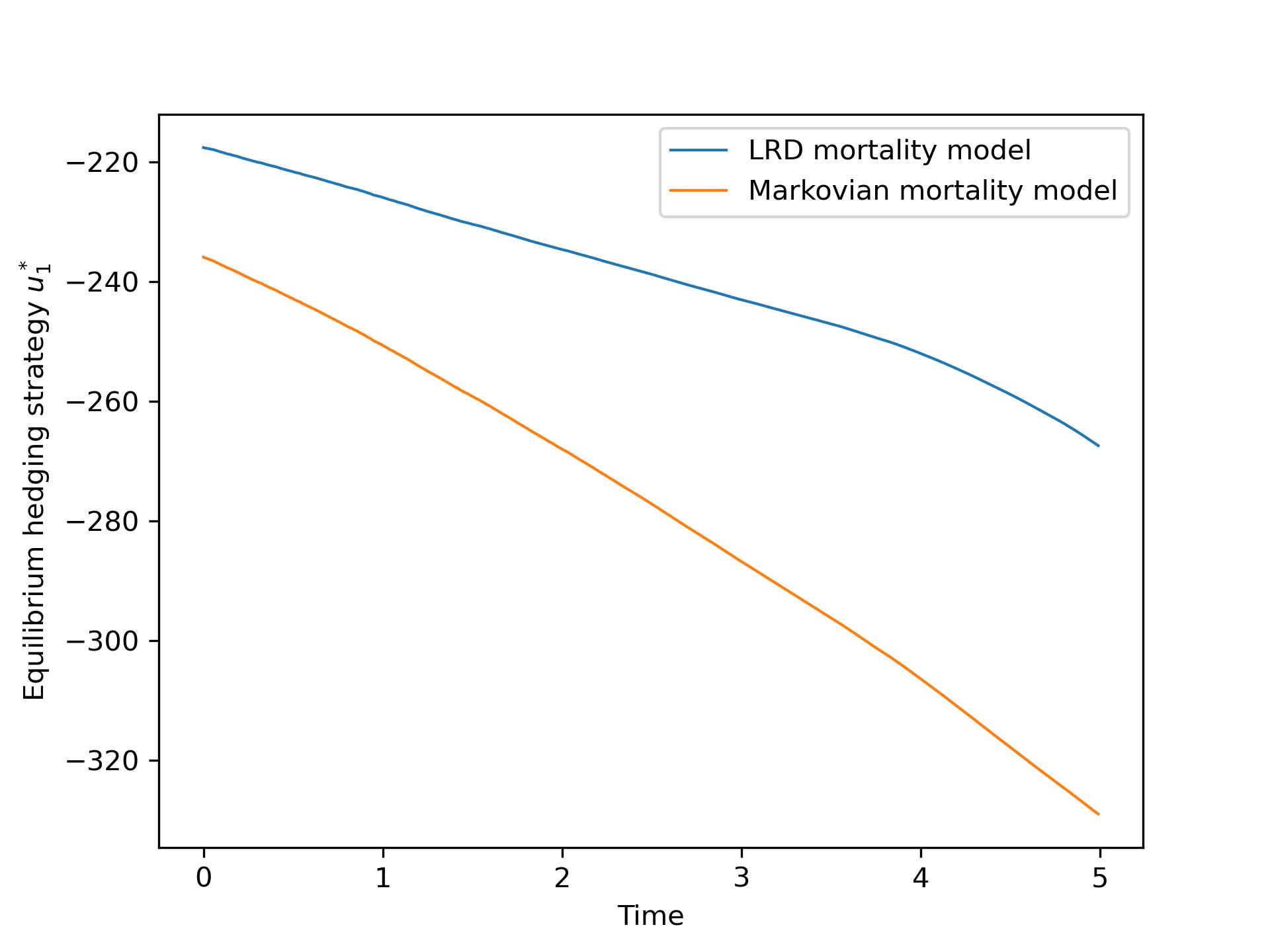}}
	\subfigure{\includegraphics[width= 7.5cm]{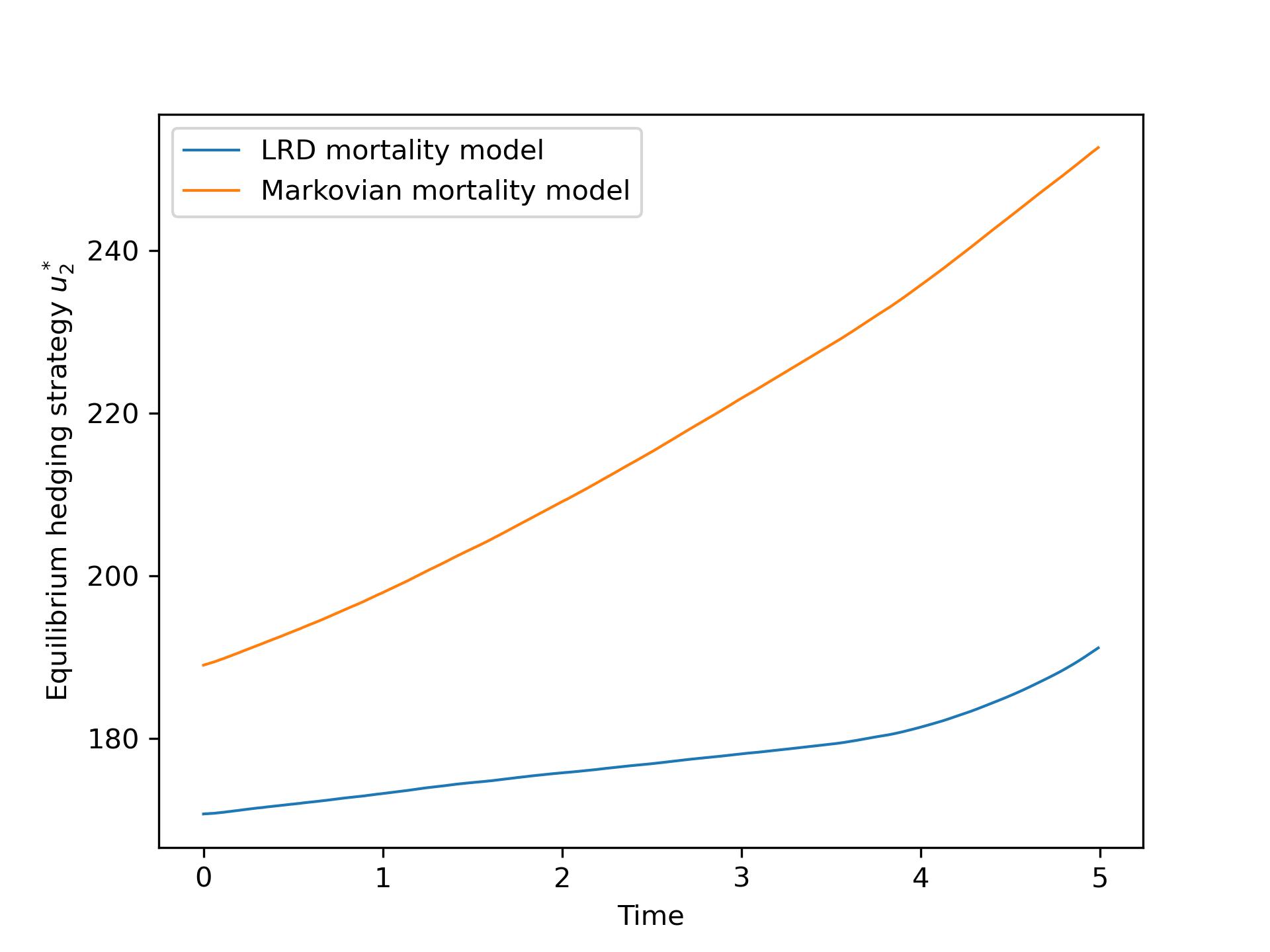}}
		\caption{The equilibrium hedging strategies under  under the LRD mortality model and its Markovian counterpart.}
\label{fig:u*-LRD}
\end{figure}
\begin{figure}[H]
    \centering
    \includegraphics[width = 8cm]{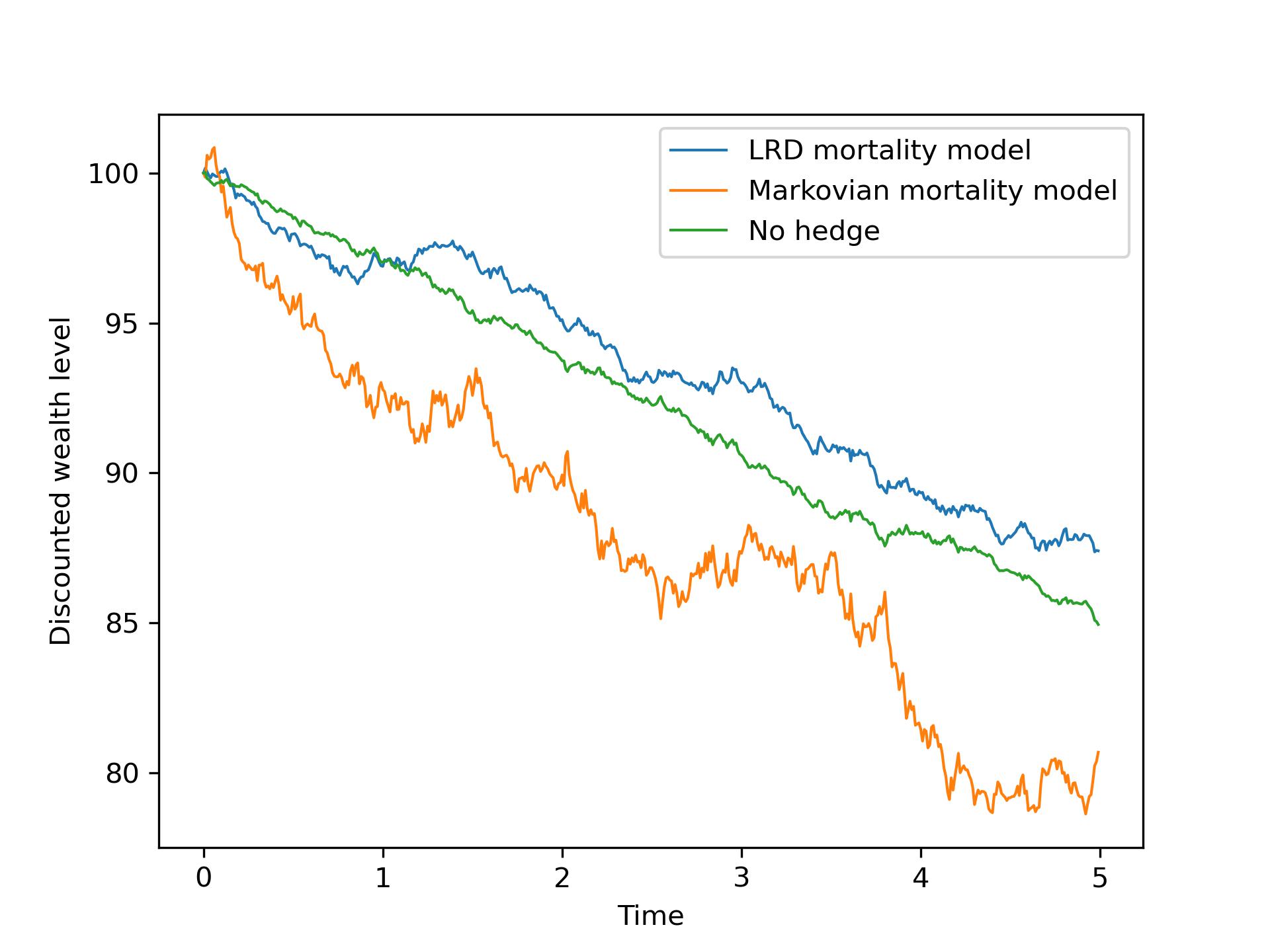}
    \caption{Discounted wealth processes under the LRD mortality model, its Markovian counterpart, and no hedging case.}
    \label{fig:X_lrd}
\end{figure}

Similar  to that  in Section \ref{Sec:num_coin}, we also plot  the efficient frontiers for  LRD mortality model  and its Markovian counterpart  by varying $\lambda$ from 1 to 40, as shown in Figure \ref{fig:EF_lrd}. From Figure \ref{fig:EF_lrd}, we can see that equilibrium hedging strategy under  LRD mortality model outperforms its Markovian counterpart. 
\begin{figure}[H]
    \centering
    \includegraphics[width = 8 cm]{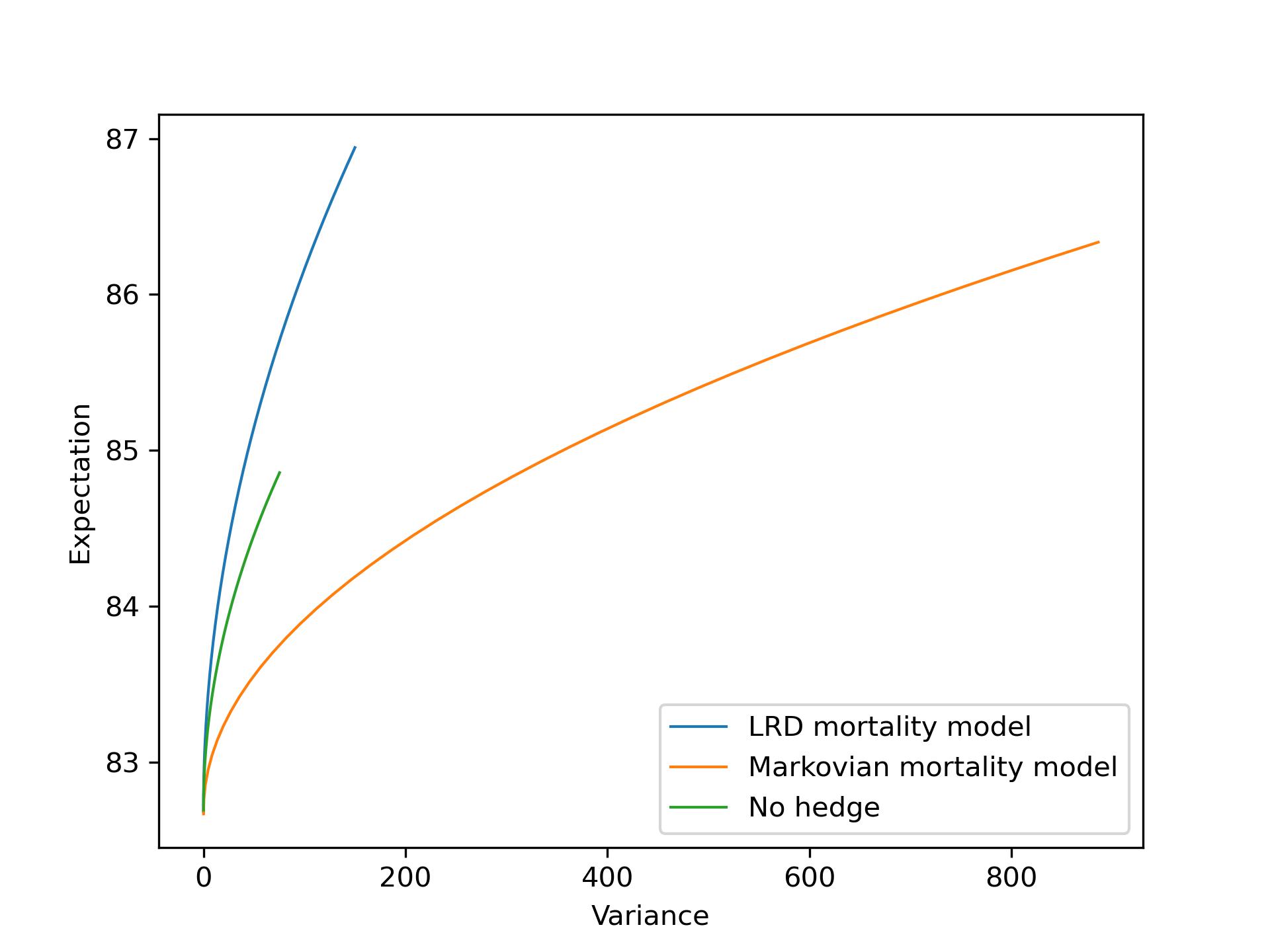}
    \caption{Efficient frontiers under the LRD mortality model, its Markovian counterpart, and no hedging case.}
    \label{fig:EF_lrd}
\end{figure}

\section{Conclusion}
\label{Sec:conclude}
This paper proposes the first  mortality model which incorporates the cointegration and the LRD property of   {force of mortality}  simultaneously by using a two-dimensional affine Volterra process.  The cointegration between the national and the insurer's experienced   {forces of mortality} is theoretically shown to be crucial for measuring the LRD feature. A time-consistent longevity hedging problem is investigated under the proposed mortality model.  We provide the equilibrium hedging strategies explicitly under both deterministic and volatility-driven market prices of risks, and prove the uniqueness of the equilibrium strategies. The numerical studies in this paper examine the effects of incorporating the  cointegration and correlation between   {forces of mortality} and LRD property on hedging performance.   {Numerical results indicate that both the effects of incorporating cointegration and  the correlation  between forces of mortality on hedging performance are significant.}

\bigskip
\noindent {\bf Acknowledgement:} L. Wang acknowledges the support from the National Natural Science Foundation of China (Grant number 12301615).  H.Y. Wong acknowledges the support of the Research Grants Council
of Hong Kong through the research funding schemes (RMG8601495, GRF14308422).

\bigskip

\noindent {\bf Disclosure statement:} All authors declare that they have no conflicts of interest.

\begin{appendices} 
	\section{Transformation of affine Markovian processes}
	\label{appendix:affine}
	Suppose that a $\mathcal{F}_t$-adapted Markovian process $r(\cdot)\in \mathbb{R}$ is the strong solution to the SDE:
	$$dr(t)=(b_r(t) - \theta_r r(t))dt +\sigma_r(r(t))dW_r(t), $$
	where $b_r(t)\in \mathbb{R}$ is a bounded deterministic function,  $\theta_r$ is a constant, and $W_r$ is a standard Brownian motion. Suppose that $\sigma_r(r(t))^2 = a_r^0 + a_r^1r(t)$ for some constants $a_r^0$ and $a_r^1$. 
    From the study in \cite{DPS}, we have
	\[\mathbb{E}\left[\left.e^{\int_{t}^{T}-r(s)ds}\right|\mathcal{F}_t\right]= e^{d_0(t)+ d_1(t)r(t)},\]
	where the functions $ d_0(\cdot)\doteq d_0(\cdot, T)$ and $d_1(\cdot)\doteq d_1(\cdot, T)$ solve the following ordinary differential equations:
	\[\dot{d_1}(t)= 1 + \theta_rd_1(t) - \frac{1}{2}a^1_rd_1(t)^2,\]
	\[\dot{d_0}(t)= -b_r(t) d_1(t)-\frac{1}{2}a_r^0d_1(t)^2,\]
	with boundary conditions $d_1(T)=0$ and $d_0(T)= 0$.

\section{Proofs}
\subsection{Proof of Proposition \ref{pro:Hurst}}
\label{appendix:Hurst}
\begin{proof}
    Define  a sequence of stopping time as $\tau_n := \inf\{t > 0; \mu_1(t)>n, \mu_2(t)>n\}$ for $n = 1, 2, \cdots$.  For each $n$, we introduce the following dynamic:
    \begin{align*}
        \mu_2^n(t) & = \mu_2(0) + \int_{0}^{t}\beta_1 K_1(t-s)(b_1 - \theta_1 \mu_1(s))\mathbbm{1}_{\{s< \tau_n\}}ds  \\
       & + \int_0^tK_2(t-s)(b_2 - \beta_2 \mu_1(s) - \theta_2 \mu_2(s))\mathbbm{1}_{\{s< \tau_n\}}ds \nonumber \\
& + \int_{0}^{t}\beta_1 K_1(t-s)\sigma_1(\mu_1(s))\mathbbm{1}_{\{s< \tau_n\}}dW_1(s) + \int_0^t K_2(t-s)\sigma_2(\mu_2(s)) \mathbbm{1}_{\{s< \tau_n\}}dW_2(s). 
    \end{align*}
For any $p > 2$ and $0 \leq s < t \leq  T$, 
\begin{align*}
    |\mu_2^n(t)-\mu_2^n(s)|^p &\leq C\Bigg[\left| \int_{s}^{t}K_1(t-u)(b_1 - \theta_1 \mu_1(u))\mathbbm{1}_{\{u< \tau_n\}}du \right|^p \\
    & + \left|\int_0^s(K_1(t-u) - K_1(s-u))(b_1 - \theta_1 \mu_1(u))\mathbbm{1}_{\{u< \tau_n\}}du\right|^P \\
    & + \left| \int_s^tK_2(t-u)(b_2 - \beta_2 \mu_1(u) - \theta_2 \mu_2(u))\mathbbm{1}_{\{u < \tau_n\}}du \right|^p\\
    & + \left|\int_0^s(K_2(t-u) - K_2(s-u))(b_2 - \beta_2 \mu_1(u) - \theta_2 \mu_2(u))\mathbbm{1}_{\{u < \tau_n\}}du \right|^p \\
    & + \left|\int_s^t K_1(t-u)\sigma_1(\mu_1(u))\mathbbm{1}_{\{u< \tau_n\}}dW_1(u) \right|^p \\
    & + \left|\int_0^s(K_1(t-u)-K_1(s-u))\sigma_1(\mu_1(u))\mathbbm{1}_{\{u< \tau_n\}}dW_1(u)\right|^p\\
    & + \left|\int_s^t  K_2(t-u)\sigma_2(\mu_2(u)) \mathbbm{1}_{\{u< \tau_n\}}dW_2(u)\right|^p \\
    & + \left|\int_0^s(K_2(t-u) -K_2(s-u))\sigma_2(\mu_2(u)) \mathbbm{1}_{\{u< \tau_n\}}dW_2(u)\right|^p
    \Bigg]\\
    & := \bold{1} + \bold{2} + \bold{3} + \bold{4} + \bold{5} + \bold{6} + \bold{7} + \bold{8}.
\end{align*}

By applying the H\"older inequality twice, we obtain
\begin{align*}
    \bold{1} & \leq (t-s)^{p/2}\left| \int_s^t K_1^2(t-u)(b_1 - \theta_1 \mu_1(u))^2\mathbbm{1}_{\{u< \tau_n\}}du\right|^{p/2} \\
    & \leq (t-s)^{p/2}\left(\int_s^t K_1^2(t-u)du\right)^{p/2-1}\left(\int_s^t K_1^2(t-u) \left|(b_1 - \theta_1 \mu_1(u))\mathbbm{1}_{\{u< \tau_n\}}\right|^p du  \right).
\end{align*}
 By taking expectation, we have
 \begin{align*}
     \mathbb{E}[\bold{1}] & \leq \sup_{0\leq u \leq T}\mathbb{E}\left[\left|(b_1 - \theta_1 \mu_1(u))\mathbbm{1}_{\{u< \tau_n\}}\right|^p  \right] (t-s)^{p/2}\left(\int_s^t K_1^2(t-u)du\right)^{p/2-1}\\
    & \leq \sup_{0\leq u \leq T}\mathbb{E}\left[\left|(b_1 - \theta_1 \mu_1(u))\mathbbm{1}_{\{u< \tau_n\}}\right|^p  \right] (t-s)^{(H_1 + 1/2)p}. 
 \end{align*}
Similarly, we have
\begin{align*}
\mathbb{E}[\bold{2}] \leq \sup_{0\leq u \leq T}\mathbb{E}\left[\left|(b_1 - \theta_1 \mu_1(u))\mathbbm{1}_{\{u< \tau_n\}}\right|^p  \right] (t-s)^{p/2}\left(\int_0^s(K_1(t-u)-K_1(s-u))^2 du\right)^{p/2}. 
\end{align*}
For $H_1 \in (0,  1/2)$, 
\begin{align*}
    & \left(\int_0^s(K_1(t-u)-K_1(s-u))^2 du\right)^{p/2} \leq(t-s)^{H_1 p} \int_0^{\infty}\left[(u+1)^{H_1 -1/2}-u^{H_1 -1/2}\right]^2 d u \\
& \leq(t-s)^{H_1 p}\left\{\int_0^1\left[(u+1)^{H_1 -1/2}-u^{H_1 - 1 / 2}\right]^2 d u+\int_1^{\infty} u^{2 H_1 - 3} d u\right\} \\
& \leq C(t-s)^{H_1 p} .
\end{align*}
For $H_1 \in [1/2, 1)$, it is clear that $\left(\int_0^s(K_1(t-u)-K_1(s-u))^2 du\right)^{p/2} \leq C(t-s)^{H_1 p}$. 
Applying the Burkholder-Davis-Gundy inequalities and using similar calculations, we have 
\begin{align*}
    \mathbb{E}[\bold{1} + \bold{2} + \bold{5} + \bold{6}] \leq C \sup_{0 \leq t \leq T} \mathbb{E}\left[\left| (b_1 - \theta_1 \mu_1(u))\mathbbm{1}_{\{u< \tau_n\}}\right|^p + |\sigma_1(\mu_1(u))^2\mathbbm{1}_{\{u\leq \tau_n\}}|^{p/2}\right](t-s)^{H_1 p},
\end{align*}
and 
\begin{align*}
    & \mathbb{E}[\bold{3} + \bold{4} + \bold{7} + \bold{8}] \\
    & \leq C \sup_{0 \leq t \leq T} \mathbb{E}\left[\left| (b_2 - \beta_2 \mu_1(u) - \theta_2 \mu_2(u))\mathbbm{1}_{\{u< \tau_n\}}\right|^p + |\sigma_2(\mu_2(u))^2\mathbbm{1}_{\{u\leq \tau_n\}}|^{p/2}\right](t-s)^{H_2 p}. 
\end{align*}
As $\sigma_1(\mu_1)^2$ and $\sigma_2(\mu_2)^2$ are linear in $\mu_1$ and $\mu_2$ respectively, there exist a constant $C$ such that
\begin{align*}
    \mathbb{E}\left[|\mu_2^n(t) - \mu_2^n(s)|^p\right] \leq C(t-s)^{\min\{H_1, H_2\}p}. 
\end{align*}
The H\"older continuity follows from the Kolmogorov continuity theorem (see Theorem I.2.1 in \cite{RY1999}). As $\mu_2(t) = \mu_2^n(t)$ almost surely on $\{t \leq \tau_n\}$ for each $t$, the result follows from letting $\tau_n \rightarrow \infty$. 
\end{proof}

\subsection{Proof of Theorem \ref{Thm:admissible}}
\label{Appendix:Thm1}
\begin{proof}
From \eqref{Xtilde}, we have
\begin{align*}
X(t) & = \int_0^t e^{-\int_{0}^{s}r(v)dv}\left(\nu(s)^\top u(s)-\pi(s) - \mathbb{E}[z]c_1  {\hat{\mu}_2(s)} \right)ds  + \int_0^t e^{-\int_{0}^{s}r(v)dv}u(s)^\top\sigma_S(s)^\top d\boldsymbol{W}(s)\\ & -\int_0^t\int_{\mathbb{R}_+}ze^{-\int_{0}^{t}r(s)ds}\widetilde{N}(ds, dz). 
\end{align*}
By H\"older's inequality,
\begin{align*}
& \mathbb{E}\left[\int_0^{T_0}e^{-\int_0^s 2 r(v) dv}|\nu(s)^\top u(s)|^2 ds \right]\leq \left(\mathbb{E}\left[\int_0^{T_0}e^{-\int_0^s 4r(v)dv } ds\right]\right)^{\frac{1}{2}} \left(\mathbb{E}\left[\int_0^{T_0}|\nu(s)^\top u(s)|^4 ds \right] \right)^{\frac{1}{2}}   \\
& \leq \left(T_0\sup_{0\leq s \leq T_0}\mathbb{E}\left[e^{-\int_0^s 4r(v)dv}\right]\right)^{\frac{1}{2}} \left(\mathbb{E}\left[\int_0^{T_0}|\nu(s)^\top u(s)|^4 ds \right] \right)^{\frac{1}{2}} < \infty. 
\end{align*}
  {As $m_1$ and $m_2$ are bounded deterministic functions, it follows Lemma  \ref{lemma:E_mu} that for any constant $p \geq 2$, there exists a  constant $C$ such that $\sup_{t\leq T}\mathbb{E}[|\hat{\mu}|^p] \leq C$. } By applying H\"older's inequality, it follows that  $\mathbb{E}\left[\int_0^{T_0}e^{-\int_0^s 2 r(v) dv}|\pi(s)|^2 ds \right] <\infty$ and $\mathbb{E}\left[\int_0^{T_0}e^{-\int_0^s 2 r(v) dv}|\mathbb{E}[z]c_1  {\hat{\mu}_2(s)}|^2 ds \right] <\infty$. 
As we have $\mathbb{E}\left[\int_0^{T_0}\int_\mathbb{R}z^2e^{-2\int_0^tr(s)ds}\delta(dz)dt\right]\leq c_1\mathbb{E}[z^2]\int_0^{T_0}\mathbb{E}\left[e^{-2\int_0^t r(s)ds}\right]\mathbb{E}[  {\hat{\mu}_2(t)} 
]dt < \infty$, 
by the Burkholder--Davis--Gundy (BDG) inequality, 
\begin{align*}
\mathbb{E}\left[\sup_{0\leq t \leq T_0}\left(\int_0^t\int_{\mathbb{R}_+}ze^{-\int_{0}^{t}r(s)ds}\widetilde{N}(ds, dz)\right)^2\right] \leq C_1 \int_0^{T_0} \mathbb{E}\left[e^{-\int_0^s2r(v)dv}\right]\mathbb{E}\left[|  {\hat{\mu}_2(s)} |\right]ds < \infty
\end{align*}
for a constant $C_1$.
Then, by the BDG and H\"older's inequalities, we have
\begin{align*}
& \mathbb{E}\left[\sup_{0\leq t \leq T_0}|X(t)|^2\right] < C_2 + C_2\mathbb{E}\left[\int_0^{T_0}e^{-\int_0^s 2 r(v) dv}(|\nu(s)^\top u(s)|^2+ |\pi(s)|^2 + |\mathbb{E}[z]c_1  {\hat{\mu}_2(s)} |^2)ds \right] \\
& + C_2 \mathbb{E}\left[\int_0^{T_0} e^{-\int_0^s2r(v)dv}|\sigma_S(s)u(s)|^2ds\right]   + C_2 \int_0^{T_0} \mathbb{E}\left[e^{-\int_0^s2r(v)dv}\right]\mathbb{E}\left[|  {\hat{\mu}_2(s)}|\right]ds < \infty
\end{align*}
for a positive constant $C_2$. 
\end{proof}
\subsection{Proof of Theorem \ref{Thm:spike}}
\label{Appendix:Thm2}
\begin{proof}
Let $X^\epsilon(\cdot)$ denote the state process corresponding to $u^{t,\epsilon, \eta}(\cdot)$. Using the standard perturbation approach, 
	\[X^\epsilon(s) = X^*(s) + Y^\epsilon(s) + Z^\epsilon(s), s\in[t, T_0], \]
	where $Y^\epsilon(s)$ and $Z^\epsilon(s)$ are the solutions to the following SDEs respectively, 
	\[dY^\epsilon(s) = e^{-\int_0^s r(v)dv}\eta^\top\sigma_S^\top\boldsymbol{1}_{[t, t+\epsilon]}(s) d\boldsymbol{W}(s), ~ Y^\epsilon(t) =0, ~ s\in [t, T_0],  \]
	\[dZ^\epsilon(s) = e^{-\int_0^s r(v)dv} \nu^\top\eta\boldsymbol{1}_{[t, t+\epsilon]}(s)ds, ~ Z^\epsilon(t) =0, ~ s\in [t, T_0]. \]
	By the BDG inequality, there exists a positive constant $C_1$ such that
	\begin{align*}
	&\mathbb{E}_t\left[\sup_{s \in [t, T_0]}(Y_s^\epsilon)^2\right] \leq \mathbb{E}_t\left[\sup_{s\in[t, T_0]}\left(\int_{t}^{s}e^{-\int_{0}^{v}r(u)du} \eta^\top\sigma_S^\top\boldsymbol{1}_{[t, t+\epsilon]} d\boldsymbol{W}(v)\right)^2\right]\notag\\
	&\leq C_1\mathbb{E}_t\left[\int_{t}^{T_0}e^{-\int_{0}^{v}r(u)du} \eta^\top\sigma_S^\top\sigma_S\eta\boldsymbol{1}_{[t, t+\epsilon]}dv\right]
	\end{align*}
	 Let $\sigma_S^\top\sigma_S = (\Sigma_S^{ij}), ~i, j = 1, 2$. By Lemma \ref{Lemma:interest} and H\"older's inequality, there exists a positive constant $C_2$ such that
	\begin{align*}
	  & \mathbb{E}_t\left[\int_{t}^{T_0}e^{-\int_{0}^{v}r(u)du} \eta^\top\sigma_S^\top\sigma_S\eta dv\right]\\
	  & \leq C_2 \left(\int_t^{T_0}  \mathbb{E}[e^{-\int_{0}^{v}2r(u)du}]dv\right)^{\frac{1}{2}}\left(\int_t^{T_0}\sum_{ij}\mathbb{E}\left[(\Sigma_S^{ij})^2\right]dv\right)^{\frac{1}{2}}\mathbb{E}\left[|\eta|^2\right] < \infty.
	\end{align*}
	Therefore, $\mathbb{E}_t\left[\sup_{s \in [t, T_0]}(Y_s^\epsilon)^2\right] = O(\epsilon)$.
	Moreover, 
	\begin{align}
	&\mathbb{E}_t\left[\sup_{s \in [t, T_0]}(Z_s^\epsilon)^2\right] \leq \mathbb{E}_t\left[\sup_{s\in[t, T_0]}\left(\int_{t}^{s}e^{-\int_{0}^{v}r(u)du}\nu^\top\eta\boldsymbol{1}_{[t, t+\epsilon]}dv\right)^2\right]\notag\\
	&= \mathbb{E}_t\left[\left(\int_{t}^{t+\epsilon}e^{-\int_{0}^{v}r(u)du}\nu^\top\eta dv\right)^2\right] = O(\epsilon^2). 
	\end{align}
	Applying the same method as in \cite{HJZ2012}, the result follows. 
\end{proof}

\subsection{Proof of Theorem \ref{Thm:unique}}
\label{Appendix:unique}
\begin{proof}
Suppose that there is another admissible equilibrium control $u$ that is different from $u^*$ in \eqref{u*} and has a corresponding wealth process $X$. The equation \eqref{BSDE:p} admits a unique solution, with $X^*$ replaced by $X$. Let $(p(s; t), k_1(s; t), k_2(s, z; t))$ denote the solution. By Proposition \ref{Pro:Lam}, 
we have
 \begin{equation}\label{p*}
  e^{-\int_{0}^{t}r(v)dv}\nu(t) p(t;t) + e^{-\int_{0}^{t}r(v)dv}\sigma_S(t)^\top k_1(t; t) = 0, ~ a.s.,~a.e.,~ t \in [0, T_0].
 \end{equation}
 We define that
 $$\bar{p}(s; t) = p(s; t) - \left(X(s) + \Gamma_s^{(1)} - \lambda -\mathbb{E}_t[X(s) + \Gamma_s^{(1)}]\right), $$
 $\bar{k}_1(s; t) = k_1(s; t) - e^{-\int_{0}^{s}r(v)dv}\sigma_S(s)u(s) - \gamma^{(1)}(s)$, and $ \bar{k}_2(s, z; t) =  k_2(s, z; t) +  ze^{-\int_{0}^{s}r(v)dv}$. Then, we have $\bar{p}(s; t)\in S_\mathcal{F}^2(t,T_0;\mathbb{R}, \mathbb{P})$,  $\bar{k}_1(s; t)\in H_\mathcal{F}^2(t,T_0;\mathbb{R}^3, \mathbb{P})$, and $\bar{k}_2(s; t)\in F^2(t,T_0;\mathbb{R})$.
 
 Substituting the above into \eqref{p*} yields
 \begin{equation*}
 	\left[ \bar{p}(t; t) - \lambda\right]\nu(t) + \sigma_S(t)^\top\left( \bar{k}_1(t; t) + e^{-\int_{0}^{t}r(v)dv}\sigma_S(t)u(t) + \gamma^{(1)}(t)\right)= 0.
 \end{equation*}
 Thus,
 $u = -e^{\int_{0}^{t}r(v)dv}\left(\sigma_S(t)^\top \sigma_S(t)\right)^{-1}\left( \left[ \bar{p}(t; t) - \lambda\right]\nu(t) + \sigma_S(t)^\top\bar{k}_1(t; t) + \sigma_S(t)^\top\gamma^{(1)}(t)  \right) \triangleq  e^{\int_{0}^{t}r(v)dv}\left(\sigma_S(t)^\top \sigma_S(t)\right)^{-1}\left( \lambda \nu(t) - \sigma_S(t)^\top\gamma^{(1)}(t)  + D(t)\right)$ with $D(t) = -\bar{p}(t; t)\nu(t) - \sigma_S(t)^\top\bar{k}_1(t; t)$. We aim to prove the uniqueness by showing that $D \equiv 0$. We have
 \begin{align*}
 	d\bar{p}(s; t) & = dp(s; t) - d\left(X(s) + \Gamma_s^{(1)} - \lambda -\mathbb{E}_t[X(s) + \Gamma_s^{(1)}]\right)\\
 	&= -\left\{\nu(s)^\top\left(\sigma_S(t)^\top \sigma_S(t)\right)^{-1} D_s - \mathbb{E}_t\left[\nu(s)^\top\left(\sigma_S(t)^\top \sigma_S(t)\right)^{-1}D_s \right] \right\}ds \\
 	&+ \bar{k}_1(s; t)^\top d\boldsymbol{W}(s) +  \int_{\mathbb{R}_+} \bar{k}_2(s,z;t)\widetilde{N}(ds,dz), 
 \end{align*}
 and $\bar{p}(T_0; t) = 0$. By taking the conditional expectation on both sides of the above equation, it follows that $\mathbb{E}_t[\bar{p}(s; t)] = 0$ for $s > t$. Specifically, we have $\bar{p}(t; t) = 0$. Thus, $D_t = -\sigma_S(t)^\top \bar{k}_1(t; t).$ In addition, we have $k_1(s; t_1) = k_1(s; t_2)$ for a.e. $s \geq \max(t_1, t_2)$ similar to that in Proposition \ref{Pro:Lam}. We define that $\hat{p}(s; t) = \bar{p}(s; t) + \int_{s}^{T_0} \mathbb{E}_t\left[\nu(v)^\top \left(\sigma_S(t)^\top \sigma_S(t)\right)^{-1} D_v\right]dv$. Then,
 \begin{align*}
 	d\hat{p}(s; t) = \nu(s)^\top\left(\sigma_S(t)^\top \sigma_S(t)\right)^{-1}\sigma_S(s)^\top\bar{k}_1(s; s)ds + \bar{k}_1(s; s)^\top d\boldsymbol{W}(s) +  \int_{\mathbb{R}_+} \bar{k}_2(s,z;t)\widetilde{N}(ds,dz). 
 \end{align*}
It is easy to check $\mathbb{E}\left[\sup_{t\leq s\leq T_0}\left|\int_{s}^{T_0}\mathbb{E}_t\left[\nu(v)^\top\left(\sigma_S(t)^\top \sigma_S(t)\right)^{-1}\sigma_S(v)^\top\bar{k}_1(v; v)\right]dv\right|^{q}\right] < \infty$ for any $q\in (1, 2)$. As $\bar{p}(s; t)\in S_\mathcal{F}^2(t,T_0;\mathbb{R}, \mathbb{P})$, we have $\mathbb{E}\left[\sup_{t\leq s\leq T_0}\left|\hat{p}(s; t)\right|^{q}\right] < \infty$ for any $q \in (1, 2)$. Let 
  \begin{align*}
    & \mathcal{E}_t\left(\sigma_S\left(\sigma_S^\top \sigma_S\right)^{-1}\nu\right) \\
    & = \exp\left(\int_{0}^{t}-\nu(s)^\top\left(\sigma_S(t)^\top \sigma_S(t)\right)^{-1}\sigma_S(s)^\top d{\boldsymbol{W}}(s)-\frac{1}{2}|\sigma_S(s)\left(\sigma_S(t)^\top \sigma_S(t)\right)^{-1}\nu(s)|^2ds \right).
   \end{align*}
  We introduce the measure $\hat{\mathbb{P}}$ by $\frac{d\hat{\mathbb{P}}}{d\mathbb{P}}|_{\mathcal{F}_t} =  \mathcal{E}_t\left(\sigma_S\left(\sigma_S^\top \sigma_S\right)^{-1}\nu\right)$. Under measure $\hat{\mathbb{P}}$,  $\hat{\boldsymbol{W}}(t) \triangleq \boldsymbol{W}(t) + \int_0^t \sigma_S(s)\left(\sigma_S(s)^\top \sigma_S(s)\right)^{-1}\nu(s)ds$ is a standard Brownian motion. 
  Under $\hat{\mathbb{P}}$, we have
 \[d\hat{p}(s; t) =  \bar{k}_1(s; s)^\top d\hat{\boldsymbol{W}}(s) + \int_{\mathbb{R}_+} \bar{k}_2(s,z;t)\widetilde{N}(ds,dz).\]
 For any $m >1$, 
 \begin{align*}
  & \mathbb{E}\left[\mathcal{E}\left(\sigma_S\left(\sigma_S^\top \sigma_S\right)^{-1}\nu\right)^m\right] \leq \\
  & \left(\mathbb{E}\left[\exp\left((2m^2 - m)\int_{0}^{T_0}|\sigma_S(s)\left(\sigma_S(t)^\top \sigma_S(t)\right)^{-1}\nu(s)|^2ds \right) \right] \mathbb{E}\left[\mathcal{E}_{T_0}\left(2m\sigma_S\left(\sigma_S^\top \sigma_S\right)^{-1}\nu\right)\right]\right)^{\frac{1}{2}}\\
  & < \infty. 
 \end{align*}
 Therefore, if $m > \frac{q(q +1)}{(q-1)^2}$ in the above equation, there exists a value of $\bar{q}_0\in (1, q)$ for any $q \in (1, 2)$ such that 
 \[\hat{\mathbb{E}}\left[\sup_{t\leq s \leq T_0}|\hat{p}(s; t)|\right] \leq \left(\mathbb{E}\left[\sup_{t\leq s \leq T_0}|\hat{p}(s; t)|^{\bar{q}_0} \right] \right)^{\frac{1}{q_0}}\left(\mathbb{E}\left[\mathcal{E}\left(\sigma_S\left(\sigma_S^\top \sigma_S\right)^{-1}\nu\right)^{\frac{\bar{q}_0}{\bar{q}_0-1}}\right] \right)^{\frac{\bar{q}_0}{\bar{q}_0-1}} < \infty,\]
 so that $\hat{p}(s; t)$ is a $\hat{\mathbb{P}}$-martingale. Hence, $\hat{p} \equiv 0$ and $D \equiv 0$.
\end{proof}

\subsection{Proof of Proposition \ref{pro:Vas}}
\begin{proof}
Under the pricing measure $\mathbb{Q}$, we have
\begin{align*}
\mu_1(t) & = \mu_1(0) + \int_{0}^{t}K_1(t-s)(b_1 + \varphi_1(s)\sigma_1 - \theta_1 \mu_1(s))ds + \int_{0}^{t}K_1(t-s)\sigma_1 dW_1^\mathbb{Q}(s),\\
\mu_2(t) & = \mu_2(0) + \int_{0}^{t}\beta_1 K_1(t-s)(b_1 + \varphi_1(s)\sigma_1  - \theta_1 \mu_1(s))ds  + \int_0^tK_2(t-s)(b_2 - \beta_2 \mu_1(s)  \nonumber \\ 
& - \theta_2 \mu_2(s))ds + \int_{0}^{t}\beta_1 K_1(t-s)\sigma_1dW_1^\mathbb{Q}(s) + \int_0^t K_2(t-s)\sigma_2 dW_2^\mathbb{Q}(s).
\end{align*}
Hence, the {force of mortality} maintains the affine structure and satisfies \eqref{mu0}, with $b(s) = (b_1 + \varphi_1(s)\sigma_1, b_2)^\top$ and $W(s)$ replaced by $W^\mathbb{Q}(s)$. 
The interest rate follows the SDE:
\[ dr(t)=(b_r +\vartheta(s)\sigma_r -  \theta_r r(t))dt +\sigma_rdW^\mathbb{Q}_r(t). \]
By Appendix \ref{appendix:affine}, we have 
\[\mathcal{B}(t,T) = \mathbb{E}^{\mathbb{Q}}\left[\left.e^{-\int_{t}^{T}r(s)ds}\right|\mathcal{F}_t\right] = e^{d_0(t) + d_1(t)r(t)}, \]
where $d_1(s) = \frac{1}{\theta_r}\left(e^{-\theta_r(T-s)} -1\right)$. Hence, $\sigma_b(s) = d_1(s)\sigma_r$.
By Lemma \ref{lemma:expmu}, the price of the zero-coupon longevity bond is given by
\[ \mathcal{B}_L(t,T) =\mathbb{E}^{\mathbb{Q}}\left[\left.e^{-\int_{t}^{T}r(s)+ \mu_1(s)ds}\right|\mathcal{F}_t\right] =\mathcal{B}(t,T)e^{\int_0^t\mu_1(s)ds}\exp(Y_t(T)),\]
where $Y_t(T)$ is as defined in \eqref{Y} in Lemma \ref{lemma:expmu} with $f = (-1, 0)$, $b(s) = (b_1 + \varphi_1(s)\sigma_1, b_2)^\top$, $\sigma(\mu(s)) \equiv \sigma_\mu$, and $W$ replaced by $W^\mathbb{Q}$. Therefore, by It\^o's formula, $\sigma_l = \psi(T-s)\sigma_\mu$, where $\psi = (\psi_1, \psi_2)$ solves the Volterra--Riccati equation $\psi = (f- \psi\Theta )*K$ with $f =(-1, 0)$. Thus, the result for $\sigma_l$ follows.
By Lemma \ref{lemma:expmu}, for any constant $q > 1$, we have $$\mathbb{E}\left[e^{-q\int_0^{t}\mu_2(s)ds}\right] = \exp\left\{{\int_0^t\left(-q\mu_2(0)+\bar{\psi}(s)(b - \Theta \mu(0)) +\frac{1}{2}\bar{\psi}(s)\sigma_\mu\sigma_\mu^\top\bar{\psi}(s)^\top \right)ds}\right\}$$
where $\bar{\psi}$ solves the Volterra--Riccati equation $\bar{\psi} = (f -\bar{\psi}\Theta)*K$ with $f = (0, -q)$. By Lemma 4.4 in \cite{AJ}, $\bar{\psi}$ admits the closed-form solution $\bar{\psi} = f*E_\Theta $, where $E_\Theta$ is as defined in Lemma \ref{lemma:expmu}. Therefore, $\pi(t) \in L_\mathcal{F}^q(0, T_0; \mathbb{R}, \mathbb{P})$.
\end{proof}

\subsection{Proof of Proposition \ref{pro:RB}}
\begin{proof}
By calculation, we have
$K\Theta = \left(\begin{array}{cc}
    \theta_1K_1 &  0 \\
    \beta_1\theta_1K_1 + \beta_2K_2 & \theta_2K_2
\end{array} \right)$. As $K_1, K_2 \in L_{loc}^2(\mathbb{R}_+, \mathbb{R})$, we have $R_{1\theta_1}, R_{2\theta_2} \in L_{loc}^2(\mathbb{R}_+, \mathbb{R})$ and thus $(R_{2\theta_2}*R_{1\theta_1})*\theta_1K_1 = R_{2\theta_2}*(R_{1\theta_1}*\theta_1K_1)$ \citep{GLS1990}. Then,
\begin{align*}
   & \left[\beta_1R_{1\theta_1} + \frac{\beta_2}{\theta_2} R_{2\theta_2}-(\beta_1 + \frac{\beta_2}{\theta_2})R_{2\theta_2}*R_{1\theta_1}\right]* \theta_1K_1 + R_{2\theta_2}*(\beta_1\theta_1K_1 + \beta_2K_2)\\
   & = \beta_1(\theta_1K_1 - R_{1\theta_1}) + \frac{\beta_2}{\theta_2} R_{2\theta_2}*\theta_1K_1 -(\beta_1 + \frac{\beta_2}{\theta_2})R_{2\theta_2}*(\theta_1K_1 - R_{1\theta_1}) + R_{2\theta_2}*\beta_1\theta_1K_1\\
   & + \frac{\beta_2}{\theta_2}(\theta_2K_2 - R_{2\theta_2})\\
   & = (\beta_1 + \frac{\beta_2}{\theta_2})(R_{2\theta_2}*R_{1\theta_1}) = \beta_1\theta_1K_1 + \beta_2K_2 - \left[\beta_1R_{1\theta_1} + \frac{\beta_2}{\theta_2} R_{2\theta_2}-(\beta_1 + \frac{\beta_2}{\theta_2})R_{2\theta_2}*R_{1\theta_1}\right].
\end{align*}
Thus, we have 
\begin{align}\label{RB}
R_\Theta = \left(\begin{array}{cc} 
    R_{1\theta_1} & 0 \\
    \beta_1R_{1\theta_1} + \frac{\beta_2}{\theta_2} R_{2\theta_2}-(\beta_1 + \frac{\beta_2}{\theta_2})R_{2\theta_2}*R_{1\theta_1} & R_{2\theta_2}
\end{array} \right),
\end{align}
as it satisfies $R_\Theta*K\Theta = K\Theta - R_\Theta$ and  $K\Theta*R_\Theta = K\Theta - R_\Theta$.  

 As shown in Lemma \ref{lemma:expmu}, $E_\Theta$ is defined as $E_\Theta = K- R_\Theta*K$, where $R_\Theta$ is the resolvent of $K\Theta$. Then,
$$E_\Theta \Theta = K \Theta  - R_\Theta*K \Theta  = K \Theta  - (K \Theta - R_\Theta) = R_\Theta.$$
Therefore, $E_\Theta = R_\Theta  \Theta ^{-1}$. The result in Proposition \ref{pro:RB} follows.
\end{proof}

\subsection{Proof of Proposition \ref{Pro:CIR}}
\begin{proof}
Similar to the proof of Proposition \ref{pro:Vas},  $\mu$ maintains the affine structure under the pricing measure. Specifically, under $\mathbb{Q}$, $\mu$ satisfies \eqref{mu0} with $b(s) = (b_1, b_2)^\top$, where $\Theta$ is replaced by
\[\left(\begin{array}{cc}
    \theta_1- \widetilde{\varphi}_1\widetilde{\sigma}_1^2 & 0 \\
    \beta_2 &  \theta_2
\end{array}\right) \triangleq \Theta_q,\]
and $W$ is replaced by $W^\mathbb{Q}$. By Lemma \ref{expmu} and It\^o's formula, we have $\sigma_l(s) = \psi'(T-s)\sigma(\mu(s))$, where $\psi' = (\psi'_1, \psi'_2)$ solves the Volterra--Riccati equation $\psi' = (f - \psi'\Theta_q+ \frac{1}{2} (\widetilde{\sigma}_1^2 (\psi'_1)^2, \widetilde{\sigma}_2^2 (\psi'_2)^2))*K$ with $f = (-1, 0)$. Thus, $\psi'_2 = 0$ and $\psi'_1$ solves $\psi'_1 = (-1 - (\theta_1 - \widetilde{\varphi}_1 \widetilde{\sigma}_1^2)\psi'_1+ \frac{1}{2}\widetilde{\sigma}_1^2 (\psi'_1)^2)*K_1$.
Under the pricing measure $\mathbb{Q}$, the interest rate follows
\[	dr(t)=(b_r -  (\theta_r - \vartheta\sigma_r^2 )r(t))dt +\sigma_r \sqrt{r(s)} dW^\mathbb{Q}_r(t).\]
The result follows from Appendix \ref{appendix:affine} and It\^o's formula.
\end{proof}

\end{appendices}
	
\end{document}